\documentclass [11pt]{article}
\usepackage{amsmath,amsthm,amsfonts,amscd,eucal,latexsym,amssymb,mathrsfs}
\usepackage{epsfig}  
\usepackage{simplewick}
\usepackage{tikz-cd}
\usepackage{hyperref}
\input xy
\xyoption{all}
 
\def\WF{{\text{WF}}}

\def\cI{{\mathcal I}}

\def\cO{{\mathcal O}}

\def\cS{{\mathcal S}}

\def\sM{{\mathsf M}}
\def\sD{{\mathsf D}}

\def\cc{\text{c.c.}}

\def\bC{{\mathbb C}}           

\def\bN{{\mathbb N}} 
\def\bM{{\mathbb M}} 
 
\def\bR{{\mathbb R}}

\def\mA{\mathscr{A}} 
 
\def\mD{\mathscr{D}} 
\def\mE{\mathscr{E}} 
\def\mO{\mathscr{O}} 
\def\mS{\mathscr{S}}  
\def\mF{\mathscr{F}}  
\def\mI{\mathscr{I}}  
\def\mM{\mathscr{M}} 

\def\mV{\mathscr{V}}

\def\mN{\mathscr{N}} 
\def\mR{\mathscr{R}} 
\def\mP{\mathscr{P}}

\def\beq{\begin{eqnarray}}
\def\eeq{\end{eqnarray}}

\def\ep{\varepsilon}

\def\la{\lambda}

\def\ph{\varphi}

\def\De{\Delta}

\def\Si{\Sigma}

\def\Id{\mathbb{I}}

\def\supp{\mathrm{supp}}

\def\Floc{{\mathscr{F}_{\mathrm{loc}}}}
\def\Freg{\mathscr{F}_{\mathrm{reg}}}
\newcommand{\Ftloc}[1][]{\mathscr{F}_{T_{#1}\mathrm{loc}}}
\def\Fmloc{{\mathscr{F}_{\mathrm{mloc}}}}
\newcommand{\Fmuc}[1][]{\mathscr{F}_{\mu\textrm{c}_{#1}}}

\def\Areg{\mathscr{A}^{\mathrm{reg}}}
\def\Aon{\mathscr{A}^{\mathrm{on}}}
\def\A0{\mathscr{A}^{0}}

\def\tAreg{{\widetilde{\mathscr{A}}^{\mathrm{reg}}}}

\newcommand{\T}[1][]{\cdot_{T_{#1}}}

\newcommand{\molh}[2]{\mR^\hbar_{#1,#2}}
\newcommand{\molhh}[1]{\mR^\hbar_{#1}}

\def\kms{\omega^\beta}

\def\mink{\bM}
\def\vhlim{\mathrm{vH-}\lim}

\def\qmoller{\mR_{1,Q}}
\def\qcmoller{R_{1,Q}}
\def\bqmoller{\molh{1}{Q}}

\def\chimoller{\mR^\hbar_{V(\chi)}}
\def\chikmoller{\mR^\hbar_{V(h_k\chi)}}

\def\bq{\beta_{1,Q}}
\def\wbq{\widetilde{\beta}_{1,Q}}

\def\vstar{\star_{1,V}}
\def\qstar{\star_{1,Q}}
\def\vast{\ast_{1,V}}
\def\qast{\ast_{1,Q}}

\def\q{Q^{(1)}}

\def\cc{c_{2,1}}
\def\cinv{\frac{1}{\cc}}

\def\bx{\boldsymbol{x}}
\def\by{\boldsymbol{y}}
\def\bz{\boldsymbol{z}}
\def\bp{\boldsymbol{p}}
\def\bk{\boldsymbol{k}}
\def\bm{\boldsymbol{m}}

\newtheorem{theorem}{Theorem}[section]

\newtheorem{proposition}{Proposition}[section]
\newtheorem{lemma}{Lemma}[section]
\theoremstyle{definition}
\newtheorem{definition}{Definition}[section]
\newtheorem{remark}{Remark}[section]
 
\begin{document} 
%
%
 
\par 
\bigskip 
\LARGE 
\noindent 
{\bf The generalised principle of perturbative agreement and the thermal mass} 
\bigskip 
\par 
\rm 
\normalsize 
 

\large
\noindent 
{\bf Nicol\`o Drago$^{1,a}$}, {\bf Thomas-Paul Hack$^{1,b}$}, {\bf Nicola Pinamonti$^{1,2,c}$} \\
\par
\small
\noindent$^1$ Dipartimento di Matematica, Universit\`a di Genova - Via Dodecaneso 35, I-16146 Genova, Italy. \smallskip

\noindent$^2$ Istituto Nazionale di Fisica Nucleare - Sezione di Genova, Via Dodecaneso, 33 I-16146 Genova, Italy. \smallskip
\smallskip

\noindent E-mail: 
$^a$drago@dima.unige.it, 
$^b$hack@dima.unige.it,
$^c$pinamont@dima.unige.it\\ 

\normalsize

\par 
 
\rm\normalsize 

\rm\normalsize 
 
 
\par 
\bigskip 
 
\rm\normalsize 
\noindent {\small Version of \today}

\par 
\bigskip

\rm\normalsize

\small 

\noindent {\bf Abstract}. The Principle of Perturbative Agreement, as introduced by Hollands \& Wald, is a renormalisation condition in quantum field theory on curved spacetimes. This principle states that the perturbative and exact constructions of a field theoretic model given by the sum of a free and an exactly tractable interaction Lagrangean should agree. We develop a proof of the validity of this principle in the case of scalar fields and quadratic interactions without derivatives which differs in strategy from the one given by Hollands \& Wald for the case of quadratic interactions encoding a change of metric. Thereby we profit from the observation that, in the case of quadratic interactions, the composition of the inverse classical M\o ller map and the quantum M\o ller map is a contraction exponential of a particular type. Afterwards, we prove a generalisation of the Principle of Perturbative Agreement and show that considering an arbitrary quadratic contribution of a general interaction either as part of the free theory or as part of the perturbation gives equivalent results. Motivated by the thermal mass idea, we use our findings in order to extend the construction of massive interacting thermal equilibrium states in Minkowski spacetime developed by Fredenhagen \& Lindner to the massless case. In passing, we also prove a property of the construction of Fredenhagen \& Lindner which was conjectured by these authors.

\normalsize

\vskip .3cm

\section{Introduction}
In the last twenty years a solid conceptual framework of perturbative algebraic quantum field theory on curved spacetimes has been established \cite{BFK,BF00,HW01, HW02, BFV, HW05,Hollands:2007zg,  BDF,FredenhagenRejzner3}. This was possible thanks to the seminal work \cite{Radzikowski}, which introduced the powerful tools of microlocal analysis to algebraic quantum field theory. In particular, a number of axioms for the time--ordered products of Wick--polynomials on curved spacetime have been proposed in \cite{HW01, HW05}, based on earlier work in \cite{BFK,BF00}. In the former works it has been argued that these axioms are the renormalisation conditions which should be satisfied by any regularisation scheme on curved spacetimes compatible with general covariance and the principles of quantum field theory. In \cite{HW01,HW02,HW05} it has been proven that renormalisation schemes satisfying these conditions actually exist and the renormalisation freedom compatible with these renormalisation conditions has been classified.

One of these axioms, introduced in \cite{HW05}, is the Principle of Perturbative Agreement (PPA). This principle states that the perturbative and exact constructions of a field theoretic model given by the sum of a free and an exactly tractable interaction Lagrangean should agree. To explain this in more detail for  the example of a quadratic interaction, let us consider a quadratic action $\cS_1$ and a quadratic interaction potential $Q$ with at most two derivatives, which, in the case of a real scalar field, can quantify a change of the mass, a change of the coupling to the scalar curvature, a change of the metric and a change of the external current. We may construct the free exact algebras of observables $\mA_1$ and $\mA_2$ corresponding to the actions $\cS_1$ and $\cS_1 + Q$, as well as the perturbative algebra of interacting observables $\mA_{1,Q}$ corresponding the free action $\cS_1$ and the perturbation $Q$. The classical M\o ller map $\qmoller$ intertwines the classical dynamics of $\cS_1+Q$ and $\cS_1$ and is an isomorphism between $\mA_2$ and $\mA_1$, whereas the quantum M\o ller map $\bqmoller$, which enters the construction of the perturbative algebra $\mA_{1,Q}$, may be thought of as an isomorphism between $\mA_{1,Q}$ and a (subalgebra) of $\mA_{1}$. The PPA as stated and proved in \cite{HW05}, see also \cite{Zahn} for the case of higher spin fields and variations of a background gauge field, requires that $\bq\doteq \qmoller^{-1}\circ \bqmoller$ intertwines the time--ordered products corresponding to $\cS_1$ and $\cS_1 + Q$ and thus imposes renormalisation conditions on the time--ordered products of both theories. In fact, the validity of the PPA implies that $\qmoller$ is an isomorphism between $\mA_{1,Q}$ and a (subalgebra of) $\mA_2$. In \cite{HW05} several important physical implications of the validity of the PPA are discussed. In particular it is proven that a number of identities valid in a classical field theory hold in the corresponding quantum theory as well. 

In this work, we generalise the proof of the Principle of Perturbative Agreement for arbitrary quadratic $Q$ to the setting where an additional, not necessarily quadratic, interaction potential $V$ is present. We prove that, provided the time--ordered product satisfies the PPA, the classical M\o ller map extends to an isomorphism between the algebras $\mA_{2,V}$ and $\mA_{1,Q+V}$, i.e. the perturbative algebra corresponding to the free action $\cS_1 + Q$ and perturbation $V$, and the perturbative algebra corresponding to the free action $\cS_1$ and the perturbation $Q+V$, respectively. To this avail, we develop a proof of the PPA for quadratic $Q$ without derivatives which differs in strategy from the proof of \cite{HW05} for $Q$ encoding a change of metric and is based on the observation that $\bq\doteq \qmoller^{-1}\circ \bqmoller$ is, in a certain sense, a contraction exponential w.r.t. the difference of the (chosen) Feynman propagators $\De^F_{1+Q}-\De^F_1$ of the two quadratic models. In particular, we show that $\bq$ is a deformation in the sense of deformation quantization when applied to sufficiently regular functionals. We note that the proof of the PPA for quadratic $Q$ without derivatives has not been explicitly spelt out in \cite{HW05}, although the proof of the PPA for $Q$ encoding changes of the metric given in \cite{HW05} can be extended to general quadratic $Q$ without much effort. 

A further independent proof of the PPA for scalar fields and quadratic interactions without derivatives has been given in \cite{GHP}, where an explicit analytic regularisation scheme on curved spacetime is developed and shown to satisfy the PPA for this class of quadratic perturbations.

After proving the generalised PPA, we construct as an application interacting thermal equilibrium states for massless Klein--Gordon fields on Minkowski spacetime with an arbitrary interaction $V$. We accomplish this by combining the construction of such states in the massive case developed in \cite{FredenhagenLindner,Lindner:2013ila} with the idea of the thermal mass, which is an effective mass term appearing in $\phi^4$--theory upon changing the normal--ordering prescription from normal--ordering with respect to the free vacuum to normal--ordering with respect to the free thermal state. However, our construction also applies to the case of zero temperature, because we introduce a temperature--independent positive virtual mass whose magnitude turns out to be inessential for the construction. In passing, we also prove that the construction of \cite{FredenhagenLindner,Lindner:2013ila} is, in the adiabatic limit, independent of the temporal cut--off which enters this construction. This property was conjectured in \cite{FredenhagenLindner} and proved up to the strong clustering property which we discuss in Appendix \ref{sec:clustering}.

Our paper is organised as follows. We begin by reviewing the functional approach to perturbative algebraic quantum field theory on curved spacetimes in Section \ref{sec:pAQFT}, as this is the framework in which we shall work throughout. In Section \ref{sec:PPA}, we first review the precise formulation of the Principle of Perturbative Agreement introduced in \cite{HW05} for the case of at most quadratic interactions and sketch the strategy of the alternative proof of its validity given in the present work. After elaborating this proof of the PPA for quadratic interactions in several steps, where the final steps are only given for the case of quadratic interactions without derivatives, we prove the generalisation to the case where an additional general interaction is present in Section \ref{sec:gPPA}. In Section \ref{sec:KMS}, we review the construction of massive interacting KMS states on Minkowski spacetime developed in \cite{FredenhagenLindner, Lindner:2013ila} and use our results to generalise this construction to the massless case. The appendix contains the proofs of several subsidiary results.

\section{Functional Approach to Quantum Field Theory}
\label{sec:pAQFT}
In this section we briefly recall the functional approach to perturbative algebraic quantum field theories on curved spacetimes \cite{BDF,FredenhagenRejzner,FredenhagenRejzner2}, which is the setting of our work. For simplicity we will consider only the case of the real Klein--Gordon field. However, this framework can be applied to more general field theories including theories with local gauge symmetries, see e.g. \cite{FredenhagenRejzner3}. We shall consider only spacetimes $(\mM,g)$ which are globally hyperbolic throughout, see e.g. \cite{BGP} for a definition and properties. Moreover, we set $\mE(\mM^n)\doteq C^\infty(\mM^n,\bR)$, $\mE_\bC(\mM^n)\doteq C^\infty(\mM^n,\bC)$, $\mD(\mM^n)\doteq C_0^\infty(\mM^n,\bR)$, $\mD_\bC(\mM^n)\doteq C_0^\infty(\mM^n,\bC)$ and denote by $J^\pm(\mO)$ the causal future/past of $\mO\subset \mM$.

The functional approach can be thought of as to provide a concrete realisation of the abstract Borchers--Uhlmann algebra of quantum fields and its extension to include Wick polynomials and time--ordered products thereof, see \cite{HW01, HW02}. In this approach, the off--shell observables are described by complex--valued {\bf functionals} $F\in\mF$ over real--valued {\bf off-shell field configurations} $\phi\in\mE(\mM)$. Physically interesting observables can detect only local perturbations of the field configuration, hence physically relevant functionals are required to be supported in {\bf compact regions} of spacetime only. Here the support of a functional $F$ is defined as
\begin{align}
\supp\, F\doteq\{ & x\in \mM\,|\,\forall \text{ neighbourhoods }U\text{ of }x\ \exists \phi,\psi\in\mE(\mM),\,\supp\,\psi\subset U,
\\ & \text{ such that }F(\phi+\psi)\not= F(\phi)\}\,.\nonumber
\end{align}

Moreover, in order to be able to compute products among observables, certain regularity conditions are necessary. 
In particular, it is assumed that functionals representing observables are {\bf smooth}, namely  that, for each $\ph,\psi\in\mE(\mM)$, the function $\la\to F(\ph+\la\psi)$ has to be infinitely often differentiable with the $n$--th functional derivative at $\la=0$ represented by a compactly supported symmetric distribution $F^{(n)}(\ph)\in\mE_\bC^\prime(\mM^n)$ in the sense that
\begin{gather}\label{eq: smooth functional}
\left.\frac{d^n}{d\la^n}F(\ph+\la\psi)\right|_{\la=0}=
\left\langle F^{(n)}(\ph),\psi^{\otimes n}\right\rangle.
\end{gather}
We further restrict the possible distributions appearing as functional derivatives by demanding that their wave front set $\WF\left(F^{(n)}\right)$ has a particular form. This defines the set of {\bf microcausal functionals}, which is the maximal set of functionals we shall need.
\begin{gather}\label{def:microcausal functionals}
\Fmuc\doteq
\left\{
F\in\mF\left|\ F\textrm{ smooth, compactly supported, } \WF\left(F^{(n)}\right)\cap\left(\overline{V}^n_+\cup\overline{V}^n_-\right)=\emptyset
\right.\right\},
\end{gather}
here $V_{+/-}$ is a subset of the cotangent space formed by the elements whose covectors are contained in the future/past light cones and $\overline{V}_{+/-}$ denotes is closure.
This set contains the set $\Floc$ of {\bf local functionals} formed by the microcausal functionals whose $n$-th functional derivatives $F^{(n)}$ are supported only on the diagonal $\sD_n\subset \mM^n$ and satisfy $\WF(F^{(n)})\perp T\sD_n$, as well as the {\bf regular functionals} $\Freg$, which are the subset of $\Fmuc$ consisting of functionals with $\WF\left(F^{(n)}\right)=\emptyset$ for all $n$. A typical element of $\Floc$ is a smeared field polynomial
\begin{gather}\label{eq:smearedpolynomial}
F_{m,f}(\phi)\doteq
\int_\mM f\phi^m\;d\mu_g\qquad f\in \mD_\bC(\mM)\,, \qquad \phi\in\mE(\mM)\,,
\end{gather}
where $d\mu_g$ is the canonical volume form induced by the metric $g$. The case $m=1$, i.e. a linear functional 
\begin{gather}\label{def:linear functionals}
F_f(\phi)\doteq
\int_\mM f\phi\;d\mu_g\qquad f\in \mD_\bC(\mM)\,, \qquad \phi\in\mE(\mM)\,,
\end{gather}
is a special case of an element of $\Freg$.

\subsection{Algebras of observables in linear and affine theories}

All these sets of functionals are linear spaces, which we can endow with an involution $*$, defined as the complex conjugation $F^*(\phi)\doteq\overline{F(\phi)}$. Functionals encoding physical observables have to satisfy $F^*=F$. We may further give the linear involutive space $\Fmuc$ and its linear subspaces the structure of an algebra by equipping them with a product which encodes the quantum commutation relations. 
Whenever the equation of motion of the model we are going to quantize is an affine Klein--Gordon--type equation, i.e.
\beq\label{def:KGeqn}
P \phi = j\qquad P\doteq-\Box_g+M\,,
\eeq
where $\Box_g$ is the d'Alembert operator associated to $g$ and $M, j\in\mE(\mM)$, the construction of the above--mentioned product can be made explicit by formally deforming the pointwise product 
\[
(F\cdot G) (\phi) \doteq \sM(F\otimes G)(\phi) \doteq F(\phi)G(\phi)
\]
to a non--commutative product which we shall indicate by $\star$.

To this end, we consider a distribution $\De^+\in\mD_\bC^\prime(\mM^2)$ which is of the form
\[
\De^+ = \De^S + \frac{i}{2}\De
\]
where $\De^S$ is real and symmetric and $\De$ is the real and antisymmetric {\bf causal propagator} corresponding to the normally hyperbolic Klein--Gordon operator $P$ \eqref{def:KGeqn}. We recall that the causal propagator is uniquely defined as 
$\De = \De^R-\De^A$,
the retarded--minus--advanced fundamental solution of the linear part of \eqref{def:KGeqn}, (see e.g. \cite{BGP} for their rigorous unique construction in every globally hyperbolic spacetime), whereas $\De^S$ is non--unique. Moreover, we demand that $\De^+$ satisfies the {\bf Hadamard condition} in its microlocal form \cite{Radzikowski}, i.e. that the wave front sent $\WF(\De^+)$ reads
\beq\label{def:muc}
\WF(\De^+)=\mathscr{V}^+\doteq \{
(x,x^\prime,\xi,-\xi^\prime)\in T^*M^2\backslash\{0\}\,|\,(x,\xi)\sim(x^\prime,\xi^\prime)\textrm{ and }\xi\rhd 0
\},
\eeq
where $(x,\xi)\sim(x^\prime,\xi^\prime)$ means that there exists a null geodesic $\gamma$ connecting $x$ to $x^\prime$ such that $\xi$ is coparallel and cotangent to $\gamma$ at $x$ and $\xi^\prime$ is the parallel transport of $\xi$ from $x$ to $x^\prime$ along $\gamma$, whereas $\xi\rhd 0$ indicates that $\xi$ is future directed.

Given a distribution $\De^+$ with these properties, we define the $\star$--product on $\Fmuc$ as
\beq\label{eq:exp-product}
F \star G = \sM \circ \exp \hbar\,\Gamma_{\De^+} (F \otimes G)  , \qquad \Gamma_{\De^+} = \int_{\mM^2}\De^+(x,y) \frac{\delta}{\delta \phi(x)} \otimes \frac{\delta}{\delta \phi(y)} \;d\mu_g(x)\,d\mu_g(y)\,,
\eeq
explicitly,
\begin{equation}\label{eq:exp-product explicit}
F\star G\doteq F\cdot G+
\sum_{n\geq 1}\frac{\hbar^n}{n!}\left\langle \left(\De^+\right)^{\otimes n},F^{(n)}\otimes G^{(n)}\right\rangle.
\end{equation}
This product is well--defined on account of the regularity properties of microcausal functionals as well as the Hadamard property of $\De^+$ and has to be understood in the sense of formal power series in $\hbar$ for functionals which do not have only finitely many non--vanishing functional derivatives \cite{BDF,FredenhagenRejzner,FredenhagenRejzner2}. By construction, the $\star$--product is compatible with canonical commutation relations among linear fields
\beq\label{eq:commutation}
[F_f,F_g]_\star\doteq F_f\star F_g-F_g\star F_f= i\hbar \De(f,g)\doteq i\hbar\langle f,\De g\rangle
\eeq
and with the involution $*$
\beq\label{eq:involution}
(F\star G)^* = G^* \star F^*\,.
\eeq
Altogether we arrive at the following definition.
\begin{definition}\label{def:free_algebra_off} The {\bf off--shell algebra of observables} corresponding to the field theory defined by \eqref{def:KGeqn} is the involutive algebra
\beq
\mA := \left(\Fmuc,\star,*\right)\,.
\eeq
\end{definition}

The Hadamard condition for $\De^+$ determines $\De^S$ only up to a smooth (and symmetric) part, however, algebras constructed with different choices of $\De^+$ are isomorphic. Given two $\De^+$, $\De^{+\prime}$ with the above--mentioned properties and the related algebras $\mA = \left(\Fmuc,\star,*\right)$, ${\mA}^\prime = \left(\Fmuc,\star^\prime,*\right)$, the isomorphism between the algebra $\mA$ and ${\mA}^\prime$ is given by
\beq
\label{def:alpha}
\alpha_w: \mA^\prime\to  {\mA}\qquad 
F\mapsto\alpha_w(F)\doteq\exp \hbar\,\Gamma_{w}(F)=
\sum_{n\geq 0}\frac{\hbar^n}{n!}\left\langle w^{\otimes n},F^{(2n)}\right\rangle\qquad
w\doteq \De^+-\De^{+\prime}\,.
\eeq
Note that $\alpha_w$ maps $\Floc$ into itself and that $w$ is real by definition. We now define the on--shell version of $\mA$.

\begin{definition}\label{def:free_algebra} The {\bf on--shell algebra of observables} corresponding to the field theory defined by \eqref{def:KGeqn} is the quotient
\beq
\Aon := \mA/\mI\,,
\eeq
where, whenever the $\star$--product is implemented by a $\De^+$ satisfying $P\circ \De^+ = \De^+ \circ P=0$, $\mI$ is the (closed) $\ast$--ideal in $\mA$ generated by functionals $F\in\Fmuc$ of the form
\begin{align*}
F(\phi)&=\int_{\mM^n}d\mu_g(x_1)\ldots d\mu_g(x_n)\big\{\left(P_{x_1}T(x_1,\ldots,x_n)\right)\phi(x_1)\ldots\phi(x_n)-\\
&\qquad\qquad-T(x_1,\ldots,x_n)j(x_1)\phi(x_2)\ldots \phi(x_n)\big\}\,,
\end{align*}
with $T$ being an arbitrary symmetric distribution which satisfies $\WF(T)\cap\left(\overline{V}^n_+\cup\overline{V}^n_-\right)=\emptyset$. If a $\star$--product $\star^\prime$ is implemented by a general $\De^{+\prime}$ of Hadamard form, then the corresponding ideal $\mI^\prime\subset \mA^\prime=(\Fmuc,\star^\prime,\ast)$ is defined as $\mI^\prime\doteq\alpha^{-1}_{w}(\mI)=\alpha_{-w}(\mI)$, where $\alpha_w:\mA^\prime\to\mA$ \eqref{def:alpha} is the $*$--isomorphism between $\mA^\prime$ and an algebra $\mA=(\Fmuc,\star,\ast)\supset \mI$, in which the $\star$--product is implemented by a $\De^+$ satisfying $P\circ \De^+ = \De^+ \circ P=0$.
\end{definition}

Note that the off--shell $\mA$ does not depend on the source term $j$ in \eqref{def:KGeqn}. We shall remain off--shell in the following as this is necessary for a consistent discussion of perturbation theory.

\begin{remark}\label{rem:bu}
 The subalgebra $\Areg$ of $\mA$ defined as 
\[
\Areg \doteq (\Freg,\star,*)
\]
is generated by the identity and linear functionals and is a concrete representation of the ``normal--ordered'' Borchers--Uhlmann algebra, see e.g. \cite{HW01}. The usual Borchers--Uhlmann algebra is obtained from $\Areg$ as $\alpha_{-\De^S}(\Areg)$, where $\alpha$ is defined as in \eqref{def:alpha}, and is characterised by a $\star$--product defined only in terms of the causal propagator $\Delta$.
\end{remark}

\subsection{Algebras of interacting observables in general non--linear theories}
\label{sec:interactingobservables}

Theories whose dynamics are governed by general non--linear equations are usually constructed perturbatively over linear (or affine) theories, as this is often the only available possibility. Consequently, interacting observables are identified as formal power series in the perturbation with coefficients in $\mA$ by means of a suitable map. In order to explicitly construct this map, we need to introduce a new product among elements of $\Floc$, the time--ordered product $\T$.

The time--ordered product is a product characterised by symmetry and 
\begin{gather}\label{def:time ordered product}
F \T G\doteq F\star G\qquad\textrm{if } F\gtrsim G,
\end{gather}
where $F\gtrsim G$ means that there exists a Cauchy surface (see \cite{BGP} for a definition) $\Si$ such that $\supp(F)\subseteq J^+(\Si)$ and $\supp(G)\subseteq J^-(\Si)$. We may first consider the time--ordered product on $\Freg$. In this case, the time--ordered product can be written by a ``contraction exponential'' similar to the one defining the $\star$--product, namely as
\beq\label{def:timeordered2}
F\T G=F\cdot G+ 
\sum_{n\geq 1}\frac{\hbar^n}{n!}\left \langle\left(\De^F\right)^{\otimes n} , F^{(n)}\otimes  G^{(n)}\right\rangle\qquad F,G\in \Freg,
\eeq
where $\De^F=\Delta^++i\Delta^A$ is the {\bf Feynman propagator} associated to $\De^+$. 

The product defined in \eqref{def:timeordered2} is not well--defined on generic elements of $\Fmuc$. One reason for this is the fact that the wave front set of the integral kernel of $\Delta^A$ contains the wave front set of the $\delta$--distribution because $P\Delta^A = \Id$ with $P$ as in \eqref{def:KGeqn}. Consequently, the product \eqref{def:timeordered2} is ill--defined on $\Floc$ because pointwise powers of $\Delta^F$ are ill--defined. To overcome this problem it is convenient to consider the time--ordered product of local functionals as a multilinear map from multilocal functionals to $\mA$
\beq\label{def:Fmloc} T:\Fmloc \to \mA\,,\qquad \Fmloc \doteq \bigoplus^\infty_{n=0}\Floc^{\!\!\!\!\!\!\!\otimes n}\eeq
satisfying a set of axioms \cite{BF00, HW01,HW02,HW05, BDF}:

\begin{enumerate}
\item $T(F)=F$ for all constant and linear functionals $F$, symmetry and the causal factorisation property:
$$T(F_1,\ldots,F_n)=T(F_1,\ldots,F_k)\star T(F_{k+1},\ldots,F_n)\,,$$
if the supports $\supp F_i, i = 1,\ldots,k$ of the first $k$ entries do not intersect the past of the supports $\supp F_j,j = k+1,\ldots,n$ of the last $n-k$ entries,
\item  the product rule for functional derivatives (also called $\phi$--independence):
\beq\label{eq:phi_independence_of_T}
 T\left(F_1,\ldots, F_n\right)^{(1)}=\sum^n_{j=1} T(F_1,\ldots, F^{(1)}_j,\ldots, F_n)\,,\eeq
\item suitable locality and covariance properties w.r.t. isometric and causality preserving embeddings $(\mM_1,g_1)\to (\mM_2,g_2)$,
\item the microlocal spectrum condition, which is a remnant of translation invariance:
let us consider $n$ local functionals $F_{m_i,f_i}$, $i\in\{1,\ldots,n\}$ of the form \eqref{eq:smearedpolynomial}, i.e. $n$ smeared field polynomials, where also polynomials in derivatives of the field are allowed and let us set 
$$\omega_n(f_1,\ldots,f_n)\doteq \left.T\left(F_{m_1,f_1},\ldots,F_{m_n,f_n}\right)\right|_{\phi=0}\,.$$
This defines a distribution $\omega_n\in\mD^{\prime}_\bC(\mM^n)$ and one demands $\WF(\omega_n)\subset \mV^T_n$, where $\mV^T_n\subset T^*\mM^n\setminus\{0\}$ is defined as follows. We define ``decorated graphs'' $G\subset (\mM,g)$ as graphs embedded in $\mM$ whose 
vertices are points $x_1, \dots, x_n \in \mM$ and whose edges $e$ are oriented null--geodesics $\gamma(e)$. 
We denote by $p_e$ the coparallel and cotangent covectorfield of $\gamma(e)$ and denote by $s(e)=i$ and $t(e)=j$ the source and target of an edge connecting the points $x_i$, $x_j$ with $i<j$. Moreover, if $x_{s(e)} \notin J^\pm(x_{t(e)})$, then $p_e$ is required to be future/past--directed. We may now define $\mV^T_n$ as
\begin{align*}
\label{gamtdef}
\mV^T_n \doteq &\;
\Big\{(x_1, \ldots, x_n, \xi_1,\ldots\xi_n) \in T^*\mM^n\setminus\{0\}\mid 
\exists \,\, \text{decorated graph $G$ with vertices} \nonumber\\
& \quad\text{$x_1, \dots, x_n$ s.t.
$\xi_i = \sum\limits_{e: s(e) = i} p_e(x_i) - \sum\limits_{e: t(e) = i} p_e(x_i) 
\quad \forall i$} \Big\}. 
\end{align*}
\item unitarity
\beq\label{eq:unitarity}
T\left(F_1,\ldots,F_n\right)^* = \sum_{\mP=I_1\uplus\ldots\uplus I_j}(-1)^{n+j}\prod^{\star}_{I\in\mP}T\left(\bigotimes_{i\in I}F^*_i\right)\,,
\eeq
where $\mP=I_1\uplus\ldots\uplus I_j$ denotes a partition of $\{1,\ldots,n\}$ into $j$ pairwise disjoint, non--empty subsets $I_i$,
\item the Leibniz rule (also called action ward identity)
\beq\label{def:Leibniz}
T(F_1,\ldots, F_n) = 0\quad\text{if}\quad F_i(\phi) = \int_\mM d B(\phi)\quad \text{for at least one $i\in\{1,\ldots,n\}$}\,
\eeq
where $B(\phi)$ is a compactly supported three-form,
\item suitable scaling properties, cf. \cite{HW01}, and and a suitable smooth or analytic dependence on the metric, see \cite{HW01}, however also \cite{Khavkine:2014zsa},
\item the Principle of Perturbative Agreement for at most quadratic interactions with at most two derivatives, cf. \cite{HW05} and Section \ref{sec:PPA}.
\end{enumerate}

Following the approach of Epstein and Glaser \cite{EG}, causal factorisation property is used to construct a solution to these axioms by an induction over the number of factors. In particular, at each induction step certain expressions of $\De^F$, which are a priori defined only outside the total diagonal of $\mM^n$, are extended to the full $\mM^n$ by using methods similar to the ones introduced by Steinmann \cite{Steinmann}, see e.g. 
\cite{BF00, HW02,HW05,BDF,FredenhagenRejzner,FredenhagenRejzner2} for details. The resulting time--ordered map is not unique but the freedom left is local and the freedom of the time--ordered products appearing in the perturbative construction of renormalisable interacting models may be absorbed by a redefinition of the parameters of the interacting model under consideration (ibid.). 

For the special case of local functionals, the time--ordered map $T$ corresponds to the definition of local and covariant Wick polynomials as elements in $\mA$ (ibid.). In the case of several factors, one can define a time--ordered product $\T$ induced by the map $T$ as
\beq \label{def:timeorderedproductmap}F_1 \T \ldots \T F_n \doteq T\left(T^{-1}(F_1),\ldots,T^{-1}(F_n)\right).\eeq
It has been proven in \cite{FredenhagenRejzner3}, that, as suggested by the notation, the resulting $\T$ is indeed an iterated binary and associative operation on a subset of $\Fmuc$ which contains the space of functionals consisting of formal series of time--ordered products of local and/or regular functionals

\beq\label{def:Ftloc}
\Ftloc \doteq \left\{F\in\Fmuc\,|\, F=\sum_{n\ge 0} F_{n,1}\T \,\ldots \,\T F_{n,n}\,, F_{k,l}\in\Floc\cup \Freg \right\}\,.
\eeq
Here, the time--ordered product of a regular functional $F$ with an arbitrary $G\in\Ftloc$ is uniquely defined by \eqref{def:timeordered2} also for non--regular $G$. 
For later purposes we define a subalgebra $\A0\subset\mA$ as
\beq\label{def:A0}
\A0\doteq \text{the algebra $\star$--generated by $F\in\Ftloc$}\,.
\eeq
\begin{remark}\label{remark:TonMUC}
The renormalised time--ordered product constructed as explained above fails to be well--defined on the full space of microcausal functionals $\Fmuc$. This may be seen by naively evaluating the expression $F \T (G\star H)$, where $F$, $G$ and $H$ are non--linear local functionals, by means of \eqref{def:timeordered2} and by observing that non--local divergencies occur in the evaluation of this expression. Notwithstanding, the domain of $\T$ is sufficiently large for perturbation theory.
\end{remark}

Once the renormalised time--ordered product is constructed, it is possible to define interacting observables as formal power series with coefficients in $\mA$. To this avail we consider a model defined by the action
\[
\cS = \cS_1 + V\,
\]
where $\cS_1$ is an action functional corresponding to Euler--Lagrange equations of the form \eqref{def:KGeqn} and $V\in\Floc$ is arbitrary, but usually considered to be real $V^*=V$. For an arbitrary functional $F\in\Ftloc$, we define the {\bf$S$-matrix} $S_{1,F}\in\Ftloc$ by 
\beq\label{def: S-matrix}
S_{1,F}\doteq
\exp_{T_1}\left(\frac{i}{\hbar}F\right)\doteq
\sum_{n\geq 0}\frac{i^n}{n!\hbar^n}\underbrace{F\T[1] \ldots\T[1] F}_{n\textrm{ times}}\,,
\eeq
and we define the {\bf relative $S$-matrix} $\mS_{1,V}(F)$ associated to $\cS_1$, $V$ and an arbitrary functional $F\in\Ftloc$ as
\beq\label{def:relativeSmatrix}
\mS_{1,V}(F)\doteq S^{-1}_{1,V}\star_1 S_{1,V+F}\,.
\eeq
Here, $\star_1$ and $\T[1]$ are $\star$-- and time--ordered products related to $\cS_1$ and $S^{-1}_{1,V}$ is the inverse of $S_{1,V}$ w.r.t. $\star_1$.  We may then define the {\bf (retarded) quantum M\o ller map} via the Bogoliubov formula as
\beq
\label{def:quantum moller operator}
\molh{1}{V}(F)\doteq \left.\frac{\hbar}{i}\frac{d}{d\lambda}\mS_{1,V}(F)\right|_{\lambda=0}=S_{1,V}^{-1}\star_1 (S_{1,V}\T[1] F)\qquad F\in\Ftloc[1]\,,
\eeq
where $\Ftloc[1]$ is defined by means of $\T[1]$ as in \eqref{def:Ftloc}, which is a formal power series in $V$ (and its functional derivatives). As the name suggests, the quantum M\o ller map at zeroth order in $\hbar$ equals, in the sense of formal power series in $V$, the classical M\o ller map which we shall discuss in the next section \cite{DF2}. Note that $\molh{1}{V}(F)$ is a formal power series with values in $\A0_1$, cf. \eqref{def:A0}.

By means of the quantum M\o ller map we can define the algebra of interacting observables $\mA_{1,V}$ and its regular version $\Areg_{1,V}$ corresponding to the base theory $\cS_1$ and the interaction $V$.

\begin{definition}\label{def:interactingalgebra}
The {\bf off--shell algebra of interacting observables} $\mA_{1,V}$ is the $*$--algebra that is $\star_1$--generated by the functionals $\molh{1}{V}(F)$ with $F\in \Ftloc[1]$. The {\bf off--shell algebra of regular interacting observables} $\Areg_{1,V}$ is the $*$--algebra that is $\star_1$--generated by the functionals $\molh{1}{V}(F)$ with $F\in\Freg$. The {\bf on--shell algebra of interacting observables} $\Aon_{1,V}$ is the quotient $\mA_{1,V}/\mI_{1,V}$ where, considering $\mA_{1,V}$ as a subalgebra of $\mA_1$, the $\ast$--ideal $\mI_{1,V}$ is defined as $\mI_{1,V}\doteq\mI_1\cap \mA_{1,V}$, where $\mI_1$ is defined as in Definition \ref{def:free_algebra}.
\end{definition}

If one considers a regular interaction $V\in\Freg$, and corresponding regular interacting observables, one may define 
their algebra in a direct manner as 
\beq
\tAreg_{1,V}\doteq(\Freg,\vstar,\vast)\eeq 
with the interacting $\star$--product $\vstar$ and the interacting involution $\vast$ defined as
\beq\label{def:interactingalgebradirect}
\label{operations on the interacting algebra}
F\vstar G\doteq(\molh{1}{V})^{-1}\left(\molh{1}{V}(F)\star_1\molh{1}{V}(G)\right),\qquad
F^{\vast}\doteq (\molh{1}{V})^{-1}\left(\molh{1}{V}(F)^*\right)\,.
\eeq
Here, the inverse quantum M\o ller map is given explicitly by
\beq
(\molh{1}{V})^{-1}(F)=S_{1,-V}\T[1] (S_{1,V}\star_1 F)\,.
\eeq
However, $\vstar$, $\vast$ are not well--defined for general $F,G\in\Ftloc$ and $V\in\Floc$ because the inverse M\o ller map is not well--defined on a sufficiently large domain on account of Remark \ref{remark:TonMUC}, see \cite{KasiaThesis} for details. Consequently, \eqref{def:interactingalgebradirect} and the naively defined $\widetilde{\mA}_{1,V}\doteq (\mF_{1,V},\vstar,\vast)$, where $\mF_{1,V}$ is the space of functionals $\vstar$--generated by $\Ftloc[1]$, are in general ill--defined for $V\in\Floc$. In spite of this fact, one may always think of $\widetilde{\mA}_{1,V}$ as being represented via the quantum M\o ller map on the well--defined algebra $\mA_{1,V}$, because formally $\molh{1}{V}$ is a $*$--isomorphism between $\widetilde{\mA}_{1,V}$ and $\mA_{1,V}$. This point of view is in general sufficient for perturbation theory. However, we shall demonstrate in the following section that $\tAreg_{1,V}$ is well--defined at least if $V$ is a quadratic local functional. Note that the interacting involution $\vast$, provided it is well--defined, in general differs from the simple complex conjugation $*$ if $V$ is non--linear.

We close our brief review of the functional approach to perturbative quantum field theory by the simple, but powerful observation that the time--ordered product associated to $\vstar$, provided the latter is well--defined, is $\T[1]$, see e.g. \cite{Lindner:2013ila}.

\begin{lemma}\label{pr:timeorderedprodof1Q}
For all $V,F,G\in \Ftloc[1]$ such that $F\gtrsim G$
$$\molh{1}{V}\left(F\T[1]G\right)=\molh{1}{V}(F)\star_1\molh{1}{V}(G)$$
\end{lemma}
\begin{proof} We first note that
\[
\molh{1}{V}(F\T[1] G) =  \left.\frac{\hbar^2}{i^2}\frac{d}{d\lambda} \frac{d}{d\mu} \mathscr{S}_{1,V}(\lambda F+\mu G) \right|_{\lambda,\mu=0}.
\]
By the properties of the time--ordered product $\T[1]$, $F\gtrsim G$ implies $\mathscr{S}_{1,V}(F+G)=\mathscr{S}_{1,V}(F)\star_1\mathscr{S}_{1,V}(G)$, i.e. the causal factorisation property of the relative $S$--matrix, see Appendix \ref{sec:factorisation}. Using this, we find
\[
\molh{1}{V}(F\T[1] G) =  \left.\frac{\hbar^2}{i^2}\frac{d}{d\lambda} \frac{d}{d\mu} 
\mathscr{S}_{1,V}(\lambda F) \star_1 \mathscr{S}_{1,V}(\mu G) \right|_{\lambda,\mu=0} = 
\molh{1}{V}(F) \star_1  \molh{1}{V}(G).
\]
\end{proof}

\section{The Principle of Perturbative Agreement}
\label{sec:PPA}
We are now in position to start our discussion of the Principle of Perturbative Agreement (PPA). This principle has been introduced in \cite{HW05} as an axiom to fix the freedom arising from the extension of the time--ordered product from regular to local functionals. To discuss this principle in the case considered in this work (cf. Remark \ref{rem:PPAother} for the remaining cases), we consider two at most quadratic action functionals $\cS_1$, $\cS_2$, i.e. two actions of the form
\beq\label{def:quadraticactions}
\cS_i(\phi)=\cS_i^*(\phi) \doteq \frac12\int_\mM f\left(g_i^{\mu\nu} (\nabla_\mu \phi) \nabla_\nu \phi + M_i\phi^2  - 2j_i \phi \right)d\mu_{g_i}\,.
\eeq
Here $M_i, j_i\in \mE(\mM)$ and $g_i$ are two Lorentz metrics such that the spacetimes $(\mM,g_i)$ are globally hyperbolic. Moreover, we need the technical condition \cite{Brunetti:2012ar,Khavkine:2012jf} that 
\beq\label{def:metriccondition}
g_2 > g_1\qquad \Leftrightarrow\qquad J_1^\pm(p)\subset J_2^\pm(p)\; \forall\;p\in\mM.
\eeq
$f$ is a test function introduced to achieve $\cS_i\in \Floc$, however, the Euler--Lagrange equations of $\cS_i$, $\cS_i^{(1,i)}(\phi)=0$, do not depend on $f$ if we choose $f$ equal to 1 on the region of spacetime of interest, which we shall do implicitly throughout. Consequently, we set
\beq\label{def:euler-lagrange}
\cS_i^{(1,i)}(\phi)\doteq P_i\phi - j_i\,,\qquad P_i\doteq -\Box_{g_i}+M_i\,,\qquad i\in\{1,2\}\,.
\eeq

Here and in the following we stress with our notation that the functional derivative in the sense of an integral kernel of a distribution depends on the background metric. I.e. we define for a smooth functional $F$ and $h\in\mD(\mM)$
\beq\label{def:funcderivmetric}
\left\langle F^{(1,i)}(\phi),h\right\rangle_i \doteq \frac{d}{d\lambda} F(\phi+\lambda h)|_{\lambda=0}\,,\qquad \langle f ,h\rangle_i\doteq \int_\mM f h\, d\mu_{g_i}\,.
\eeq
In particular, it holds
\beq\label{def:funcderivmetric2}
F^{(1,1)}(\phi) = \cc F^{(1,2)}(\phi)\,,\qquad \cc\doteq \frac{\sqrt{|\det g_2|}}{\sqrt{|\det g_1|}}\,.
\eeq

 A term of the form $\phi A^\mu\nabla_\mu \phi$ has been excluded from $\cS_i$ because it is equivalent to $-\phi^2 \nabla^\mu A_\mu$ which can be subsumed under $M_i$. We may consider the difference of the two actions as a perturbation $Q(\phi)\doteq \cS_2(\phi)-\cS_1(\phi)$ and write it w.r.t. the volume measure of $g_1$ as
\beq\label{quadratic potential}
Q(\phi)=Q^*(\phi)=\frac12\int_\mM \left(G^{\mu\nu} (\nabla_\mu \phi) \nabla_\nu \phi + M\phi^2 - 2j \phi \right)d\mu_{g_1}
\eeq
where now $M$, $j$ and $G$ are smooth and compactly supported so that $Q\in\Floc$. We note that $M$ does not only quantify a perturbative mass correction, but also a perturbative correction to the coupling of the Klein--Gordon field to the scalar curvature. Moreover, $M$ and $j$ also quantify a change of the metric because e.g. $M=\cc M_2 - M_1$.

\subsection{Formulation of the Principle of Perturbative Agreement}
\label{sec:formulation_PPA}

One can now proceed in two ways to construct the algebra of observables corresponding to the model defined by $\cS_2$. On the one hand we can construct the (off--shell) algebra of observables $\mA_2\doteq (\Fmuc[2],\star_2)$ (we shall suppress the involutions in the following) directly, where $\star_2$ is a $\star$--product corresponding to $P_2$ and $\Fmuc[2]$ is defined as in \eqref{def:microcausal functionals}, but with respect to $g_2$. Alternatively, we can construct it perturbatively over the exact algebra $\mA_1\doteq (\Fmuc[1],\star_1)$, where $Q$ plays the role of the interaction and $\Fmuc[1]$ is defined as in \eqref{def:microcausal functionals}, but with respect to $g_1$. We have argued in the previous section that the algebra $\widetilde{\mA}_{1,Q}$, generated by $\qstar$--products of $F\in\Ftloc[1]$ and with $\qstar$ defined as in \eqref{operations on the interacting algebra}, is ill--defined, but may be considered as represented on the well--defined algebra $\mA_{1,Q}$, defined in Definition \ref{def:interactingalgebra}, via the quantum M\o ller map $\molh{1}{Q}$. Notwithstanding, we shall first discuss the PPA heuristically in terms of the algebra $\widetilde{\mA}_{1,Q}$ in order to outline the essential idea of this principle. 

By the covariance axiom, the time--ordered maps $T_i:\Fmloc\to \mA_i$ \eqref{def:Fmloc} have to be understood as $T_i \doteq T(g_i,M_i,j_i,\star_i)$, i.e. as particular evaluations of a map $T(g,M,j,\star)$, which for an arbitrary but fixed choice of renormalisation constants depends on the background fields $(g,M,j)$ via the terms multiplying these constants. Moreover, we spell out the dependence of $T$ on $\star$ because we consider the free algebras $\mA$ always as algebras with a concrete product $\star$ defined by a fixed $\Delta^+$ of Hadamard form rather than as abstract algebras defined only up to isomorphy.

The motivation which lead the authors of \cite{HW05} to formulate the PPA originates from the point of view that this map $T(g,M,j,\star)$ ought to satisfy covariance conditions w.r.t. the background fields $(g,M,j)$ which are stronger than the locality and covariance conditions formulated in \cite{HW01, BFV}. As we shall see, the PPA enforces coherence conditions between $T(g,M,j,\star)$ evaluated at two arbitrary but fixed different sets of background fields and thus severely restricts the renormalisation freedom of the time--ordered maps. We recall that this implies that the renormalisation freedom of both time--ordered products and Wick polynomials is restricted, because the time--ordered map encodes both types of renormalisation freedom.

With these points in mind, one may formally say that the PPA demands that the time--ordered map is defined in such a way that the perturbatively defined algebra $\widetilde{\mA}_{1,Q}$ is $*$--isomorphic to the subalgebra $\A0_2\subset \mA_2$ given in \eqref{def:A0}. Moreover, recalling that the time--ordered product formally corresponding to $\qstar$ is $\T[1]$ by Lemma \ref{pr:timeorderedprodof1Q}, one may further demand that the $*$--isomorphism intertwines not only $\qstar$ and $\star_2$, but also the corresponding time--ordered products $\T[1]$ and $\T[2]$ and even the full time--ordered maps $T_1$ and $T_2$. All these statements are meant in the sense of formal power series in $Q$ and its functional derivatives with values in $\mA_1$, because $\widetilde{\mA}_{1,Q}$ is defined only perturbatively. For this reason it is necessary to represent the algebra $\mA_2$ on the base algebra $\mA_1$. As we shall discuss in detail in the following section, this may be achieved by means of the {\bf classical M\o ller map} $\qmoller$. We shall define this map precisely in what follows, but for the sake of this discussion the reader may think of $\qmoller$ as the $\hbar\to 0$ limit of $\bqmoller$. Consequently, we may heuristically describe the PPA in terms of the following commutative diagram, where dashed arrows indicate that the source of the arrows and thus the full arrow is formal.

\beq\label{diagramheuristic}
\begin{tikzcd}[column sep=large, row sep=huge] \mA_1   & \widetilde{\mA}_{1,Q} \arrow[dashed]{l}[above]{\text{\normalsize $\bqmoller$}} \arrow[dashed]{dl}{\text{\normalsize $\bq \doteq \qmoller^{-1}\circ \bqmoller$}}  \\
 \mA_2  \arrow{u}{\text{\normalsize $\qmoller$}}
 \end{tikzcd}
\eeq

As this diagram anticipates, and as we shall prove in the following section, the classical M\o ller map $\qmoller$ is a well--defined exact isomorphism between $\mA_2$ and (a subalgebra of) $\mA_1$ if $\star_2$ is defined as the pullback of $\star_1$ via $\qmoller$. Moreover, $\bqmoller$ maps $\widetilde\mA_{1,Q}$ to $\mA_{1,Q}$ and thus effectively to $\mA_{1,Q}$ in the sense of formal power series, as we have already discussed. The candidate for the heuristic isomorphism between $\widetilde\mA_{1,Q}$ and a suitable subalgebra of $\mA_{2}$ may be thus read off from the diagram to be $\bq\doteq \qmoller^{-1}\circ \bqmoller$, which one may think of as being the identity plus the ``pure quantum part of $\bqmoller$''. 

After this heuristic discussion of the PPA, we may now state this principle, as introduced in \cite{HW05}, precisely and rigorously. As our previous discussion implies, the following formulation of the PPA is essentially the strongest condition one can require on the comparison between the perturbative and exact constructions of models with quadratic interactions because $\beta_{1,Q}$ can not be a proper isomorphism between  $\widetilde\mA_{1,Q}$ and $\mA_2$ since the former algebra is ill--defined. On the other hand, since by Lemma \ref{pr:timeorderedprodof1Q} the time--ordered product corresponding to $\vstar$, provided the latter is well--defined, is given by $\T[1]$, one may consider a functional $F\in\Ftloc[1]$, e.g. $F_1 \T[1] \cdots \T[1] F_n$ with $F_i\in \Floc$, heuristically as an element of $\widetilde{\mA}_{1,Q}$. Moreover, as we shall discuss in the subsequent part of the paper, $\beta_{1,Q}$ is well--defined on $\Ftloc[1]$.

 \begin{definition} \label{def:PPA} Consider two Lorentz metrics $g_2>g_1$, cf. \eqref{def:metriccondition}, on $\mM$, two corresponding at most quadratic actions $\cS_1$ and $\cS_2$ of the form \eqref{def:quadraticactions} and $Q=\cS_2-\cS_1$. Moreover, let $\bqmoller$ be the quantum M\o ller map defined in \eqref{def:quantum moller operator}, let $\qmoller$ be the classical M\o ller map\footnote{The classical M\o ller map is defined in \eqref{eq:definition r}, we recall that for the purpose of Definition \ref{def:PPA} it is sufficient to note that $\qmoller=\lim_{\hbar\to 0}\bqmoller$ in the sense of formal power series in $Q$, cf. Proposition \ref{pr:classicallimit}.}, and set $\bq\doteq \qmoller^{-1}\circ \bqmoller$. Finally, let $\star_1$ be a $\star$--product corresponding to $\cS_1$ and let $\star_2$ be the $\star$--product induced by $\star_1$ via $\qmoller$ as $F\star_2 G\doteq \qmoller^{-1}\left(\qmoller(F)\star_1\qmoller(G)\right)$ for arbitrary $F,G\in\Fmuc[2]$. The time ordered map $T$, considered as map $T(g,M,j,\star)$, is said to satisfy the {\bf Principle of Perturbative Agreement (PPA)} if, for $T_i\doteq T(g_i,M_i,j_i,\star_i)$, $i=1,2$, 
\beq\label{eq:PPA}
T_2 = \bq \circ T_1.
\eeq
on $\Fmloc$.
\end{definition}

\begin{remark}\label{rem:PPAequivalent}
Note that, by \eqref{def:timeorderedproductmap}, one may formulate the PPA equivalently by demanding that  
for all $n\in\bN$ and all $F_0, F_1, \ldots, F_n\in\Floc$,
$$T_2(F_0) = [\bq \circ T_1](F_0)\,,\qquad F_1\T[2]\ldots \T[2] F_n = \bq\left(\bq^{-1}(F_1)\T[1]\ldots\T[1]\bq^{-1}(F_n)\right)\,.$$
The first of these two conditions has the following physical interpretation. Local observables in QFT may be expressed in terms of local and covariant Wick polynomials. The time--ordered maps $T_i$ evaluated on local functionals identify local and covariant Wick polynomials as particular elements of the algebras $\mA_i$. Consequently, one may interpret the first of the above conditions by stating that the PPA demands in particular that $\bq$ is effectively the identity on local observables.
\end{remark}

The authors of \cite{HW05} prove that the PPA can be satisfied and demonstrate that the PPA implies that several essential identities from classical field theory hold also in the quantum case, in particular the interacting quantum stress--energy tensor for an arbitrary, not necessarily quadratic interaction is automatically conserved if the PPA holds \cite{HW05}. 

Our strategy to prove the validity of the PPA for quadratic interactions without derivatives differs somewhat from the proof strategy of \cite{HW05} (see also \cite{Zahn} for a proof for higher spin fields and variations of a background gauge field). In fact, we shall first prove a rigorous version of the heuristic commutative diagram \eqref{diagramheuristic}, namely, 

\beq\label{diagramexact}
\begin{tikzcd}[column sep=large, row sep=huge] \Areg_1   & \tAreg_{1,Q} \arrow{l}[above]{\text{\normalsize $\bqmoller$}} \arrow{dl}{\text{\normalsize $\bq \doteq \qmoller^{-1}\circ \bqmoller$}}  \\
 \Areg_2  \arrow{u}{\text{\normalsize $\qmoller$}}
 \end{tikzcd}
\eeq
where $\Areg_1\doteq(\Freg,\star_1)$, $\Areg_2\doteq(\Freg,\star_2)$, and $\tAreg_{1,Q}=(\Freg,\qstar)$, with $\qstar$ defined as in \eqref{operations on the interacting algebra}. To this end, we first first discuss the left arrow in \eqref{diagramexact} (in the stronger sense as it appears in \eqref{diagramheuristic}) in the following Section \ref{sec:exactrelations}, and the top arrow in the subsequent Section \ref{sec:perturbativerelations}, where we establish in particular that $\tAreg_{1,Q}=(\Freg,\qstar)$ is a well--defined algebra and that $\bq$ is a deformation of the type \eqref{def:alpha} and thus a homomorphism for both the $\star$--product and the time--ordered product. These results will hold for arbitrary $Q$ of the form \eqref{def:striclyquadratic}. Finally, in Section \ref{sec:extensionnonlinear} we extend the action of $\bq$ from regular functionals to $\Ftloc[1]$ and prove that the PPA for strictly quadratic $Q$ without derivatives can be satisfied. 

\begin{remark}\label{rem:linear}
We shall prove most results only for purely quadratic $Q$ in order to avoid the distinction between strictly quadratic $Q$ and linear $Q$. However, the PPA for linear $Q$ may be proven along the same lines. In particular, one may show that, if $Q$ is linear, $\star_2=\star_1$ and $\bqmoller = \qmoller$, and thus $\bq$ is the identity on $\Fmuc[1]=\Fmuc[2]$ for linear $Q$.
\end{remark}

\begin{remark}\label{rem:PPAother}
The PPA as reviewed above is only part of the PPA as introduced in \cite{HW05}. Namely, the authors of  \cite{HW05} also require that the perturbative and exact constructions of a field theoretic model given by the sum of a free Lagrangean and an interaction Lagrangean which is of arbitrary order in the field, but a pure divergence, should agree. This leads to the Leibniz rule for time--ordered products, cf. Axiom 6. in the list reviewed in Section \ref{sec:interactingobservables}.
\end{remark}

For later reference, we explicitly state the form of purely quadratic $Q$ for which we shall prove most results, where we recall \eqref{def:metriccondition} and \eqref{def:funcderivmetric2}

\begin{gather}\label{def:striclyquadratic}
Q=Q^*\in\Floc\,,\qquad Q(\phi)\doteq\frac12\int_\mM \left(G^{\mu\nu} (\nabla_\mu \phi) \nabla_\nu \phi + M\phi^2\right)d\mu_{g_1}\,,\\
G^{\mu\nu}=\cc \,g_2^{\mu\nu}-g_1^{\mu\nu}\,,\qquad g_2>g_1\notag\,.
\end{gather}

\subsection{The classical M\o ller map and relations between the exact algebras \texorpdfstring{$\mA_1,\mA_2$}{A1, A2}}
\label{sec:exactrelations}

The elements of $\mA_1=(\Fmuc[1],\star_1)$ and $\mA_2=(\Fmuc[2],\star_2)$ are functionals over field configurations endowed with products pertaining to the actions $\cS_1$ and $\cS_2$. Hence, a mapping of field configurations which intertwines the equations of motions associated to $\cS_1$ and $\cS_2$ and satisfies suitable further properties may be expected to induce by pullback a $*$--homomorphism between $\mA_2$ and $\mA_1$. A candidate for such a map is the classical M\o ller map on configurations. In the following, quantities indexed by ``1'' and ``2'' shall always denote quantities associated to $\cS_1$ and $\cS_2$. Results which are essentially analogous to the ones stated in this section have already been obtained in \cite{HW05, Zahn}.

\begin{definition}\label{def:definition r} Consider two Lorentz metrics $g_2, g_1$ on $\mM$, such that $(\mM,g_1)$ and $(\mM,g_2)$ are globally hyperbolic, two corresponding at most quadratic actions $\cS_1$ and $\cS_2$ of the form \eqref{def:quadraticactions} and $Q=\cS_2-\cS_1\in\Floc$. The {\bf (retarded) classical M\o ller map on configurations} $\qcmoller:\mE(\mM)\to\mE(\mM)$ is defined by the conditions
\beq\label{eq:definition r}
\cS_2^{(1,1)}\circ \qcmoller = \cS_1^{(1,1)}\,,\qquad \left.\qcmoller(\phi)\right|_{\mM\setminus J_2^+(\supp(Q))}=\phi|_{\mM\setminus J_2^+(\supp(Q))}\,,\eeq
where $J_2^+$ indicates the causal future with respect to the metric $g_2$ and where we recall that $F^{(1,1)}$ indicates the first functional derivative of $F$ w.r.t. $g_1$ in the sense of \eqref{def:funcderivmetric}.
\end{definition}

\begin{remark}\label{rem:operatordistributionconventions}
In the following, we shall deal with linear operators and their integral kernels, where the integral kernel of an operator depends on the metric via the covariant volume measure. I.e. denoting by $[A]_i$ the $g_i$--integral kernel, we have $[A]_1(x,y)=[A]_2(x,y)\cc(y)$. We shall use the following convention regarding integral kernels throughout the remaining part of work. Integral kernels of operators pertaining to $\cS_i$ shall always be considered w.r.t. $g_i$, i.e. $\Delta^\sharp_i(x,y)\doteq [\Delta^\sharp_i]_i(x,y)$ for $\sharp\in\{+,-,F,R,A\}$ and $i=1,2$. Correspondingly, all contraction formulae such as \eqref{eq:exp-product explicit} and \eqref{def:alpha} are considered to be defined by means of a fixed distribution kernel and are thus independent of the chosen metric, because e.g.
$$F\star_1 G=\sum^\infty_{n=0}\frac{\hbar^n}{n!}\left\langle \Delta^{\otimes n}_1,F^{(n,1)}\otimes G^{(n,1)}\right\rangle_1 = \sum^\infty_{n=0}\frac{\hbar^n}{n!}\left\langle \Delta^{\otimes n}_1,F^{(n,2)}\otimes G^{(n,2)}\right\rangle_2\,.$$
For the remainder of the paper, we will thus mostly suppress the dependence of the functional derivative and canonical pairing on the metric where it is understood that whenever this dependence is omitted, the $g_1$--functional derivative and $g_1$--pairing are implied.
\end{remark}

We now show that the precise form of $\qcmoller$ is an expression of Yang--Feldman type, where we recall \eqref{def:funcderivmetric2}.
\begin{proposition}\label{pr:rq-inverse}
The unique solution of the conditions \eqref{eq:definition r} is the map $\qcmoller:\mE(\mM)\to\mE(\mM)$
\[
  \qcmoller \doteq \Id-\De^R_{2}\cinv\circ \q  
\]
which moreover satisfies
\[
\qcmoller\circ \left(\Id+\De^R_{1}\circ \q\right)= \left(\Id+\De^R_{1}\circ \q\right)\circ\qcmoller=\Id
\]
on $\mE(\mM)$. Here, the functional derivative $\q$ defines an affine map on $\mE(\mM)$ with formally selfadjoint linear part, which we denote by the same symbol.
\end{proposition}
\begin{proof}
First of all we recall that $\De^R_{2}$ and  $\De^R_{1}$ map smooth functions with past--compact support to smooth functions with past--compact support continuously in the topology of $\mE(\mM)$ \cite{BGP}. Here $\mO\subset (\mM,g_i)$ is called past--compact, if for each compact $K\subset \mM$, $\mO\cap J_i^{-}(K)$ is compact. Hence, $\qcmoller$ is well--defined on $\mE(\mM)$ due to the compactness of the support of $Q$. Moreover, the fact that the affine map induced by the first functional derivative of $Q$ has a formally selfadjoint linear part follows from the fact that $Q^{(2)}$ is symmetric. 

$\qcmoller$ manifestly satisfies the second condition in \eqref{eq:definition r}. To check that it satisfies the first condition, we set $P_i:= -\Box_{g_i} + M_i$ and may compute for an arbitrary $\phi\in\mE(\mM)$
\beq\label{eq:classmoellerproof}
(\cS_2^{(1)}\circ\qcmoller)(\phi)=\cS_2^{(1)}\left(\phi-\De^R_2\cinv\q(\phi)\right)= \cc P_2 \phi + \cc j_2 - \q(\phi) = P_1 \phi + j_1\,,
\eeq
where we have used that $P_2\circ \De^R_2=\Id$ on smooth functions with past--compact support. Uniqueness follows by setting $\psi = \qcmoller(\phi)$ and realising that $\psi$ is the unique solution to the normally hyperbolic Cauchy problem 
$$P_2 \psi = \frac{P_1 \phi + j_1}{\cc} - j_2\,,\qquad  \psi|_{\mM\setminus J_2^+(\supp(Q))}=\phi|_{\mM\setminus J_2^+(\supp(Q))}\,,$$ see \cite[Corollary 5]{Ginoux} for normally hyperbolic Cauchy problems with sources of arbitrary support.

In order to prove that $\left(\Id+\De^R_{1}\circ \q\right)$ is the right inverse of  $\qcmoller$ we may compute
$$
\left(\Id-\De^R_{2}\cinv\circ\q\right)\circ \left(\Id+\De^R_{1}\circ \q\right) =$$$$= \Id -\De^R_{2}\cinv\circ\q + \De^R_{1}\circ\q - \De^R_2\cinv\circ (\cc P_2-P_1) \circ \De^R_1\circ \q
$$
and observe that the last three terms cancel if, for all $\phi\in\mE(\mM)$, $\supp(\De^R_1\q \phi)$ is past--compact w.r.t. $g_2$, because $\De^R_2\circ P_2 = P_1 \circ \De^R_1 = \Id$ on smooth functions with past--compact support. To show the former condition for arbitrary causal relations between $g_2$ and $g_1$, we consider an arbitrary compact set $K$ and an arbitrary $\phi\in\mE(\mM)$ and have to prove that $J^-_2(K)\cap \supp(\De^R_1\q \phi)$ is compact. To this avail, we consider two non--intersecting Cauchy surfaces $\Sigma_1$ and $\Sigma_2$ of $(\mM,g_1)$ s.t. $\Sigma_2\subset J^+_1(\Sigma_1)$ and $(K\cup\supp(Q))\subset (J^+_1(\Sigma_1)\cap J^-_1(\Sigma_2))$; this implies in particular that $\Sigma_1$ and $\Sigma_2$ are also Cauchy surfaces of $(\mM,g_2)$. We then set $A\doteq \supp(\De^R_1\q \phi)\cap J_1^-(\Sigma_2)$ and $B\doteq J^-_2(K)\cap J^+_2(\Sigma_1)$ and note that both $A$ and $B$ are compact. Finally, we observe that $J^-_2(K)\cap \supp(\De^R_1\q \phi)=A\cap B$ is compact. By a similar computation one can show that  $\left(\Id+\De^R_{1}\circ \q\right)$ is the left inverse of  $\qcmoller$.
\end{proof}
Note that, at least in the case when $Q$ does not contain two derivatives, $\qcmoller$ is well--defined and satisfies the above properties also for $Q$ which have past--compact support w.r.t. to $g_1=g_2$.

In the following, we shall deal exclusively with strictly quadratic $Q$. In this case, $\q$ and $\qmoller$ are linear maps. Denoting by $T^\dagger$ the adjoint of the operator $T$ with respect to the canonical pairing $\langle f , h\rangle_1 = \int_\mM f g \,d\mu_{g_1}$, Proposition \ref{pr:rq-inverse} then implies $\qcmoller^\dagger =  \left(\Id-{\q}\circ \cc^{-1}\De^A_2\right)$. Furthermore, $\qcmoller^\dagger$ maps solutions of $P_2\phi=0$ to solutions of $P_1\phi=0$ and we note that $P_1$ and $\cc P_2$ are formally selfadjoint w.r.t. $\langle \cdot,\cdot\rangle_1$. With this in mind, we shall now discuss the relation between the advanced and retarded Green's operators of $P_1$ and $P_2$.

\begin{lemma}\label{pr: delta intertwined} If $Q\in\Floc$ is of the form \eqref{def:striclyquadratic}, but $g_2>g_1$ is not necessarily true, the advanced and retarded Green's operators of $P_1$ and $P_2$ are related as
\[
\De^R_2\cinv=\qcmoller \circ \De^R_1\,,
\qquad
\De^A_2\cinv =\De^A_1 \circ \qcmoller ^\dagger\,,
\]
whereas their causal propagators satisfy
\[
\Delta_2\cinv = \qcmoller  \circ \Delta_1 \circ \qcmoller^\dagger \,.
\]
\end{lemma}
\begin{proof}
Let us indicate $\qcmoller \circ \De^R_1\cc$ by $\De_{R,I}$ and $\De^A_1 \circ \qcmoller^\dagger\cc$ by $\De_{A,I}$. Both $\De_{R/A,I}$ satisfy $P_2\circ \De_{R/A,I}=\De_{R/A,I}\circ P_2=\Id$ as can be shown either directly or by duality with respect to the standard pairing 
$\langle \cdot, \cdot \rangle_1$. We would now like to show $\supp(\De_{R/A,I}f)\subseteq J_2^\pm(\supp(f))$ for all $f\in\mD(\mM)$ independent of the causal relations between $g_2$ and $g_1$; this would imply $\De_{R/A,I}=\De^{R/A}_{2}$ by uniqueness of retarded/advanced fundamental solutions. If $g_2>g_1$, the statement follows immediately. In the remaining cases, we consider an arbitrary $f\in\mD(\mM)$ and assume that there is a point $x\in\mM$ s.t. $x\notin\supp\left(\De^R_2 f\right)$, but $x\in\supp\left(\qcmoller\De^R_1 \cc f\right)$. We now consider two non--intersecting Cauchy surfaces $\Sigma_2$, $\Sigma_1$ of $(\mM,g_1)$ s.t. $\Sigma_2\subset J^+_1(\Sigma_1)$, $J^+_1(\Sigma_1)\cap \left(\supp(Q)\cup \supp(f)\cup\{x\}\right)=\emptyset$; this implies in particular that these Cauchy surfaces are also Cauchy surfaces for $(\mM,g_2)$. We then consider a smooth function $\chi$ which equals $1$ on $J^-_1(\Sigma_1)$ and $0$ on $J^+_1(\Sigma_2)$. We may then compute
\begin{align*}
\qcmoller\De^R_1 \cc f &= \left(\Id - \De^R_2\cinv\circ \q\right)\chi\De^R_1 \cc f + \left(\Id - \De^R_2\cinv\circ \q\right)(1-\chi)\De^R_1 \cc f \\
&= \De^R_2\cinv P_1\chi\De^R_1 \cc f + \left(\Id - \De^R_2\cinv\circ \q\right)(1-\chi)\De^R_1 \cc f\\
&= \De^R_2\cinv (-\Box_{g_1}\chi)\De^R_1 \cc f +\De^R_2\chi f+ \left(\Id - \De^R_2\cinv\circ \q\right)(1-\chi)\De^R_1 \,,
\end{align*}
where we used that $\supp(\chi\De^R_1 \cc f)$ is compact. The retardation properties of $\De^R_{1/2}$ and the support properties of $\chi$ and $f$ now imply that sufficiently small neighbourhoods of $x$ can not be contained in the support of either of the terms in the last expression. The advanced case can be shown analogously.

By Proposition \ref{pr:rq-inverse}, the linear operator $\qcmoller $  is the left-- and right--inverse of $(\Id - r)$ with $r\doteq -\De^R_1 \circ \q$. Hence, by direct computation $\qcmoller -\Id=\qcmoller \circ r=r \circ \qcmoller $. 
Consider now $\qcmoller \circ \De_1 \circ \qcmoller ^\dagger$. Using the previous relation and its dual, this operator be can factorized as
\begin{gather*}
\qcmoller \circ\De_1\circ\qcmoller^\dagger=
\qcmoller \circ\De^R_1-
\De^A_1\circ\qcmoller ^\dagger+
\qcmoller \circ\De^R_1\circ r^\dagger \circ\qcmoller^\dagger -
\qcmoller \circ r \circ\De^A_1\circ\qcmoller^\dagger.
\end{gather*}
Using the formal selfadjointness of $\q$, we notice that 
the last two summands cancel each other because $\De^R_1\circ r^\dagger =  - \De^R_1\circ \q \circ\De^A_1 = r \circ \De^A_1$. We may now conclude the proof by using the first two equalities of the present lemma
\begin{gather*}
\qcmoller\circ \De_1 \circ\qcmoller^\dagger=
\qcmoller\circ \De^R_1-
\De^A_1\circ\qcmoller^\dagger = \left(\De^R_2-
\De^A_2\right)\cinv= \De_2\cinv\,.
\end{gather*}
\end{proof}

We shall now demonstrate that the classical M\o ller map on configurations preserves the Hadamard condition and maps Gaussian states to Gaussian states.

\begin{proposition}\label{pr:H intertwiner}
The following statements hold for a $Q\in\Floc$ of the form \eqref{def:striclyquadratic}, omitting the condition $g_2>g_1$.
\begin{itemize}
\item[(i)] If $\Delta^+_1$ is an operator whose integral kernel is of Hadamard type for the theory $\cS_1$, then 
\[
\Delta^+_2\doteq \qcmoller\circ \Delta^+_1\circ \qcmoller^\dagger \cc
\]
is an operator whose integral kernel is of Hadamard type for the theory $\cS_2$.
\item[(ii)] Moreover if $\Delta^+_1$ is an operator whose integral kernel is the two--point function of a Hadamard state on $\mA_1$, then $\Delta^+_2$ is an operator whose integral kernel is the two--point function of a Hadamard state on $\mA_2$.
\end{itemize}

\end{proposition}

\begin{proof}Proof of (i): First of all, we notice that, since $Q$ is has compact support and since the integral kernel of $\Delta^+_1$ satisfies the Hadamard condition, $\Delta^+_2\doteq \qcmoller \circ \Delta^+_1 \circ\qcmoller^\dagger\cc$ is a well--defined operator from compactly supported smooth functions to smooth functions, hence, by the Schwartz kernel theorem, it gives rise to a well--defined distribution. To prove the statement, it thus suffices to show that $WF(\De^+_2)=\mV_2^+$ with $\mV_2^+$ defined as in \eqref{def:muc} by means of the causal structure induced by $g_2$.
We already know by the previous Lemma that the antisymmetric part of $\De^+_2$ is the causal propagator $\Delta_2$. Furthermore, $\De^+_2$ is a weak bisolution of the Klein--Gordon equation up to smooth functions because the same holds for $\De^+_1$ and because $\cc P_2\circ \qcmoller = P_1$. 
The statement now follows by standard arguments including the propagation of singularity theorem from the observation that $\qcmoller$ is the identity outside of the causal future of $\supp(Q)$, and thus $\De^+_2 = \De^+_1$ there.

Proof of (ii): We need to check positivity of $\De^+_2$, i.e. we have to check that $\langle \overline{f},\De^+_2 f\rangle_2\ge 0$ for all $f\in\mD_\bC(\mM)$. 
Using the definition of $\De^+_2$ and recalling \eqref{def:funcderivmetric} and \eqref{def:funcderivmetric2}, we may compute,  
\[
\left\langle \overline{f},\De^+_2 f\right\rangle_2 = \left\langle \cc \overline{f},\De^+_2 f\right\rangle_1 = \left\langle \overline{\qcmoller^\dagger \cc f},\Delta^+_{1}  \qcmoller^\dagger \cc f\right\rangle_1\ge 0 \,,
\]
where we have used the fact that $Q$ is real and thus $\qcmoller$ commutes with complex conjugation, and where we note that $\qcmoller^\dagger$ maps smooth functions to smooth functions of compact support.\end{proof}

With the map $\qcmoller$ between configurations at our disposal, we may construct by pullback a map on functionals.

\begin{definition}\label{def:classicalmoller} For an arbitrary $F\in\Fmuc[2]$, where $\Fmuc[2]$ is defined as in \eqref{def:microcausal functionals}, but with respect to $g_2$, and $Q=\cS_2-\cS_1\in\Floc$ of the form \eqref{quadratic potential}, we  define the {\bf (retarded) classical M\o ller map} by
$$\qmoller(F) \doteq F\circ \qcmoller$$
\end{definition}

The map $\qmoller$ may be thought of as being the off--shell version of the map $\tau^{\text{ret}}$ in \cite{HW05}, which is defined on--shell. Since $\qcmoller$ is invertible, its pullback is invertible as well. The next theorem shows that $\qmoller$ is in fact a $*$--isomorphism between $\mA_2$ and a subalgebra of $\mA_1$.

\begin{theorem}\label{th:classical homomorphism}
Let $Q\in\Floc$ be of the form \eqref{def:striclyquadratic}, and let $\Delta^+_1$ be a linear operator whose integral kernel is of  Hadamard form w.r.t. the Klein--Gordon operator $P_1$. Moreover, let $\De^+_2\doteq \qcmoller\circ \Delta^+_1 \circ \qcmoller^\dagger \cc$ and let $\mA_1\doteq (\Fmuc[1],\star_1)$, $\mA_2\doteq (\Fmuc[2],\star_2)$, where $\star_1$ and $\star_2$ are the $\star$--products constructed by means of $\Delta^+_1$ and $\Delta^+_2$, respectively, and where $\Fmuc[i]$ is defined as in \eqref{def:microcausal functionals}, but with respect to $g_i$. Then, $\qmoller$ satisfies the following properties.
\begin{itemize}
\item[(i)] The inverse of $\qmoller$, defined by $\qmoller^{-1}(F)\doteq F\circ \qcmoller^{-1}$, is well--defined on $\Fmuc[1]$ and maps $\Fmuc[1]$ to $\Fmuc[1]$.
\item[(ii)] $\qmoller$ is well--defined on $\Fmuc[2]$ and maps $\Fmuc[2]$ injectively to $\Fmuc[2]\subset \Fmuc[1]$.
\item[(iii)] $\qmoller$ restricts to a bijective map between $\qmoller:\Freg\to\Freg$.
\item[(iv)]$\qmoller$ induces a $*$--isomorphism $\qmoller:\mA_2\to\qmoller(\mA_2)\subset\mA_1$, which restricts to a $*$--isomorphism between $\Areg_2$ and $\Areg_1$ and descends to a $*$--isomorphism between the on--shell algebras $\Aon_2$ and $\qmoller(\Aon_2)\subset\Aon_1$ constructed as in Definition \ref{def:free_algebra}.
\end{itemize}
\end{theorem}
\begin{proof} Proof of (i). The statement that $\qmoller^{-1}$ is well--defined on $\Fmuc[1]$ and maps $\Fmuc[1]$ to itself follows by an application of \cite[Theorem 8.2.14]{Hormander} about compositions of distributions and from $\qcmoller^{-1}=\Id+\De^R_1\circ \q$.

Proof of (ii). We first observe that $\Fmuc[2]\subset\Fmuc[1]$ follows directly from $g_2>g_1$. The statement that $\qmoller$ is a well--defined on $\Fmuc[2]$ and maps $\Fmuc[2]$ to itself follows by a further application of \cite[Theorem 8.2.14]{Hormander}. Finally, injectivity follows from $\qmoller^{-1}\circ\qmoller=\Id$ on $\Fmuc[2]$, which in turn follows from (i) and $\Fmuc[2]\subset\Fmuc[1]$.

Proof of (iii). This statement follows from the fact that $\qcmoller:\mD_\bC\to\mD_\bC$ is a bijection on account of the compact support of $Q$.

Proof of (iv). Since $Q=Q^*$, $\qmoller$ is real--linear and preserves the $*$--operation on $\mA_1$ and $\mA_2$ defined by complex conjugation. Thus, on account of (ii) and (iii), the off--shell part of the statement is proven if we show that $\qmoller$ intertwines $\star_1$ and $\star_2$. By the definition of these products and of $\qmoller$ as a pullback of $\qcmoller$, this follows from
$$\qmoller(F)\star_1\qmoller(G)=\sum_{n\geq 0}\frac{\hbar^n}{n!}\left\langle \left(\De^+_1\right)^{\otimes n},(F\circ \qcmoller)^{(n)}\otimes (G\circ \qcmoller)^{(n)}\right\rangle = $$
$$=\qmoller\left(\sum_{n\geq 0}\frac{\hbar^n}{n!}\left\langle\left(\qcmoller \circ \De^+_1\circ \qcmoller^\dagger \right)^{\otimes n}, F^{(n)}\otimes G^{(n)}\right\rangle\right)=\qmoller(F \star_2 G)
$$
for arbitrary $F,G\in\Fmuc[2]$, where we note that $\qmoller(F)\star_1\qmoller(G)$ is well--defined because $\qmoller(F)\subset \Fmuc[1]$ and where we recall Remark \ref{rem:operatordistributionconventions}. The on--shell part of the statement follows from the fact that $\cc P_2\circ \qcmoller = P_1$, which implies $\qmoller(\mI_2)\subset\mI_1$ for the on--shell ideals $\mI_2\subset \mA_2$ and $\mI_1\subset \mA_1$.
\end{proof}

\begin{remark} If $g_2\neq g_1$ and $g_2>g_1$, then $\Fmuc[1]\not\subset \Fmuc[2]$. An example of $F\in\Fmuc[1]$, $F\notin\Fmuc[2]$ may be given by choosing a coordinate system on a patch $U\subset \mM$ and setting $F(\phi)\doteq \langle h,\phi\rangle$, $h\doteq  f/(x^\mu v_\mu + i \epsilon)$ where $\supp(f)\subset U$ is compact and $v$ is a covector which is time--like w.r.t. $g_2$ but space--like w.r.t. $g_1$. Since we have $P_1\circ \qcmoller^{-1} =\cc P_2$, we may bound the wave front set of $ \qcmoller^{-1} h = (\qmoller^{-1}(F))^{(1)}$ as $\WF(P_2 h)\subset \WF(\qcmoller^{-1} h)\subset \WF(h)\cup \text{Char}(P_1)$, which proves that $\qmoller^{-1}(F)\notin \Fmuc[2]$. Consequently, $\qmoller$ can not be a bijection between $\Fmuc[2]$ and $\Fmuc[1]$ and a $*$--isomorphism between $\mA_2$ and $\mA_1$ if $g_2\neq g_1$. Moreover, one can check $\qmoller(\A0_2)\not\subset \A0_1$ for the subalgebras $\A0_i\subset \mA_i$, $\star_i$--generated by $\Ftloc[i]$, cf. \eqref{def:A0}. However, presumably $\qmoller$ is a $*$--isomorphism between suitable topological completions of $\A0_2$ and $\A0_1$. Finally, by the time--slice--axiom and the fact that $\qmoller$ is the identity for functionals which are supported outside of $J^+_2(\supp(Q))$, one can indeed demonstrate that $\qmoller$ is a $*$--isomorphism between the on--shell versions of $\mA_2$ and $\mA_1$, which is proven in \cite{HW05}.
\end{remark}

In the subsequent part of the paper we need to view $\qmoller$ for $Q$ as in \eqref{def:striclyquadratic}, and all quantities pertaining to $\cS_2$, as a formal power series in $Q$ and its functional derivatives with coefficients given in terms of quantities associated to $\cS_1$. By means of Proposition \ref{pr:rq-inverse}, this is achieved by writing $\qcmoller$ as the Neumann series
\begin{gather}\label{eq:Neumann series}
\qcmoller=
\left(\Id+\De^R_1\circ \q\right)^{-1}=
\sum_{n\geq 0}\underbrace{r\circ\,\ldots\, \circ r}_\text{$n$ times}\,,\qquad r\doteq -\De^R_1\circ\q\,.
\end{gather}
Recalling $\De^R_2 = \qcmoller \circ \De^R_1\cc$, $\De^A_2 = \De^A_1\circ \qcmoller^\dagger\cc$ and $\Delta^{(+)}_2 = \qcmoller \circ \Delta^{(+)}_1 \circ \qcmoller^\dagger\cc$, we may view these linear maps, as well as $\star_2$, as formal power series in in $Q$ and its functional derivatives as well, which shall be our point of view throughout the remaining part of this work. While the Neumann series \eqref{eq:Neumann series} is in general formal, which is sufficient for our needs, Lemma \ref{pr:neumann series} in the appendix shows that this series converges in the special case of a pure Minkowski background with a pure mass perturbation. The series does not converge if $Q$ contains second derivatives and thus encodes a perturbative change of the background metric because $R_{1,Q}$ has causal properties pertaining to $g_2$, while the causal properties of the Neumann series \eqref{eq:Neumann series} are determined by $g_1$.

\begin{remark}\label{rem:Hadamardform}
We recall that  the integral kernels of $\De^+_1$ and $\De^+_2$ w.r.t. the metrics $g_1$ and $g_2$ respectively are locally of the form \cite{Radzikowski}
\beq\label{def:Hamardplus}
\De^+_i(x,y)= \lim_{\epsilon\downarrow 0}\frac{1}{8\pi^2}\left(\frac{U_i(x,y)}{\sigma^{\epsilon +}_i(x,y) }+V_i(x,y) \log\left( \lambda^2\sigma^{\epsilon +}_i(x,y)\right)\right)+W_i(x,y)\,,
\eeq
where $\sigma^{\epsilon +}_i(x,y) \doteq \sigma_i(x,y) + i\epsilon (t_i(x)-t_i(y))+\epsilon^2/2$, $t$ is an arbitrary time function, $U_i$ and $V_i$ are the Hadamard coefficients corresponding to the models $\cS_1$ and $\cS_2$, $2\sigma_i$ is the squared geodesic distance corresponding to the metric $g_i$, $\lambda$ is a dimensionful constant and $W_i$ is a smooth and symmetric function which is not uniquely determined by the Hadamard condition. $\De^+_2$ has this local form in the exact sense on account of Proposition \ref{pr:H intertwiner}. However considered as a formal series in $Q$, it still has this form where now $\sigma_2$, $U_2$, $V_2$ and $W_2$ are taken as formal series in $Q$.
\end{remark}

We close the discussion of the classical M\o ller map by stating the already anticipated important result proved in \cite{DF2} (for arbitrary local interactions) that $\qmoller$ is the classical limit of $\molh{1}{Q}$ in the sense of formal power series in $Q$. For this, it is essential that $\hbar$ appears in the correct place in the definition of the $S$--matrix \eqref{def: S-matrix} and the products $\star_1$ and $\T[1]$.

\begin{proposition}\label{pr:classicallimit}
Let $Q\in\Floc$ be of the form \eqref{def:striclyquadratic} and let $\Ftloc[1]$ and $\molh{1}{Q}$ be defined by \eqref{def:Ftloc} and \eqref{def:quantum moller operator} respectively. Then for all $F\in\Ftloc[1]$
$$\lim_{\hbar \to 0}\molh{1}{Q}(F)=\qmoller(F)$$
in the sense of formal power series in $Q$ and its functional derivatives, i.e. for $\qmoller$ considered as the pullback of the Neumann series \eqref{eq:Neumann series}. In particular,
$$\molh{1}{Q}(F)=\qmoller(F)+\hbar\, G\,,$$
where $ G$ is an in general non--vanishing formal power series in $Q$ and its functional derivatives with coefficients in $\Fmuc[1]$.
\end{proposition}

\subsection{Characterisation of the perturbative algebra \texorpdfstring{$\tAreg_{1,Q}$}{A1,Q,reg} and perturbative agreement between \texorpdfstring{$\Areg_2$}{A2,reg} and \texorpdfstring{$\tAreg_{1,Q}$}{A1,Q,reg}}
\label{sec:perturbativerelations}

For the remainder of this paper, we shall work exclusively in the following setting, unless explicitly mentioned otherwise. We consider actions $\cS_1$ and $\cS_2$ of the form \eqref{def:quadraticactions} where $\Floc\ni Q=\cS_2-\cS_1$ is of the form \eqref{def:striclyquadratic}. We consider an arbitrary but fixed $\De^+_1$ of Hadamard form w.r.t. $\cS_1$ and the induced algebras $\mA_1=(\Fmuc[1],\star_1)$ and $\Areg_1=(\Freg,\star_1)$. Moreover, we consider the algebras $\mA_2=(\Fmuc[2],\star_2)$, $\Areg_2=(\Freg,\star_2)$ and propagators $\De^+_2$, $\De^{R/A}_2$ and $\De_2$ related to $\cS_2$ which are uniquely induced by the same quantities related to $\cS_1$ via $\qmoller$ as discussed in the previous section, in particular $\De^+_2\doteq \qcmoller\circ\De^+_1\circ \qcmoller^\dagger\cc$. Given an arbitrary but fixed prescription for the time--ordered map $T(g,M,j,\star)$ we set $T_1 \doteq T(g_1,M_1,0,\star_1)$, define $\T[1]$ on local functionals via \eqref{def:timeorderedproductmap} and recall that $\T[1]$ is in fact well--defined on $\Ftloc[1]$ and unambiguous if one of the factors is regular. Analogously the time--ordered product $\T[2]$ corresponding to $\star_2$ is unambiguously given on regular functionals by 
\beq\label{def:t2} 
F\T[2] G=
\sum_{n\geq 0}\frac{\hbar^n}{n!}\left \langle\left(\De_2^F\right)^{\otimes n}, F^{(n)}\otimes G^{(n)}\right\rangle\,,\qquad F,G\in\Freg\,,\qquad \De^F_2\doteq \De^+_2 + i \De^A_2\,.
\eeq

The aim of this section is to prove the PPA for regular functionals, i.e. to prove that \eqref{diagramexact} is a commutative diagram with $\bq$ intertwining $\T[1]$ and $\T[2]$. To this avail, we first show that $\tAreg_{1,Q}\doteq (\Freg,\qstar)$, with $\qstar$ as in \eqref{operations on the interacting algebra}, is well--defined and analyse the precise form of $\qstar$. As $\tAreg_{1,Q}\doteq (\Freg,\qstar)$ is defined in such a way that $\bqmoller:\tAreg_{1,Q}\to\Areg_1$ is a $*$--isomorphism, we automatically get that $\bq\doteq \qmoller^{-1}\circ \bqmoller:\tAreg_{1,Q}\to\Areg_2$ is a $*$--isomorphism as well. The PPA for regular functionals then follows from the observation that $\bq$ is a deformation and the identity on linear functionals.

\begin{remark}\label{rem:Qrenfreedom} In perturbative QFT on curved spacetimes one would like to work only with interactions corresponding to local and covariant observables. Consequently, given a quadratic interaction functional $Q$, one should rather consider the local Wick polynomials $T_1(Q)$ and correspondingly the quantum and classical M\o ller maps $\molh{1}{T_1(Q)}$ and $\mR_{1,T_1(Q)}$. However, the axioms for local Wick polynomials imply that $T_1(Q)-Q$ is a constant functional for quadratic $Q$ and thus $\molh{1}{T_1(Q)}=\bqmoller$ and $\mR_{1,T_1(Q)}=\qmoller$.
\end{remark}

We begin by demonstrating that $\tAreg_{1,Q}\doteq (\Freg,\qstar)$ is well--defined and isomorphic to $\mA_2$.

\begin{proposition}\label{pr:qmoellermapsreg}
The following statements hold.
\begin{itemize}
\item[(i)] $\bqmoller$ and $(\bqmoller)^{-1}$ are well--defined on $\Freg$ and map regular functionals to formal power series in the functional derivatives of $Q$ with values in $\Freg$.
\item[(ii)] The actions of $\bqmoller$ and $(\bqmoller)^{-1}$ on $\Freg$ are independent of the renormalisation freedom of the time--ordered product. 
\item[(iii)] The interacting $\star$--product $\qstar$ and the interacting involution $*_{1,Q}$, defined as in \eqref{operations on the interacting algebra}, are well--defined on $\Freg$. Consequently
$$\tAreg_{1,Q}\doteq(\Freg,\qstar,*_{1,Q}) $$
is well--defined.
\item[(iv)] $\bq\doteq \qmoller^{-1}\circ \bqmoller:\tAreg_{1,Q}\to\Areg_2$ is a $*$--isomorphism.  
\end{itemize}
\end{proposition}
\begin{proof}Proof of (i). It is sufficient to prove the statement for a regular functional $F$ of order $n$ in $\phi$. To this avail, we compute
\beq\label{eq:regproof0}
\bqmoller(F) = S^{-1}_{1,Q}\star_1 \left(S_{1,Q}\T[1] F\right) = F  + S^{-1}_{1,Q}\star_1\left(S_{1,Q}\T[1] F-S_{1,Q}\star_1 F\right)\,.
\eeq
We now recall that $S_{1,Q}\star_1 F$ is given by an exponential contraction formula of the form \eqref{eq:exp-product explicit}, and the same holds for $S_{1,Q}\T[1] F$, because \eqref{def:timeordered2} is well--defined without renormalisation also for $G_1\T[1] G_2$ where $G_1\in\Ftloc[1]$ and $G_2\in\Freg$. Consequently, $S_{1,Q}\T[1] F-S_{1,Q}\star_1 F$ contains at least one functional derivative of $S_{1,Q}$ and may be computed as
\beq
S_{1,Q}\T[1] F-S_{1,Q}\star_1 F=\sum^n_{k= 1}\frac{\hbar^k}{k!}\left\langle \left(\De^F_1-\De^+_1\right)^{\otimes k} ,(S_{1,Q}\T[1]B_k)\otimes F^{(k)}\right\rangle\,,
\eeq
where
\beq\label{eq:regproof1} B_k = i^k \underbrace{\q\T[1]\,\ldots\,\T[1]\q}_{\text{$k$ times}}+\,i^{k-1}\begin{pmatrix}k\\k-2\end{pmatrix}Q^{(2)}\T[1]\underbrace{\q\T[1]\,\ldots\,\T[1]\q}_{\text{$k-2$ times}}+\ldots\eeq
and where we have used the $\phi$--independence of the time--ordered product, cf.  \eqref{eq:phi_independence_of_T}. Note that the sum in \eqref{eq:regproof0} stops at $n$ because $F$ is of $n$--th order in $\phi$. We would now like to write this sum as $S_{1,Q}\T[1] C$ with a suitable functional $C$. That this holds is not obvious at first glance, because the terms in the aforementioned sum for $k<n$ contain implicitly the pointwise product $(S_{1,Q}\T[1] B_k)\cdot F^{(k)}= S_{1,Q}^{(k)}\cdot F^{(k)}$ as $F^{(k)}$ is a regular functional of order $n-k$ in $\phi$ (with values in $\mD_\bC(\mM^k)$). Nevertheless, we can prove that the sum in \eqref{eq:regproof1} is of the wanted form in the following way. We observe 
$$S_{1,Q}^{(k)}\cdot F^{(k)} = S_{1,Q}^{(k)}\T[1] F^{(k)}-\sum^{n-k}_{j=1}\frac{\hbar^j}{j!}\left\langle \left(\De^F_1\right)^{\otimes j},S_{1,Q}^{(k+j)}\otimes F^{(k+j)}\right\rangle\,$$
where the functional derivatives of $F$ in the second term of this sum are regular functionals of at most $n-k-1$--th order in $\phi$. Iterating this procedure, we can write $S_{1,Q}^{(k)}\cdot F^{(k)}$ as a finite sum of time--ordered products and thus obtain
\beq\label{eq:regproof2}\bqmoller(F)=F + \bqmoller(C)\,,\eeq
where $C$ is a sum of terms of the form $\langle B_k, \Delta^{\otimes k}_\sharp F^{(k)}\rangle$ with  $\Delta_\sharp$ being either $\De^F_1$ or $\De^F_1 - \De^+_1$, and $k\ge 1$. Since $Q$ is quadratic, $B_k$ is a regular functional with values in $\mD_\bC(\mM^k)$ for all $k\ge 1$. Consequently, $C$ is a regular functional which is at least of first order in the functional derivatives of $Q$. The statement for $\bqmoller$ now follows from \eqref{eq:regproof2} by an induction over the order in perturbation theory.

To show the corresponding statement for $(\bqmoller)^{-1}$, we consider again an arbitrary $F\in\Freg$ of $n$--th order in $\phi$ and recall that $(\bqmoller)^{-1}$, provided it is well defined, is of the form
$$\left(\bqmoller\right)^{-1}(F)=S_{1,-Q}\T[1]\left(S_{1,Q}\star_1 F\right)$$
Using the results from the proof of the statement for $\bqmoller$, we may compute
$$\left(\bqmoller\right)^{-1}(F)=S_{1,-Q}\T[1]\left(S_{1,Q}\star_1 F\right) = F + S_{1,-Q}\T[1]\left(S_{1,Q}\star_1 F-S_{1,Q}\T[1] F\right) = F - C\,,$$
where $C$ is a functional of the type appearing in \eqref{eq:regproof2}; this concludes the proof.

Proof of (ii). This assertion follows by an induction over the order of perturbation theory from the form of $B_k$ in \eqref{eq:regproof1} and the fact that $Q$ is quadratic.

Proof of (iii). This statement follows immediately from (i).

Proof of (iv). $\bqmoller:\tAreg_{1,Q}\to\Areg_1$ is a $*$--isomorphism by construction. Moreover, by Theorem \ref{th:classical homomorphism}, $\qmoller:\Areg_2\to \Areg_1$ is a $*$--isomorphism as well. Consequently, $\bq\doteq \qmoller^{-1}\circ \bqmoller:\tAreg_{1,Q}\to\Areg_2$ is a $*$--isomorphism as a composition of two such morphisms.
\end{proof}

We would now like to compute the form of $\qstar$ on regular functionals explicitly. To this avail it is convenient to analyse the form of the commutator w.r.t. $\qstar$ among linear functionals. The following observation will prove to be useful in this respect.

\begin{proposition}\label{pr:beta-on-linear-fields}
The action of $\qmoller$ and $\bqmoller$ on linear functionals coincides in the sense of formal power series in $Q$ and its functional derivatives, in particular $\bq\doteq \qmoller^{-1}\circ \bqmoller$ is the identity on linear functionals.
\end{proposition}
\begin{proof}
Applying the Bogoliubov formula to an arbitrary linear functional $F_f$ we find
\begin{align*}
\bqmoller(F_f) &= S^{-1}_{1,Q} \star_1\left( S_{1,Q}\T[1] F_f \right)  \\
 &=S^{-1}_{1,Q} \star_1\left( S_{1,Q}\star_1 F_f \right) + 
S^{-1}_{1,Q} \star_1\left\langle \Delta^F_1 , S_{1,Q}^{(1)}\otimes f\right\rangle -
S^{-1}_{1,Q} \star_1\left\langle \Delta^+ , S_{1,Q}^{(1)}\otimes f\right\rangle   \\  
 &=F_f - S^{-1}_{1,Q} \star_1\left( S_{1,Q}\T[1]\left\langle \Delta^A,\q\otimes f\right\rangle\right)\,,
\end{align*}
where $S_{1,Q}^{(1)}=iS_{1,Q}\T[1]\q$ follows from the $\phi$--independence of the time--ordered product, see \eqref{eq:phi_independence_of_T}. However, $-\left\langle\Delta^A, \q\otimes f\right\rangle=F_{-{\q}\Delta_A f}$ is again a linear functional, where we recall that $\q$ defines a linear map on $\mE(M)$ which we denote by the same symbol.  Consequently, we obtain the Yang--Feldman type equation
\[
\bqmoller(F_f) = F_f + 
\bqmoller\left(F_{-{\q}\Delta_A f}\right).
\]
Applying this recursively, we have that 
\[
\bqmoller(F_f) = F_f +  F_{r^\dagger f} +  F_{r^\dagger r^\dagger f} +    \ldots 
\]
where $r^\dagger \doteq   -{\q}\Delta_A$. Recalling $F_f(\phi)=\int_M \phi f d\mu_{g_1}$ and the formal expansion of $\qmoller$ via \eqref{eq:Neumann series} we conclude that 
$\bqmoller(F_f)=\qmoller(F_f)$.
\end{proof}

Using the previous result in conjunction with Proposition \ref{pr:qmoellermapsreg} (iii), we may directly  compute the $\qstar$--commutator of two arbitrary linear functionals $F_f$ and $F_g$ as
\beq\label{eq:star_1q_commutator}
[F_f, F_g]_{\qstar} = \bq^{-1}\left([\bq(F_f),\bq(F_g)]_{\star_{2}}\right) =  i\hbar\bq^{-1}(\Delta_2(f,g))=i\hbar\Delta_2(f,g),
\eeq
where $\Delta_2(f,g)\doteq \langle\Delta_2, f\otimes g\rangle_1$ and where we recall Remark \ref{rem:operatordistributionconventions}. The next proposition ensures that the last observation together with the fact that, by Lemma \ref{pr:timeorderedprodof1Q}, the time--ordered product of regular functionals corresponding to $\qstar$ is $\T[1]$, already determine $\qstar$ uniquely. In this context, we note that, while the time--ordered product on regular functionals is uniquely determined by the corresponding $\star$--product, the inverse holds true only if one takes into account the involution on the algebra under consideration. Since the interacting involution $\qast$ differs in general from the free involution $*$, it is possible that the interacting $\star$--product $\qast$ and the free $\star$--product $\star_1$ have the same time--ordered product on regular functionals although they differ themselves.

\begin{proposition}\label{pr:structure-result} We recall Remark \ref{rem:operatordistributionconventions} and set
$\Delta^+_{1,Q}(f,g)\doteq \langle \Delta^+_{1,Q} , f \otimes g\rangle_1$, $\Delta^{\sharp}_i(f,g)\doteq \langle \Delta^\sharp_i,f\otimes  g\rangle_1$ for $\sharp\in\{+,-,R,A,F\}$ and $i=1,2$. The product $\qstar$ on $\Freg$ is uniquely determined by the following two conditions:
\begin{itemize}
\item[(i)] $[F_f, F_g]_{\qstar}=i\hbar\Delta_2(f,g)$, i.e. the antisymmetric part of the $\qstar$--product of two linear functionals $F_f, F_g$ is $\frac{i}{2}\hbar \Delta_2(f,g)$.
\item[(ii)] For arbitrary $F,G\in\Freg$ with $F\gtrsim G$, $F\T[1] G = F\qstar G$.
\end{itemize}
These two conditions imply that $\qstar$ is given by the usual exponential contraction formula \eqref{eq:exp-product explicit}, where the bi--distribution defining the contraction is the integral kernel w.r.t. $d\mu_{g_1}$ of the linear operator
\begin{equation}\label{eq:delta+1q}
\Delta^+_{1,Q} \doteq \De^F_1-i\De^A_2\cinv =\De^+_1 + i\Delta^A_1 - i \Delta^A_2\cinv=  \De^+_2\cinv + \De^F_1-\De^F_2\cinv.
\end{equation}
In particular, $$\Delta^+_{1,Q}(x,y) = \De^+_2(x,y) + \De^F_1(x,y) -  \De^F_2(x,y)\,.$$
\end{proposition}
\begin{proof}
First of all we recall that the product $\qstar$ is homomorphic to $\star_1$. Hence, since $\star_1$ is associative, $\qstar$ is associative as well. Moreover we recall that a regular functional of order $n$ in $\phi$ can be seen as a series of pointwise products of $n$ linear functionals, and that a pointwise products of $n$ linear functionals can be written as a $\star$--product of $n$ linear functionals plus regular functionals of lower order in $\phi$ if the $\star$--product is defined by an exponential contraction formula. Hence, it is sufficient to prove that the $\qstar$--product of $n$ linear functionals is of the form stated in the proposition.

To this avail, let us then consider two linear functionals $F_j\doteq F_{f_j},\ j=1,2$ defined as in \eqref{def:linear functionals} with $f_j\in\mD_\bC(\mM)$ and using the measure $d\mu_{g_1}$. If $f_1\gtrsim f_2$ we immediately get $F_1\qstar F_g\stackrel{(ii)}{=}F_1\T[1] F_2=F_1\cdot F_2+\hbar\De^F_1(f_1,f_2)$, while in the opposite case we can
use $(i)$ to get $F_1\qstar F_2=F_2\qstar F_1+i\hbar \Delta_2(f_1,f_2)=F_1\cdot F_2+\hbar(\De^F_1(f_1,f_2)+i\Delta_2(f_1,f_2))$.
Summing up we find
\begin{gather}
F_1\qstar F_2 =
F_1\cdot F_2+
\hbar\De^+_{1,Q}(f_1,f_2),
\end{gather}
where $\De^+_{1,Q}(f,g)= \De^F_1(f,g)$ if $f\gtrsim g$ and $\De^+_{1,Q}(f,g)=\De^F_1(f,g)+i\Delta_2(f,g)$ in the opposite case; this defines $\De^+_{1,Q}$ everywhere up to the diagonal. Since the Steinmann scaling degree towards the diagonal of this distribution is smaller then the spacetime dimension $4$ the extension to the diagonal is unique, see e.g. Theorem 5.2 in \cite{BF00} for further details. The result is 
\[
\De^+_{1,Q}(f,g) = \Delta^F_1(f,g)-i\De^A_2(f,g),\qquad \forall f, g \in \mD_\bC(\mM).
\]

Now we proceed by induction. To outline the idea, we consider in detail the case with three linear functionals $F_i$, $i=1,2,3$ $f_i\in\mD_\bC(\mM)$ as a simplified example.
By the previous case and by associativity of $\qstar$, we know that
\begin{gather}
F_1\qstar
(F_2\qstar
F_3)=
F_1\qstar (F_2\cdot F_3)+
\hbar F_1\De^+_{1,Q}(f_2,f_3).
\end{gather}
We now need more information on the $\qstar$-commutator of the first summand on the right hand side: again, by associativity and $(i)$, we get
\begin{gather*}
F_1\qstar F_2\qstar F_3=
F_2\qstar F_1\qstar F_3+
i\hbar F_3\Delta_2(f_1,f_2)\\=
(F_2\qstar F_3)\qstar F_1+
i\hbar F_3\Delta_2(f_1,f_2)+
i\hbar F_2\Delta_2(f_1,f_3),
\end{gather*}
hence, using associativity and $(i)$ we obtain the commutator
\[
[F_1,F_2\cdot F_3]_{\qstar}=
i\hbar F_3\Delta_2(f_1,f_2)+
i\hbar F_2\Delta_2(f_1,f_3).
\]

With the same argument as above, if $f_1\gtrsim f_j\ j=2,3$ we get $F_1\qstar (F_2\cdot F_3) = F_1\cdot F_2 \cdot F_3+\contraction{\hbar\De^F_1(f_1,}{f_2}{)}{F_3}\hbar\De^F_1(f_1,f_2)F_3$ -- where the contraction $\contraction{}{a}{_i}{b}a_ib_j= a_ib_j+a_jb_i$ indicates a symmetrisation --, while in the opposite case the previous results give $F_1\qstar (F_2\cdot F_3) = F_1\cdot F_2\cdot F_3+\contraction{\hbar\De^F_1(f_1,}{f_2}{)}{F_3}\hbar\De^F_1(f_1,f_2)F_3+i\hbar F_3\Delta_2(f_1,f_2)+
i\hbar F_2\Delta_2(f_1,f_3)$.
Combining all together  we have
\begin{gather}
F_1\qstar F_2 \qstar F_3 =F_1F_2F_3+
\contraction{\hbar\De^+_{1,Q}(f_1,}{f_2}{)}{F_3}\contraction{\hbar\De^+_{1,Q}(}{f_1}{,}{f_2}\hbar\De^+_{1,Q}(f_1,f_2)F_3.
\end{gather}

The generic case is done by induction: suppose that $F_1\qstar\ldots\qstar F_n$ is given by the formula we are interested in with the correct $\De^+_{1,Q}$ whenever $n\leq N$; we would like to show that then a suitable formula holds in the case $F_1\qstar\ldots\qstar F_{N+1}$. To this avail, we observe that, if the hypothesis is true for $n\leq N$ it follows that the commutator
\begin{gather}
[F_1,F_2\cdot \ldots\cdot F_{k}]_{\qstar}=
i\hbar\sum_{j=2}^{k}\Delta_2(f_1,f_j)F_2\cdot\ldots \cdot\widehat{F_j}\cdot\ldots\cdot F_{k}\qquad k=1,\ldots,N+1
\end{gather}
where $\widehat{F_j}$ indicates that $F_j$ has been removed from the pointwise product.
For the case $k\leq N$ this follows directly from the induction hypothesis, whereas in the case $k=N+1$, it can be shown by using the inductive hypothesis and the associativity of $\qstar$ in a similar manner as before.

Once we have all these commutators we just need to argue by means of causality as above. In particular, if $f_1\gtrsim f_j\ j=2,\ldots,N+1$ we get 
\[
F_1\qstar \left(F_2\cdot \ldots\cdot  F_{N+1}\right)=F_1\cdot \ldots\cdot F_{N+1}+\hbar\sum_{k=2}^{N+1}\Delta^F_{1}(f_1,f_k)F_2\cdot\ldots\cdot\widehat{F_k}\cdot\ldots\cdot F_{N+1},
\]
while otherwise we have 
\[
F_1\qstar \left(F_2\cdot \ldots\cdot  F_{N+1}\right)=F_1\cdot \ldots\cdot F_{N+1}+\hbar\sum_{k=2}^{N+1}(\Delta^F_{1}+i\Delta_2)(f_1,f_k)F_2\cdot\ldots\cdot\widehat{F_k}\cdot\ldots\cdot F_{N+1}.
\]
Combining the two results and considering the associativity of $\qstar$ then imply the statement for the case $N+1$.
\end{proof}

By using the previous proposition we may prove the powerful result that the isomorphism $\bq:\tAreg_{1,Q}\to \mA_2$ is a deformation, which directly implies that the PPA holds for regular functionals. 
\begin{theorem}\label{th:beta-homomorphism}
The following statements hold for the isomorphism $\bq=\qmoller^{-1}\circ \bqmoller:\tAreg_{1,Q}\to \mA_2$.
\begin{itemize}
\item[(i)] $\bq$ is a deformation, i.e. 
\begin{equation}\label{eq:d-deformation}
\bq=\alpha_{d_{1,Q}} ,\qquad    d_{1,Q}(x,y) \doteq\De^+_2(x,y)-\De^+_{1,Q}(x,y) = \Delta^F_{2}(x,y)-\Delta^F_{1}(x,y)
\end{equation}
where $\alpha_{d_{1,Q}}$ is defined as in \eqref{def:alpha} and we recall Remark \ref{rem:operatordistributionconventions}.
\item[(ii)] $\bq$ intertwines the time--ordered products $\T[1]$ and $\T[2]$ on $\Freg$, i.e. for all $n$ and arbitrary $F_1,\ldots,F_n\in\Freg$,
\[
\bq\left(\bq^{-1}(F_1)\T[1]\,\ldots\,\T[1] \bq^{-1}(F_n)\right) = F_1\T[2]\,\ldots\,\T[2]F_n\,.
\]
\item[(iii)] For an arbitrary choice of renormalisation freedom, the time--ordered map satisfies the Principle of Perturbative Agreement in the sense of Definition \ref{def:PPA} for multilinear functionals, i.e. 
$$T_2(F_1,\ldots,F_n)=[\bq\circ T_1](F_1,\ldots,F_n)$$
for all linear $F_1,\ldots,F_n$.
\end{itemize}

\end{theorem}
\begin{proof}Proof of (i). The proof of this statement is an immediate consequence of the structure result on $\qstar$ found in Proposition \ref{pr:structure-result} and the fact that $\bq$ is the identity on linear functionals, cf. Proposition \ref{pr:beta-on-linear-fields}. As before, we may argue that it is sufficient to prove that $\bq$ is of the stated form for arbitrary pointwise products of linear functionals. To show the latter for the pointwise products of two arbitrary functionals $F_f,F_g$, we may compute
\[
\bq\left(F_f \cdot F_g\right)=\bq\left(F_f \qstar F_g\right)-\hbar \,\Delta^+_{1,Q}(f,g) = F_f\star_2 F_g -\hbar \,\Delta^+_{1,Q}(f,g) = F_f \cdot F_g +\hbar  \,d_{1,Q}(f,g)\,.
\]
Analogously one may show $\bq= \alpha_{d_{1,Q}}$ for arbitrary higher order pointwise products of linear functionals.

Proof of (ii). This statement follows immediately from (i) and $d_{1,Q}(x,y)=\De^F_2(x,y) - \De^F_1(x,y)$.

Proof of (iii). This follows directly from (ii) and Proposition \eqref{pr:qmoellermapsreg} (ii) as well as from $T_i(F_1,\ldots,F_n)=F_1\T[i]\ldots\T[i] F_n$ for all linear $F_1,\ldots,F_n$.
\end{proof}

Note that the maps $\qmoller:\mA_2\to \mA_1$ and $\bqmoller:\tAreg_{1,Q}\to\mA_2$ fail to intertwine time--ordered products although they are homomorphisms with respect to the $\star$--products. This is related to the failure of these maps to preserve causality relations among supports, but this failure cancels precisely in the combination $\bq=\qmoller^{-1}\circ \bqmoller$.

\subsection{Extension of the perturbative agreement to non--linear local functionals}
\label{sec:extensionnonlinear}

In this section we discuss the PPA for general non--linear local functionals and prove it for the case of quadratic $Q$ without derivatives, i.e. for 
\beq\label{def:puremassQ}Q=\frac12 \int_\mM M \phi^2 d\mu_{g_1}\,,\qquad  M\in\mD_\bR(\mM)\text{ arbitrary}.\eeq
 As argued in Section \ref{sec:formulation_PPA}, in this case the PPA can not hold in the strong sense that $\bq=\qmoller^{-1}\circ \bqmoller$ is a $*$--isomorphism between the algebra $\widetilde{\mA}_{1,Q}$, $\qstar$--generated by $F\in\Ftloc[1]$ and the subalgebra $\A0_2\subset\mA_2$ (cf. \eqref{def:A0}), because $\qstar$ is not well--defined on non--linear local functionals and thus $\widetilde{\mA}_{1,Q}$ is ill--defined from the outset. Notwithstanding, we know that  $\bqmoller$ is well--defined on $\Ftloc[1]$, and that $\qmoller^{-1}:\Fmuc[1]\to\Fmuc[1]$ is well--defined by Theorem \ref{th:classical homomorphism}. Consequently, $\bq=\qmoller^{-1}\circ \bqmoller$ is well--defined on $\Ftloc[1]$ and there is no obvious obstacle for satisfying the PPA in the weaker sense of Definition \ref{def:PPA}.

\begin{remark}\label{rem:HWproof} The following strategy is employed in \cite{HW05} in order  to prove the PPA for scalar fields and for $Q$ encoding metric changes (see also \cite{Zahn} for a similar proof for the case of Dirac fields in the presence of a classical gauge field and complementary details). First it is argued that the PPA $T_2 = \bq \circ T_1$ is satisfied if an only if it holds at the linearised level for all ``1--backgrounds". The linearised PPA is then proven by an induction over the total number of field factors in the arguments of the time--ordered map, by showing that at each induction step it is possible to redefine the time--ordered map $T$ in a way compatible with both the linearised PPA and the remaining axioms for $T$. Thereby the conservation of the free stress--energy tensor, i.e. $[T(\int f^b\nabla^aT_{ab}(\phi)d\mu_g )]=[0]\in\Aon$ for all compactly supported $f^a$ and with $T_{ab}(\phi)$ being the canonical stress--energy tensor plays an important role and is a necessary condition for the validity of the PPA.

In order to reabsorb the ``error term'' in the linearised PPA into a redefinition of $T$ one has to check that this term has the correct symmetry properties. In \cite[Section 6.2.6.]{HW05} it is argued that this is the case if the free stress--energy tensor is conserved. However, although conservation of the free stress--energy tensor holds only on--shell, the error term discussed in \cite[Section 6.2.6.]{HW05} is a \emph{constant} functional and thus vanishes on--shell if and only if it vanishes off--shell. Consequently, the proof of the PPA given in \cite{HW05} can be seen to hold also off--shell.

The proof strategy of \cite{HW05} outlined above can be used in order to prove the PPA also for $Q$ of the form \eqref{def:puremassQ}. Thereby the symmetry property discussed in \cite[Section 6.2.6.]{HW05} automatically holds (in spacetime dimensions $d\le 4$) due to the fact that $\phi^2$ has a lower ``engineering dimension'' than $T_{ab}$\footnote{We would like to thank Jochen Zahn for pointing this out to us.}. 
\end{remark}

Notwithstanding Remark \ref{rem:HWproof}, we develop in this section an alternative strategy to prove the PPA which is closer to the spirit of Section \ref{sec:perturbativerelations} and uses the results obtained there. To this end, we examine the precise action of $\bq$ on general local functionals. As we have seen in the previous section, $\bq=\alpha_{d_{1,Q}}$ on regular functionals where $d_{1,Q}(x,y)=\De^F_2(x,y)-\De^F_1(x,y)$. If we consider this as an exact expression, we observe that a direct extension of $\alpha_{d_{1,Q}}$ to non--linear local functionals by a limiting procedure is ill--defined, because the singular structures of $\De^F_2$ and $\De^F_1$ differ and thus the coinciding point limit of $d_{1,Q}(x,y)$ is ill--defined. In more detail, the integral kernels of the exact Feynman propagators have locally the form \eqref{def:Hamardplus} up to replacing $\sigma_i^{\epsilon+}$ by $\sigma_i+i\epsilon$. Thus, if $Q$ does not contain second derivatives, i.e. $g_1=g_2$, the coinciding point limit of the integral kernel of $d_{1,Q}$ is logarithmically divergent, whereas in the general case, the divergence is quadratic. However, we have to view $\De^F_2$ as a formal power series in $Q$ and its functional derivatives, because this is the only setting in which $\bq$ makes sense anyway. In the very same manner, $\bq=\alpha_{d_{1,Q}}$ may be extended to non--linear local functionals without problems, as $\bq=\qmoller^{-1}\circ \bqmoller$ is well--defined on $\Ftloc[1]$ and thus in particular on $\Floc$. Hereby, the renormalisation of the $\T[1]$--products appearing in $\bqmoller$ may be understood as effectively removing the divergencies in the coinciding point limit of $d_{1,Q}$ perturbatively.

\begin{proposition}\label{pr:betalocal}
The following statements hold for $\bq\doteq \qmoller^{-1}\circ\bqmoller:\Ftloc[1]\to \Fmuc[1]$.
\begin{itemize}
\item[(i)] To all orders in perturbation theory, the action of $\bq\doteq \qmoller^{-1}\circ\bqmoller$ on $\Ftloc[1]$ is given by the deformation $\bq=\alpha_{d_{1,Q}}$ where $d_{1,Q}(x,y)=\De^F_2(x,y)-\De^F_1(x,y)$ is understood as a formal power series in $Q$ and its functional derivatives, and for an arbitrary $F\in\Ftloc[1]$, all expressions of $\De^F_1$ in the formal expansion of $\alpha_{d_{1,Q}}(F)$ in terms of $\De^F_1$ and $\De^+_1$ are understood as being implicitly renormalised by the renormalisation of $\T[1]$.  Moreover, $\bq^{-1}$ is well--defined on $\Floc$ and is of the form $\bq^{-1}=\alpha_{-d_{1,Q}}$, understood in the sense mentioned above.
\item[(ii)] $\bq$ and $\bq^{-1}$ map local functionals to formal power series in $Q$ and its functional derivatives with values in $\Floc$.
\item[(iii)] $\bq$ is $\phi$--independent, i.e. $\bq(F)^{(1)}=\bq\left(F^{(1)}\right)$.
\end{itemize}
\end{proposition}
\begin{proof}Proof of (i). First of all, we recall that $\bq= \qmoller^{-1}\circ\bqmoller$ is well--defined on $\Ftloc[1]$ as argued in the preceding paragraph, and that $\bq=\alpha_{d_{1,Q}}$ on $\Freg$.

We discuss the remainder of the statement for the special case of a quadratic local functional without derivatives, the general case follows by analogous arguments. To this avail, we note that a generic quadratic local functional $F$ without derivatives may be written as  
\beq\label{eq:beta_local_proof}
\int h(x) \phi^2(x)   d\mu_{g_1}(x) = \sum_{n=1}^\infty  \left( F_{f_n}\T[1] F_{g_n} - \hbar \,\Delta_1^{F}(f_n,g_n) \right)
\eeq
whenever the series of $\sum_n \left(f_n\otimes  g_n+g_n\otimes  f_n\right)/2$ converges to $h(x)\delta(x,y)$ in the H\"ormander topology, see e.g. \cite{BF00,FredenhagenRejzner2}. The application of $\bq=\alpha_{d_{1,Q}}$ with $d_{1,Q}=\De^F_2-\De^F_1$ to each summand in \eqref{eq:beta_local_proof} is well--defined if we use Proposition \ref{pr: delta intertwined} and our assumptions on $\De^+_2$ in order to consider $\De^F_2$ as the exact expression (in the sense of integral kernels, cf. Remark \ref{rem:operatordistributionconventions})
\beq\label{eq:beta_local_proof2}
\De^F_2=\De^+_2 + i\De^A_2 = \qcmoller \circ \De^+_1 \circ \qcmoller^{\dagger} + i \De^A_1\circ \qcmoller^{\dagger}
\eeq
with $\qcmoller = (\Id + \De^R_1\circ \q )^{-1}$. However, this also holds if we expand $\qcmoller$ in the formal Neumann series \eqref{eq:Neumann series}. Taking the latter point of view, we can express all appearing advanced and retarded propagators $\De^{R/A}_1$ in terms of $\De^F_1$ and $\De^+_1$. If we now switch the order of the sums in \eqref{eq:beta_local_proof} and in the contraction map $\alpha_{d_{1,Q}}$, recalling that the latter is exact since \eqref{def:alpha} contains at most two summands on account of the fact that $F$ is quadratic, we encounter in the limit expressions in $\De^F_1$ which are a priori ill--defined distributions. Yet, these expressions are replaced by well--defined distributions in the construction of $\T[1]$ as a renormalised time--ordered product on $\Floc$ (see e.g. \cite{BF00, HW02,BDF,FredenhagenRejzner,FredenhagenRejzner2} and Remark \ref{rem:exactPPA}. We recall that this renormalisation of $\T[1]$ is necessary in order to make the quantum M\o ller map $\bqmoller$, and thus $\bq\doteq \qmoller^{-1}\circ\bqmoller$, well--defined on $\Floc$ in the first place. Note that no such regularisation is necessary in order to have $\qmoller^{-1}$ well--defined, i.e. $\qmoller^{-1}$ ``does not introduce new loops'' in $\bq\doteq \qmoller^{-1}\circ\bqmoller$. For definiteness, one may think of replacing $\De^F_1$ in the formal expansion of $\bq(F)=\alpha_{d_{1,Q}}(F)$ by a regularised version $\De^{F,\lambda}_{1}$ which depends meromorphically on $\lambda$ and equals $\De^F_1$ for $\lambda=0$, cf. \cite{Hollands:2010pr, Keller:2010xq,Duetsch:2013xca, GHP}. By doing so and choosing different $\lambda$ for each individual $\De^{F,\lambda}_{1}$, one obtains that, at each order in perturbation theory, $\bq(F)=\alpha_{d_{1,Q}}(F)$ is a sum of terms which are meromorphic in $(\lambda_1,\ldots,\lambda_N)$ for a suitable $N$. The renormalisation of $\T[1]$ then consists of removing the poles of this meromorphic expression and taking the limit of $\lambda_i\to0$ in a particular order (ibid.), thus demonstrating that the form of $\bq=\alpha_{d_{1,Q}}$ is preserved by the extension to general local functionals.

Proof of (ii). This statement follows immediately from (i).

Proof of (iii). This statement follows immediately from (i) or alternatively by using directly the definition of $\bq$ as follows. Using $\left(S^{-1}_{1,Q}\right)^{(1)}=-S^{-1}_{1,Q}\star_1 S^{(1)}_{1,Q}\star_1 S^{-1}_{1,Q}$ and $\phi$--independence of $\T[1]$, we may compute for an arbitrary $F\in\Ftloc[1]$
\beq\label{eq:beta_local_proof3}
\bqmoller(F)^{(1)}=- \bqmoller\left(\q\right)\star_1 \bqmoller(F) + \bqmoller\left(\q\T[1] F\right) + \bqmoller\left(F^{(1)}\right).
\eeq
On the other hand, Definition \ref{def:classicalmoller} and \eqref{pr:rq-inverse} imply for an arbitrary $F\in\Fmuc[1]$
\beq\label{eq:beta_local_proof4}
\qmoller^{-1}(F)^{(1)} = \qmoller^{-1}\left(F^{(1)}\right)+\De^A_1 \q\qmoller^{-1}\left(F^{(1)}\right).
\eeq
We now observe and recall the following facts: (a) $\bq$ is the identity on linear functionals by Proposition \ref{pr:beta-on-linear-fields}, (b) $Q$ is quadratic, (c) $\star_2$ is related to $\star_1$ via $F\star_2 G=\qmoller^{-1}\left(\qmoller(F)\star_1\qmoller(G)\right)$ for all $F,G\in\Fmuc[2]$, cf. Theorem \ref{th:classical homomorphism}, (d) the time--ordered product $\T[2]$ corresponding to $\star_2$ is uniquely defined by \eqref{def:t2} if one factor is a regular functional, (e) in this case $\bq$ intertwines $\T[2]$ and $\T[1]$ since, for arbitrary $F\in\Freg$ and $G\in\Ftloc[1]$ with $F\gtrsim G$,
\begin{align}
 \bq\left(\bq^{-1}(F)\T[1]\bq^{-1}(G)\right)& = \qmoller^{-1}\left(\molh{1}{Q}\left(\bq^{-1}(F)\T[1]\bq^{-1}(G)\right)\right)\notag\\
&= \qmoller^{-1}\left(\molh{1}{Q}\left(\bq^{-1}(F)\right)\star_1 \molh{1}{Q}\left(\bq^{-1}(G)\right)\right)\label{eq:beta_local_proof5}\\
&=\qmoller^{-1}\left(\qmoller(F)\star_1 \qmoller(G)\right)=F\star_2 G\notag\,.
\end{align}
Here we used Lemma \ref{pr:timeorderedprodof1Q}, Theorem \ref{th:classical homomorphism}, and $\molh{1}{Q}\circ \bq^{-1}=\qmoller$, which can be proved by applying $\qmoller^{-1}$ to both sides. Moreover, we used $g_2>g_1$, which implies that if $F\gtrsim G$ in the sense of $g_2$, then $F\gtrsim G$ also in the sense of $g_1$, because every Cauchy surface for $g_2$ is a Cauchy surface for $g_1$. Using now \eqref{eq:beta_local_proof3}, \eqref{eq:beta_local_proof4}, and (a)-(e), we obtain
$$
\bq(F)^{(1)}=\left(\Id+\De^A_1\circ\q\right)\left(\bq\left(F^{(1)}\right)-\De^A_2\q\bq(F)^{(1)}\right)\,.
$$
The statement then follows from the last identity by an induction over the order of perturbation theory.
\end{proof}

Motivated by the previous result and the observation that $\bq$ intertwines $\T[1]$ and $\T[2]$ on $\Freg$, we now construct a time--ordered map $T=T(g,M)$ which satisfies the PPA w.r.t. changes of $M$ as follows (we do not spell out the dependence of $T$ on $\star$ in this paragraph). We set $T_1 = T(g,0)$ where $g$ is arbitrary and where we assume that $T(g,0)$ satisfies all axioms reviewed in Section \ref{sec:interactingobservables} including the PPA for changes of the metric. Based on this, we define the time--ordered map $T_2=T(g,M)$ for arbitrary $M\in\mD_\bR(\mM)$ by setting $T_2\doteq \bq\circ T_1$, where $Q$ is given by \eqref{def:puremassQ}. We then first prove that this definition of $T(g,M)$ satisfies all axioms but the PPA w.r.t to changes of $M$. This last property is then seen to follow from a cocycle condition for $\bq$ which holds by our definition of $T(g,M)$. This construction of $T(g,M)$ is by its very nature perturbative in $M$. While this is sufficient in the context of the PPA, we shall argue that also outside of this context the given construction of $T(g,M)$ is exact in the sense of fixing the $M$--dependent renormalisation freedom of $T$ because for any given multilocal functional this freedom is a polynomial of finite order in $M$ whose coefficients are themselves determined by the renormalisation freedom of finitely many graphs in the theory with $M=0$.

\begin{proposition}\label{pr:t2unique}Let us assume the following.
\begin{itemize}
\item[(i)] Let $\cS_1$ be an arbitrary quadratic action of the form \eqref{def:quadraticactions} with $M_1=0$, $Q$ of the form \eqref{def:puremassQ} and $\cS_2=\cS_1 + Q$.
\item[(ii)] Let $\bqmoller$ be the quantum M\o ller map defined in \eqref{def:quantum moller operator}, let $\qmoller$ be the classical M\o ller map defined in \eqref{eq:definition r}, and set $\bq\doteq \qmoller^{-1}\circ \bqmoller$.
\item[(iii)] Let $\star_1$ be a $\star$--product corresponding to $\cS_1$ and let $\star_2$ be the $\star$--product induced by $\star_1$ via $\qmoller$ as $F\star_2 G\doteq \qmoller^{-1}\left(\qmoller(F)\star_1\qmoller(G)\right)$ for arbitrary $F,G\in\Fmuc[1]=\Fmuc[2]$.
\item[(iv)] Let $T_1 = T(g_1,0,j_1,\star_1)$ where $T(g,0,j,\star)$ satisfies all axioms reviewed in Section \ref{sec:interactingobservables} except for those pertaining to $M$.
\end{itemize}
Then $T_2$, defined as 
\beq\label{def:t2loc}
T_2\doteq \bq \circ T_1
\eeq
and considered as $T_2=T(g_2=g_1,M,j_2=j_1,\star_2)$ satisfies, in the perturbative sense, all axioms for time--ordered maps reviewed in Section \ref{sec:interactingobservables}, but the PPA for changes of $M$.
\end{proposition}
\begin{proof}
We first note that $T_2$ is well--defined because $\bq$ maps $\Floc$ to itself by Proposition \ref{pr:betalocal} and because $\molh{1}{Q}$, and thus also $\bq$, are well--defined on $\Ftloc[1]$. Consequently, $\bq^{-1}$ is a well--defined map from $\Ftloc[2]$ to $\Ftloc[1]$. In order to demonstrate that $T_2$ is a time--ordered map for $\star_2$, we note that $T_2$ is symmetric because $T_1$ has this property. A computation analogous to \eqref{eq:beta_local_proof5} implies the causal factorisation property of $T_2$ w.r.t. $\star_2$.

 It is not difficult to check that $T_2$ satisfies the other axioms of  time--ordered maps reviewed in Section \ref{sec:interactingobservables}. In particular, $\phi$--independence of $T_2$ follows from the same property of $T_1$ and $\phi$--independence of $\bq$, cf. Proposition \ref{pr:betalocal} (iii), whereas the Leibniz rule for $T_2$ follows from the same property of $T_1$ and the fact that, if $F$ is of the form $F(\phi)=\int_\mM d B(\phi)$ for a three--form $B(\phi)$, then $\bq(F)$ and $\bq^{-1}(F)$ are of the same form, since $\bq$ and its inverse are given by a contraction exponential.

Moreover, unitarity of $T_2$ ensues from unitarity of $T_1$ as follows. Expanding the right hand side of $T_2(F_1,\ldots,F_n)=[\bq\circ T_1](F_1,\ldots,F_n)$ for arbitrary $n$ and arbitrary $F_1,\ldots,F_n\in\Floc$ perturbatively, one obtains an expression which, at each order in $\hbar$, equals a $\hbar$--order of $\star_1$--products of $\T[1]$--products of $T_1(F_i)$, $i=1,\ldots,n$ and $Q$. This observation, together with the fact that a) $T_1$ satisfies unitarity, b) unitarity holds if and only if it holds at each order in $\hbar$ and c) the involution on $\mA_1$ and $\mA_2$ is the same because $\qmoller$ commutes with complex conjugation, imply unitarity of $T_2$.

In order to check the microlocal spectrum condition, i.e. that $T_2(F_1,\ldots,F_n)|_{\phi=0}$, viewed as a distribution evaluated on the test--sections present in $F_i,i=1,\ldots,n$, has the wave front set reviewed in Axiom 4 of Section \ref{sec:interactingobservables}, we may expand this distribution perturbatively in terms of Feynman graphs $\Gamma$ with $\Delta_1^F$ and $\Delta^\pm_1$ propagators. The distributions $u_\Gamma$ corresponding to the integral kernel of each of these graphs have a wave front set of the wanted form, cf. \cite{BFK,BF00}. In order to obtain the distribution corresponding to $T_2(F_1,\ldots,F_n)|_{\phi=0}$, we have to integrate the ``inner vertices'' of the $u_\Gamma$ against the smooth and compactly supported function $M$. The wave front set of the resulting distribution may be seen to be of the correct type by an application of \cite[Theorem 8.2.13]{Hormander}.

The continuous respectively analytic dependence of $T_2$ on $g_2=g_1$ follows directly from the corresponding property of $T_1$ by again expanding $T_2$ in terms of ``1"--quantities and observing that each term in this expansion has this property. The fact that $T_2$ depends in a local and covariant fashion on the background fields $(g_2=g_1,M,j_2=j_1)$ may be seen as resulting from the local and covariant dependence of $T_1$ on the background fields $(g_1,M=0,j_1)$ as follows. In an arbitrary but fixed geodesically convex neighbourhood of $(\mM,g_2=g_1)$ we define $W_i$ to be the smooth parts in the local Hadamard expansion \eqref{def:Hamardplus} of $\De^F_i$ (up to replacing $\sigma_i^{\epsilon+}$ by $\sigma_i+i\epsilon$, and considering the same mass scale $\lambda$ in the logarithmic term for definiteness). We then define $\widetilde T_i\doteq \alpha_{-W_i}\circ T_i$ and $\wbq\doteq \alpha_{-W_2}\circ \bq \circ \alpha_{W_1}$ with $\alpha$ being defined as in \eqref{def:alpha}. By construction we have $\widetilde T_2 = \wbq \circ \widetilde T_1$ and $\wbq = \alpha_{H^F_2-H^F_1}$, where $H^F_i \doteq \Delta^F_i-W_i$ is the ``geometric part'' of $\Delta^F_i$ and this form of $\bq$ holds up to renormalisation of $\widetilde T_1$ in the sense of Proposition \ref{pr:betalocal}. However, as this renormalisation is done in a local and covariant way by our assumptions on $T_1$, we see that $\wbq$ preserves the local and covariant dependence of $\widetilde T_1$ on $g_1$ and $j_1$ and that the dependence of $\widetilde T_2 = \wbq \circ \widetilde T_1$ on $M$ is local and covariant as well. This also implies that $\widetilde T_2$ (and thus $T_2$) has the correct scaling behaviour w.r.t. to constant rescalings of the background fields and the correct analytic dependence on $M$.

Finally, the fact that $T_2$ satisfies the PPA w.r.t. to changes of the metric follows from the corresponding property of $T_1$ by an argument similar to the one used in the proof of the following Lemma \ref{pr:cocylebeta}.
\end{proof}

In order to prove that the time--ordered map $T(g,M,j,\star)$ defined as in the previous proposition also satisfies the PPA with respect to changes of $M$, we need the following cocycle condition for $\bq$.

\begin{lemma}\label{pr:cocylebeta}
Let $M_i, i=2,3$ be arbitrary elements of $\mD_\bR(\mM)$, set $Q_i\doteq \frac12 \int_\mM M_i \phi^2 d\mu_{g_1}, i=2,3$, $\delta Q\doteq Q_3-Q_2$ and define $T_i = T(g_1,M_i,j_1,\star_i)\doteq \beta_{1,Q_i}\circ T_1$ as in Proposition \ref{pr:t2unique}. Then
\beq\label{eq:cocyle}\beta_{1,Q_3} = \beta_{2,\delta Q}\circ \beta_{1,Q_2}\,.\eeq
\end{lemma}
\begin{proof}By Proposition \ref{pr:betalocal}, we know that this identity is satisfied up to renormalisation of the time--ordered product. However, by our definition of $T_i, i=2,3$, their renormalisation is uniquely fixed by expanding them in terms of renormalised ``1''--quantities. In other words, if we apply the left hand side of \eqref{eq:cocyle} to an arbitrary $F\in\Ftloc[1]$, we may expand the result in terms of Feynman graphs with $\Delta^F_1$ and $\Delta^\pm_1$ propagators. If we instead apply the right hand side of  \eqref{eq:cocyle} to the same $F$, we may first expand the result in terms of $\Delta^F_1$ and $\Delta^\pm_1$ and $\Delta^F_2$ and $\Delta^\pm_2$, and subsequently expand $\Delta^F_2$ and $\Delta^\pm_2$ in terms of $\Delta^F_1$ and $\Delta^\pm_1$ themselves. By the exponential form of all appearing $\beta$, we know that the combinatorics of this (iterated) expansion is such that we obtain the the same (renormalised) Feynman graphs in terms of $\Delta^F_1$ and $\Delta^\pm_1$ for both sides of \eqref{eq:cocyle}.
 
The statement can also be proven by a direct computation which encodes the above line of argument. To this end, we record a few basic identities.
$$ S_{1,V+W}=S_{1,V}\T[1]S_{1,W}$$
$$ \beta_{1,Q_2}\left(S_{1,V}\right)=S_{{2},\beta_{1,Q_2}(V)}\qquad \Leftrightarrow\qquad  \beta^{-1}_{1,Q_2}\left(S_{2,V}\right)=S_{{1},\beta^{-1}_{1,Q_2}(V)} $$
$$\mR_{1,Q_2}\left(S^{-1}_{2,V}\right) =\left(\mR_{1,Q_2}\left(S_{2,V}\right)\right)^{-1}= \left(\mR^\hbar_{1,Q_2}\left(S_{1,\beta_{1,Q_2}^{-1}(V)}\right)\right)^{-1} $$
$$\mR_{1,Q_3}=\mR_{1,Q_2}\circ\mR_{2,\delta Q}\qquad\Leftrightarrow\qquad \mR_{2,\delta Q}^{-1}=\mR_{1,Q_3}^{-1}\circ  \mR_{1,Q_2}$$
The first one holds for arbitrary $V$, $W\in\Floc$ because $S_{1,V}$ is an exponential w.r.t. a symmetric product, the second follows for arbitrary $V\in\Floc$ from $T_2=\beta_{1,Q_2}\circ T_1$ and the third one holds again for arbitrary $V\in\Floc$ because of the second and $\mR_{1,Q_2}(1)= 1$. Finally, the last identity follows from the defining properties of the classical M\o ller map. We omit the meaning of the $\star$-inverse of $S$--matrices, the canonical inverse is implied. Using these identities, we find for an arbitrary $F\in\Ftloc[1]$ with $G\doteq \beta_{1,Q_2}(F)$
\begin{align*}
&\left[\beta_{2,\delta Q}\circ \beta_{1,Q_2}\right](F)\\
=&\left[\mR_{2,\delta Q}^{-1}\circ \mR_{2,\delta Q}^{\hbar}\right](G)=\left[\mR_{1,Q_3}^{-1}\circ \mR_{1,Q_2}\right]\left(S^{-1}_{2,\delta Q}\star_2 \left(S_{2,\delta Q}\T[2] G\right)\right)\\
=&\mR_{1,Q_3}^{-1}\left(\mR_{1,Q_2}\left(S^{-1}_{2,\delta Q}\right)\star_1 \mR_{1,Q_2}\left(S_{2,\delta Q}\T[2] G\right)\right)\\
=&\mR_{1,Q_3}^{-1}\left(\left(\mR^\hbar_{1,Q_2}\left(S_{1,\beta_{1,Q_2}^{-1}(\delta Q)}\right)\right)^{-1}\star_1 \mR^\hbar_{1,Q_2}\left(S_{1,\beta_{1,Q_2}^{-1}(\delta Q)}\T[1] F\right)\right)\\
=&\mR_{1,Q_3}^{-1}\left(\left(S_{1,Q_2}\T[1]S_{1,\beta_{1,Q_2}^{-1}(\delta Q)}\right)^{-1}\star_1 S_{1,Q_2}\star_1 S_{1,Q_2}^{-1}\star_1 \left(S_{1,Q_2}\T[1]S_{1,\beta_{1,Q_2}^{-1}(\delta Q)}\T[1] F\right)\right)\\
=&\left[\mR_{1,Q_3}^{-1}\circ \mR^\hbar_{1,Q_2+\beta_{1,Q_2}^{-1}(\delta Q)}\right](F)=\beta_{1,Q_3}(F)\,.
\end{align*}
In the last step, we used that $\mR^\hbar_{1,V+C}=\mR^\hbar_{1,V}$ for any $V\in\Floc$ and any constant functional $C$ and that $\beta_{1,Q_2}^{-1}(\delta Q)-\delta Q$ is a constant functional, because $\delta Q$ is quadratic and $\beta_{1,Q_2}$ is $\phi$--independent.
\end{proof}

We may now combine Proposition \ref{pr:t2unique} and Lemma \ref{pr:cocylebeta} in order to obtain the wanted result.

\begin{theorem}\label{th:PPA} Under the assumptions and using the notation of Proposition \ref{pr:t2unique} and Lemma \ref{pr:cocylebeta}, let $T_1=T(g_1,0,j_1,\star_1)$ be an arbitrary but fixed time--ordered map which satisfies all axioms reviewed in Section \ref{sec:interactingobservables} except for those pertaining to $M$, and let $Q_i,i=2,3$ be of the form \eqref{def:puremassQ} for arbitrary $M_i\in\mD_\bR(\mM),i=2,3$. Then $T_2$, defined as $T_2\doteq \bq \circ T_1$ 
and considered as $T_2=T(g_2=g_1,M_2,j_2=j_1,\star_2)$ satisfies, in the perturbative sense, all axioms for time--ordered maps reviewed in Section \ref{sec:interactingobservables}. In particular, $T_3 \doteq \beta_{1,Q_3}\circ T_1$ satisfies
$$T_3 = \beta_{2,\delta Q}\circ T_2\,.$$
\end{theorem}

\begin{remark}\label{rem:exactPPA}
Our construction of a time--ordered map satisfying the PPA for quadratic $Q$ without derivatives is defined directly in terms of perturbative quantities which seems unsatisfactory at first glance. However, as anticipated, we shall now argue that the construction is exact regarding the determination of the renormalisation freedom of $T_M=T(g,M,j,\star)$ in the sense that for arbitrary but fixed $F_1,\ldots,F_n\in\Floc$, the $M$--dependent renormalisation freedom of $T_M(F_1,\ldots,F_n)$ is a polynomial in $M$ (and its derivatives) of finite order whose coefficients are fixed by fixing the renormalisation freedom of finitely many expressions in the theory with $M=0$. We shall illustrate this at the example of a quadratic local functional.

We start by analysing in more detail the perturbative structure of $\bq=\alpha_{d_{1,Q}}$, $d_{1,Q}(x,y)=\Delta^F_2(x,y)-\Delta^F_1(x,y)$, cf. Proposition \ref{pr:betalocal}. To this avail, we combine \eqref{eq:beta_local_proof2} and \eqref{eq:Neumann series} in order to realise that, at $n$--th order in perturbation theory with $n>0$, $d_{1,Q}$ equals $d_{1,Q,n}$ given by  
\begin{align}
d_{1,Q,n}&=
\sum_{p=0}^{n}
r^{p}\Delta^+_1 (r^\dagger)^{n-p}+
i\Delta_1^A(r^\dagger)^{n} \notag\\
&= (-1)^n\left(\sum_{p=0}^{n}
\left(\Delta_1^R \q\right)^{p} \Delta_1^+\left(\q\Delta_1^A\right)^{n-p}+
i\Delta_1^A\left(\q\Delta^A_1\right)^{n}\right)\label{eq:d-expansion}\\
&= i^n\left(\sum_{p=1}^{n}
\left(\left(\Delta_1^F-\De^-_1\right) \q\right)^{p} \Delta_1^+\left(\q\left(\Delta_1^F-\De^+_1\right)\right)^{n-p}+
\Delta_1^F\left(\q\left(\Delta_1^F-\De^+_1\right)\right)^{n}\right)\notag
\end{align}
where the appearing products and exponents indicate iterated compositions and not pointwise products of distribution kernels and where we recall that $\q$ induces a formally selfadjoint linear map on $\mE(\mM)$ which we denote by the same symbol. At zeroth order in perturbation theory $d_{1,Q,0}=0$.

We would now like to analyse the regularity of this expression. While this may be deducted from structural results in perturbative QFT on curved spacetimes, see e.g. \cite{BF00, HW02,BDF,FredenhagenRejzner,FredenhagenRejzner2}, it is instructive to derive it in detail. To this end, we make a preliminary observation. Consider two distributions $A,B\in\mD_\bC^{\prime}(\mM^2)$ whose wave front set is contained in $\WF(\De^{+/-/F}_1)$. Whenever their composition is done on a set of compact support, which is the case at hand because $\supp(Q)$ is compact, we can use \cite[Theorem 8.2.14]{Hormander} and the form of $\WF(\De^{+/-/F}_1)$ in order to realise that their composition is well--defined and has the following wave front set.
\begin{align*}
\WF(A\circ B) \subset \WF(\Delta_1^\pm)\,, &\qquad    \text{if } \WF(A),\WF(B) \subset \WF(\Delta_1^{\pm}), 
\\
\WF(A\circ B) \subset \WF(\Delta_1^F)\,, &\qquad    \text{if } \WF(A),\WF(B) \subset \WF(\Delta_1^{F}), 
\\
\WF(A\circ B) \subset \WF(\Delta_1^\pm)\,, &\qquad    \text{if } \WF(A) \subset \WF(\Delta_1^{F}), \WF(B) \subset \WF(\Delta_1^{\pm})  
\\
\WF(A\circ B) = \emptyset  \,,&\qquad \text{if } \WF(A) \subset \WF(\Delta_1^{-}), \WF(B) \subset \WF(\Delta_1^{+}) 
\end{align*}
With this in mind we observe that all contributions containing $\Delta_1^-$ in \eqref{eq:d-expansion} are smooth because they are of the form $\Delta_1^-\circ A$ with $\WF(A)\subset\WF(\De^+_1)$. Moreover, by an induction over $n$, one may show that the remaining terms in \eqref{eq:d-expansion} containing both $\De_1^+$ and $\De^F_1$ cancel exactly. Consequently we find $d_{1,Q,n}= i^n \De_1^F(\q\De^F_1)^{n} + A$, where $A$ has a smooth integral kernel. When applying $\bq = \alpha_{d_{1,Q}}$ to a non-linear local functionals, we encounter products of coinciding point limits of derivatives of $d_{1,Q}$, and consequently, products of coinciding point limits of derivatives of the integral kernel of the linear map $\De_1^F(\q\De^F_1)^{n}$ for arbitrary $n\in\bN$. In the language of Feynman diagrams, these expressions correspond just to ``big loops'' made of $n+1$ vertices joined by $n+1$ propagators $\Delta_1^F$ and with the operator $\q$ being applied in $n$ of this vertices, whereas the $n+1$--th vertex is $x$. These diagrams are renormalised by standard techniques while extending $\T[1]$ to local functionals, see e.g. \cite{BF00, HW02,BDF,FredenhagenRejzner,FredenhagenRejzner2}. However, in general not all $\De_1^F(\q\De^F_1)^{n}$, which are clearly monomials of $n$--th order in $Q$, need to be renormalised.

In order to analyse the interplay between the renormalisation freedom of $T_2\doteq \bq\circ T_1$ on local functionals and the renormalisation freedom of $T_1$ on multilocal functionals we consider the functional $\phi^2$, omitting the smearing for simplicity. The most general  definition of $T_1(\phi^2)\in\mA_1$ compatible with all axioms is \cite{HW01,HW05}
$$T_1(\phi^2(x)) = \phi^2(x) + W_1(x,x) + a R(x)$$
where $W_1(x,y)$ is the smooth part of $\Delta^+_1$ in the local Hadamard expansion \eqref{def:Hamardplus}, $R$ is the scalar curvature and $a$ is a dimensionless constant. In view of our construction in Theorem \ref{th:PPA} we consider $M_1=0$ and thus no corresponding term appears in $T_1(\phi^2)$. If we evaluate $T_2(\phi^2)$ based on $\bq = \alpha_{d_{1,Q}}$ with $d_{1,Q}(x,y)=\Delta^F_2(x,y)-\Delta^F_1(x,y)$ considered as an exact expression, we find
$$T_2(\phi^2(x)) = \phi^2(x) + W_2(x,x) + a R(x) + b M_2(x)$$
where the dimensionless constant $b$ corresponds to ``$\lim_{x\to y}\log(\lambda^2\sigma(x,y))$''. This expression is obviously divergent, but is regularised implicitly by expanding $d_{1,Q}$ perturbatively as explained above. In particular the value of $b$ is in one--to--one correspondence with the renormalisation freedom of the ``fish--graph'' $\Delta^F_1(x,y)^2$, which is in fact the only divergent graph contributing to $\bq(\phi^2)$.
\end{remark}

\begin{remark}\label{rem:PPAmetricchanges}
In principle one could try to use the same proof strategy as the one used in Theorem \ref{th:PPA} in order to prove the PPA for changes of the metric as well. An natural possibility, at least on topologically trivial subsets of $M$, would be to consider the Minkowski metric as a reference metric $g_0$ and to ``add'' the dependence of $T$ on non--trivial metrics in the same manner as in Theorem \ref{th:PPA}. The time--ordered product constructed in this way may be seen to satisfy all axioms for time--ordered products including the PPA, but one: whereas it will be local and covariant with respect to the reference metric $g_0$ and the metric perturbation $\delta g =\sqrt{\det(g_1)/\det{(g_0)}}g_1-g_0$, it will in general not be local and covariant (in the perturbative sense) with respect to $g_1$\footnote{We would like to thank Jochen Zahn for pointing this out to us.}. In particular one expects that demanding locality and covariance with respect to $g_1$ imposes conditions on the reference time--ordered product which can not be satisfied in the cases where the PPA is known to fail, e.g. in two spacetime dimensions. In addition to this issue the perturbative nature of such a construction in the case of metric changes is more severe than in the case of quadratic $Q$ without derivatives, because the discussion in Remark \ref{rem:exactPPA} implies by ``power counting'' that for $Q$ containing second derivatives in principle infinitely many loop graphs in the theory with the Minkowski metric contribute to the renormalisation freedom in the theory with a general metric.
\end{remark}

\begin{remark}\label{rm:qmoller_is_isomorphism}
The PPA implies that $\qmoller$ is a $*$--isomorphism between the subalgebra $\A0_{2}$ of $\mA_{2}(\Fmuc[2],\star_{2})$ which is $\star_{2}$--generated by elements of $\Ftloc[1]$ and the well--defined algebra of interacting observables $\mA_{1,Q}$, cf. Definition \ref{def:interactingalgebra}. In particular $\qmoller$ maps an element of the form $F\star_{2} G$ with $F,G\in\Floc$ to $\bqmoller(\bq^{-1}(F))\star_1 \bqmoller(\bq^{-1}(G))$. Moreover $\qmoller$ descends to a $*$--isomorphism between the on--shell algebras $\A0_{2}/\mI^0_{2}$ and $\Aon_{1,Q}$, where $\mI^0_{2}\doteq \mI_{2}\cap \A0_{2}$, cf. Definitions \ref{def:free_algebra} and \ref{def:interactingalgebra}.
\end{remark}

\section{The generalised Principle of Perturbative Agreement}
\label{sec:gPPA}
In the previous section, we have discussed the Principle of Perturbative Agreement for the case of quadratic actions, i.e. we have seen how quadratic interaction potentials can be treated either in perturbation theory or in an exact fashion, and in which precise sense these two possibilities are related.  In this section we aim to show how this analysis can be generalised to the case where an additional interaction potential is present which is in general of higher--than--quadratic order in the field. In this case, the question to be answered is whether treating the quadratic part of a general polynomial interaction potential either perturbatively or exactly gives the same results in the algebraic sense. We have seen in the previous section that, in the purely quadratic case, the perturbative agreement does not hold in the strong sense of a $*$--isomorphism between the algebras $\widetilde{\mA}_{1,Q}$ and $\mA_2$ if one is interested in algebras containing physically interesting observables, i.e. powers of the field at the same point (see however Remark \ref{rm:qmoller_is_isomorphism}). Yet, also in the generalised case it is instructive to first discuss the perturbative agreement in heuristic terms in order to grasp the essential ideas. To this avail, we consider once more a heuristic diagram.

\beq\label{diagramheuristic_gPPA}
\begin{tikzcd}[column sep=large, row sep=huge] 
\mA_1 & & \mA_2 \arrow{ll}[above]{\text{\normalsize $\qmoller$}}\\
& \widetilde{\mA}_{1,Q} \arrow[dashed]{ul}[above]{\text{\normalsize $\quad\quad\bqmoller$}} \arrow[dashed]{ur}[above]{\text{\normalsize $\bq\quad\quad\;$}} 	&  \\
\widetilde{\mA}_{1,Q+V} \arrow[dashed]{uu}{\text{\normalsize $\molh{1}{Q+V}$}} \arrow[dashed]{rr}{\text{\normalsize $\gamma_{1,Q,V}$}} & &  \widetilde{\mA}_{2,V} \arrow[dashed]{uu}[right]{\text{\normalsize $\;\molh{2}{V}$}}
\end{tikzcd}
\eeq

In this diagram, we use again dashed arrows to indicate that their sources are ill--defined and thus formal. $\mA_1=(\Fmuc[1],\star_1)$ is the exact algebra corresponding to an (at most) quadratic action $\cS_1$, $\qmoller$ is the (modified) classical M\o ller map and $\mA_2=(\Fmuc[2],\star_{2})$ is the exact algebra corresponding to the (at most) quadratic action $\cS_2 = \cS_1 + Q$, constructed in such a way that $\qmoller:\mA_2\to \mA_{1}$ is manifestly a $*$--homomorphism. Moreover, $\molh{X}{Y}$ indicates the quantum M\o ller map corresponding to the free action $X$ and the perturbation $Y$ and $\widetilde{\mA}_{X,Y}$ indicates the heuristic algebra of interacting observables constructed in such a way that $\molh{X}{Y}$ is formally a $*$--isomorphism between $\widetilde{\mA}_{X,Y}$ and a subalgebra of $\A0_X\subset\mA_X$. The upper triangle corresponds to the PPA in the quadratic case, whereas the {\bf generalised Principle of Perturbative Agreement (gPPA)} may be formally stated as to require that $\widetilde{\mA}_{1,Q+V}$ and $\widetilde{\mA}_{2,V}$ are isomorphic, the isomorphism being indicated by $\gamma_{1,Q,V}$. The PPA in the quadratic case implies that $\bq$ effectively intertwines between the two points of view that $Q$ is either a perturbation or part of the exact theory. In fact, the forthcoming analysis will show that this persists in the presence of an additional interaction $V$, i.e. that essentially  $\gamma_{1,Q,V}=\bq$.

In contrast to the quadratic case, we shall not analyse a rigorous version of this diagram given by the restriction to regular observables, because this is in general not possible if $V$ is of higher--than--quadratic order in $\phi$. Instead we shall directly use the PPA to prove the gPPA in a version which is useful for applications in perturbation theory. Namely, as discussed in Section \ref{sec:pAQFT}, one may construct well--defined algebras of interacting observables corresponding to a free action $X$ and a perturbation $Y$ by considering the algebras $\mA_{X,Y}$ which are generated by $\molh{X}{Y}(F)$, where $F$ is a time--ordered product of local or regular functionals. Thus, for applications it is sufficient to prove the gPPA in terms of a relation between the well--defined objects $\molh{1}{Q+T_1(V)}$ and $\molh{2}{T_2(V)}$, where the time--ordered maps appear because one would like to deal with interactions corresponding to local and covariant Wick polynomials, cf. Remark \ref{rem:PPAequivalent} and also Remark \ref{rem:Qrenfreedom}.

\begin{theorem}\label{th:gPPA} In addition to the notations and assumptions of Definition \ref{def:PPA}, we consider an arbitrary $V\in\Floc$ and denote for arbitrary quadratic actions $X$ and $X+Y$ of the form \eqref{def:quadraticactions} by $\molh{X}{Y}$ the quantum M\o ller map constructed by means of the products $\star_X$ and $\T[X]$ as in \eqref{def:quantum moller operator}. If the time--ordered map $T_X$ satisfies the PPA as in Definition \ref{def:PPA}, i.e. if $T_{X+Y} = \beta_{X,Y}\circ T_X$, then the following identity holds on $\Ftloc[1]$
\beq
\molh{1}{Q+T_1(V)}=\qmoller \circ \molh{2}{T_2(V)} \circ \bq\,.
\eeq
\end{theorem}

\begin{proof}
We recall that $\bq$ is a well--defined map between $\Ftloc[2]$ and $\Ftloc[1]$ and recall a few basic identities already used in the proof of Lemma \ref{pr:cocylebeta}.
$$ S_{1,V+Q}=S_{1,V}\T[1]S_{1,Q}$$
$$ \bq\left(S_{1,V}\right)=S_{{2},\bq(V)}\qquad \Leftrightarrow\qquad  \beta^{-1}_{1,Q}\left(S_{2,V}\right)=S_{{1},\bq^{-1}(V)} $$
$$\qmoller\left(S^{-1}_{2,V}\right) =\left(\qmoller\left(S_{2,V}\right)\right)^{-1}= \left(\bqmoller\left(S_{1,\bq^{-1}(V)}\right)\right)^{-1} $$
Using all of these basic identities, Theorem \ref{th:classical homomorphism}, the PPA and $\qmoller=\bqmoller\circ\bq^{-1}$, which can be proven by applying $\qmoller^{-1}$ to both sides, we compute for an arbitrary functional $F\in\Ftloc[1]$
\begin{flalign*}
&\phantom{=}\,\left[\qmoller\circ \molh{2}{T_2(V)}\circ  \bq\right](F)\\
&=\qmoller\left(S_{2,T_2(V)}^{-1}\star_{2}\left(S_{2,T_2(V)}\cdot_{T_{2}}\bq(F)\right)\right)\\
&=\qmoller\left(S_{2,T_2(V)}^{-1}\right)\star_1\qmoller\left(S_{2,T_2(V)}\cdot_{T_{2}}\bq(F)\right) \\
&=\left(\qmoller^\hbar\left(S_{1,T_1(V)}\right)\right)^{-1}\star_1 \left[\bqmoller\circ \bq^{-1} \right]\left(S_{2,T_2(V)}\cdot_{T_{2}}\bq(F)\right)\\
&=\left(\qmoller^\hbar\left(S_{1,T_1(V)}\right)\right)^{-1}\star_1 \bqmoller\left(S_{{1},T_1(V)}\cdot_{T_1} F\right) \\
&=\left[\left(S_{1,Q}\cdot_{T_1}S_{1,T_1(V)}\right)^{-1}\star_1 S_{1,Q}\right]\star_1\left[S^{-1}_{1,Q}\star_1 \left(S_{1,Q}\cdot_{T_1}S_{{1},T_1(V)}\cdot_{T_1} F\right)\right]\\
&=S^{-1}_{1,Q+T_1(V)}\star_1 \left(S_{1,Q+T_1(V)}\cdot_{T_1} F\right)\\
&=\molh{1}{Q+T_1(V)}(F)\,.
\end{flalign*}
\end{proof}

\begin{remark}\label{rm:qmoller_is_isomorphism_V}
Theorem \ref{th:gPPA} and Theorem \ref{th:classical homomorphism} imply that $\qmoller$ is a $*$--isomorphism between the well--defined interacting algebras of observables $\mA_{2,T_2(V)}$ and $\mA_{1,Q+T_1(V)}$ and their on--shell versions $\Aon_{2,T_2(V)}$ and $\Aon_{1,Q+T_1(V)}$, cf. Definition \ref{def:interactingalgebra}. This generalises the same relation for the case $V=0$, cf. Remark \ref{rm:qmoller_is_isomorphism}.
\end{remark}

\section{The thermal mass and KMS states for interacting massless fields in Minkowski spacetime}
\label{sec:KMS}
We shall now apply the generalised Principle of Perturbative Agreement in order to construct equilibrium states for interacting massless scalar fields in Minkowski spacetime. We will accomplish this task by extending the results obtained in \cite{FredenhagenLindner} for the massive case. In fact, in \cite{FredenhagenLindner}, the authors succeeded to construct a KMS state on the perturbatively constructed algebra of interacting observables for the case of a massive free field and an arbitrary local interaction $V$. This construction is carried out by generalising techniques of quantum statistical mechanics to the field--theoretic case. To this avail, the ill--defined Hamiltonian is replaced by a well--defined time--averaged Hamiltonian and the adiabatic limit is dealt with by using algebraic isomorphisms in order to restrict the discussion to a finite--time slab of Minkowski spacetime and by proving that the remaining adiabatic limit in the spatial directions is well--defined. In \cite{FredenhagenLindner}, it is not proved explicitly that the KMS states constructed hereby are independent of the finite--time slab chosen, however, we shall demonstrate in Section \ref{sec:independence} that this is indeed the case.

The analysis of the spatial adiabatic limit in \cite{FredenhagenLindner} relies heavily on the fact that  connected correlation functions of massive free fields in KMS states decay exponentially in spatial directions. For this reason, the results of \cite{FredenhagenLindner} can not be directly applied to the massless case. However, it is widely believed that the massless $\phi^4$--model in a thermal state shares at least some of the good infrared properties of its massive counterpart due to the occurrence of the {\bf thermal mass}. In the functional picture, this quantity can be understood as follows \cite{Lindner:2013ila}. We recall that a local functional such as $V(\phi)=\int_\mM f \phi^4 d\mu_g$, $f\in\mD(\mM)$, considered as an element of the algebra of free fields $\mA=(\Fmuc,\star)$, corresponds to a smeared field polynomial which is Wick--ordered with respect to the (symmetric part of) the bidistribution $\De^+$ defining the $\star$--product. Initially, the $\phi^4$--model in Minkowski spacetime is considered to be constructed based on the field monomial Wick--ordered w.r.t. the vacuum two--point function $\De^+_\infty$ of the free field, and thus as an element of the corresponding algebra $\mA_\infty=(\Fmuc,\star_\infty)$. However, for practical computations in a KMS state with inverse temperature $\beta$ it is more convenient to pass to the algebra $\mA_\beta=(\Fmuc,\star_\beta)$, in which the $\star$--product is induced by the two--point function $\De^+_\beta$ of the free field in the $\beta$--KMS state. As we have discussed in Section \ref{sec:pAQFT}, this is implemented by the isomorphism $\alpha_{d}:\mA_\infty\to\mA_\beta$, where $d=\De^+_\beta - \De^+_\infty$ and $\alpha_d$ is a contraction exponential of the form \eqref{def:alpha}. Under this isomorphism, $V$ transforms as $V\mapsto \alpha_d(V)=V+Q+C$, and thus picks up a quadratic term $Q=\frac12\int_\mM f m^2_\beta \phi^2 d\mu_g$ (and an irrelevant constant term $C$). The coefficient $m^2_\beta$ of this quadratic term is interpreted as the square of a thermal mass, and one may compute that it is proportional to $\beta^{-2}$ in the massless case.

Motivated by the thermal mass idea, we shall prove the existence of Minkowskian KMS states on the interacting algebra $\mA_{1,V}$ corresponding to the massless quadratic action $\cS_1$ and an arbitrary local interaction $T_1(V)=V\in\Floc$ as follows, where note that on Minkowski spacetime $T_1|_\Floc = \Id$ for mass $m=0$. We split $V$ as $V=Q+V-Q$, where $Q$ is an arbitrary non--trivial positive quadratic local functional corresponding to an arbitrary non--vanishing ``virtual mass''. We then consider the KMS state $\omega^\beta_{2,T_2(V-Q)}$ on the adiabatic limit of the interacting algebra $\mA_{2,T_2(V-Q)}$ constructed as in \cite{FredenhagenLindner}. Our previous analysis implies that the renormalisation freedom of the time--ordered product can be fixed in such a way that the algebras $\mA_{2,T_2(V-Q)}$ and $\mA_{1,V}$ are isomorphic, the isomorphism being the (modified) classical M\o ller map $\qmoller$. We shall argue that $\qmoller$ preserves the defining properties of KMS states in the adiabatic limit; this implies that $\omega^\beta_{1,V}\doteq\omega^\beta_{2,T_2(V-Q)}\circ\qmoller^{-1}$ is a well--defined KMS state on $\mA_{1,V}$ in this limit.

\subsection{KMS states for interacting massive fields in Minkowski spacetime}
\label{sec:KMSmassive}
In order to pursue the plan outlined above, we briefly review the construction of KMS states for interacting massive scalar fields in Minkowski spacetime as devised in \cite{FredenhagenLindner}. To this avail, we consider a free massive Klein--Gordon field on Minkowski spacetime $\mink=(\mM,g)$, i.e. $\mM=\bR^4$ and $g$ is the Minkowski metric. We denote by $\mA=(\Fmuc,\star)$ the algebra of observables of this free theory constructed as reviewed in Section \ref{sec:pAQFT}, suppressing the dependence of quantities on the field model (i.e. the mass) and the metric throughout this subsection. Here, the $\star$--product is constructed by means of a time--translation invariant Hadamard distribution $\De^+$. We further consider an arbitrary local interaction $V\in\Floc$, and denote this explicitly as $V(f)$ for $f\in\mD(\mM)$, in order to spell out the test function $f$ which cuts off the support of $V$ in spacetime. The algebra of interacting observables corresponding to this interaction will be denoted by $\mA_{V(f)}$, and we recall that this algebra is $\star$--generated by $\molhh{V(f)}(F)$ with $F\in\Ftloc$, i.e. $F$ is a time--ordered product of local and regular functionals. Here, $\molhh{V(f)}$ is the quantum M\o ller map, cf. \eqref{def:quantum moller operator}. Finally, we denote by $\Aon$ and $\Aon_{V(f)}$ the on--shell versions of $\mA$ and $\mA_V$, cf. Definitions \ref{def:free_algebra} and \ref{def:interactingalgebra}, and we denote by e.g. $\mA(\mO)$, $\mA_{V(f)}(\mO)$ the subalgebras of $\mA=\mA(\mink)$, $\mA_{V(f)}=\mA_{V(f)}(\mink)$ containing elements $F$ with $\supp(F)\subset \mO\subset \mink$.

We start by recalling the definition of KMS (Kubo--Martin--Schwinger) states.

\begin{definition} Let $(\mA,\star,\ast)$ be a $\ast$--algebra over $\bC$ and let $\{\alpha_t\}$ be a one--parameter group of automorphisms of $\mA$. A state $\kms$ on $\mA$ is called {\bf KMS state of inverse temperature $\beta$} (also {\bf $\beta$-KMS state}) with respect to $\{\alpha_t\}$ if the functions
\begin{gather*}
\bR^n\ni(t_1,\ldots,t_n)\to\kms\left(\alpha_{t_1}(F_1)\star\ldots\star\alpha_{t_n}(F_n)\right)\quad F_1,\ldots, F_n\in\mA,
\end{gather*}
admit an analytic continuation on
\begin{gather*}
\mI^\beta_n\doteq\{(z_1,\ldots,z_n)\in\bC^n|\ 0<\Im z_i-\Im z_j<\beta\,, 1\leq i < j \leq n\}\,,
\end{gather*}
which is bounded and continuous on the boundary of $\mI^\beta_n$, where it satisfies the boundary conditions
\begin{gather*}
\kms\left(\alpha_{t_1}(F_1)\star\ldots\star\alpha_{t_{k-1}}(F_{k-1})\star\alpha_{t_k+i\beta}(F_k)\star\ldots\star\alpha_{t_n+i\beta}(F_n)\right)\\=
\kms\left(\alpha_{t_k}(F_k)\star\ldots\star\alpha_{t_n}(F_n)\star\alpha_{t_1}(F_1)\star\ldots\star\alpha_{t_{k-1}}(F_{k-1})\right)
\end{gather*}
for all $k\in\{1,\ldots,n\}$.
\end{definition}

We would like to construct an interacting KMS state for the interaction $V(f)$ in the adiabatic limit $\lim f\to 1$ starting from a KMS state $\kms$ on $\mA$, i.e. formally we are interested in a state $\kms_{V(1)}$ on $\mA_{V(1)}$. The adiabatic limit $\lim_{f\to 1} \mA_{V(f)}$ may be understood algebraically as follows \cite{BF00}. Consider $F\in\Ftloc$ and $f,f^\prime\in\mD(\mM)$ s.t. $f=f^\prime$ on the causal completion of $\supp(F)$. Then there exists an $F$--independent unitary $U_{f,f^\prime}\subset \mA$ s.t. $\molhh{V(f^\prime)}(F) =  U_{f,f^\prime}\star\molhh{V(f)}(F)\star U^{-1}_{f,f^\prime}$. Thus, the algebra $\mA_{V(f)}(\mO)$, is uniquely determined by the form of $f$ on the causal completion of $\mO$ up to isomorphy, and, in this sense, one may consider $\lim_{f\to 1}\mA_{V(f)}(\mO)$ as being equal to $\mA_{V(f)}(\mO)$ with $f=1$ on the causal completion of $\mO$. 

In order to discuss KMS states on $\mA_{V(f)}$ in the adiabatic limit, we introduce a one--parameter isomorphism group on this algebra corresponding to time--translations. Let $\mE(\mM)\ni\phi\mapsto\phi_t\in\mE(\mM)$ be the map defined by $\phi_t(x)\doteq\phi(x-te_0)$, where $e_0$ is the time--direction of an arbitrary but fixed frame of $\mink$. By pullback we get a one--parameter group of automorphisms $\{\alpha_t\}$ of $\mA$ defined by $\alpha_t(F)(\phi)\doteq F(\phi_t)$, since, by assumption, the $\star$--product is implemented by a time--translation invariant $\De^+$. The time--translations $\{\alpha_t^{V(f)}\}$ on $\mA_{V(f)}$ in the algebraic adiabatic limit are defined by demanding that $\molhh{V(f)}$ intertwines the free and interacting dynamics, i.e. for $F\in\Ftloc$ s.t. $f=1$ on the causal completion of $\supp(F)\cup\supp(\alpha_t(F))$, 
\beq\label{def:interactingdynamics}\alpha_t^{V(f)}\left(\molhh{V(f)}(F)\right)\doteq\molhh{V(f)}\left(\alpha_t(F)\right)=\alpha_{t}\left(\molhh{\alpha_{-t}(V(f))}(F)\right).\eeq 
In this sense, one may think of $\{\alpha_t^{V(f)}\}$ with $f=1$ on the causal completion of $\mO$ to represent the interacting dynamics $\{\alpha_t^{V(1)}\}$  on $\mA_{V(f)}(\mO)$ in the adiabatic limit. 

The essential starting point of the construction of \cite{FredenhagenLindner} is to restrict both observables and interactions to an $\ep$-neighbourhood $\Si_\ep\doteq(-\ep,\ep)\times\Si$ of a Cauchy surface $\Si$ of $\mink$. For the former, one may use the time--slice axiom \cite{Chilian:2008ye} which implies $\Aon(\Sigma_{\epsilon})\simeq \Aon(\mink)=\Aon,\Aon_{V(f)}(\Sigma_\epsilon)\simeq \Aon_{V(f)}(\mink)=\Aon_{V(f)}$. In order to restrict the interaction to $\Si_\ep$, we consider a temporal cut--off $\chi\in\mD(\bR)$, i.e. an element of the set 
\beq\label{def:temporal_cutoff}
\cI_\ep\doteq \left\{\chi\in\mD(\bR)\,|\,\supp(\chi)\subset(-2\ep,2\ep)\,, \chi = 1 \text{ on } (-\ep,\ep)\right\}.
\eeq
In analogy to the discussion of the algebraic adiabatic limit, one may show that the algebras $\mA_{V(f)}(\Sigma_\epsilon)$ and $\mA_{V(\chi f)}(\Sigma_\epsilon)$ are isomorphic, where the isomorphism is implemented by unitaries in $\mA$. The limit $\lim_{f\to 1}\molhh{V(\chi f)}(F)$ is well--defined for all $F\in\Ftloc$ because $J^-(\supp(F))\cap \supp(\chi)$ is compact by the compact support of $F$. Consequently, the algebra $\Aon_{V(\chi)}(\Sigma_\epsilon)$ is well--defined and may be considered as a representation of the adiabatic limit of $\Aon_{V(f)}(\mink)$. In this representation, the interacting dynamics in the adiabatic limit is implemented by the one--parameter automorphism group $\{\alpha_t^{V(\chi)}\}$ satisfying \eqref{def:interactingdynamics} with $V(f)$ replaced by $V(\chi)$ if $\supp(F)\cup\supp(\alpha_t{F})\subset \Sigma_\epsilon$. We note that the elements of $\Aon_{V(\chi)}(\Sigma_\epsilon)$ are formal power series with values in $\Aon(\Sigma_{2\epsilon})$, in this sense $\Aon_{V(\chi)}(\Sigma_\epsilon)\subset \Aon(\Sigma_{2\epsilon})$.

The construction of a $\beta$--KMS state w.r.t. the $V(\chi)$--dynamics on $\Aon_{V(\chi)}(\Sigma_\epsilon)$ now proceeds as follows. We consider an $h\in\mD(\bR^3)$ s.t. $h=1$ on $B_r$, the sphere in $\bR^3$ with radius $r$ centred at $0$. We further consider $\mO\subset\Sigma_\epsilon$ s.t. the causal completion of $\mO$ is a subset of $(-\epsilon,\epsilon)\times B_r$. Then $\Aon_{V(\chi)}(\mO)=\Aon_{V(\chi h)}(\mO)$ and $\{\alpha_t^{V(\chi h)}\}=\{\alpha_t^{V(\chi)}\}$ on $\Aon_{V(\chi h)}(\mO)$, where $\{\alpha_t^{V(\chi h)}\}$ is defined as in \eqref{def:interactingdynamics}. One can show that \cite{FredenhagenLindner}, for sufficiently small $t$, the $V(\chi h)$--interacting and the free time--evolution in $\mO$ are intertwined by unitaries $U^\chi_h(t)\in\mA(\Sigma_{2\epsilon})$
\beq
\label{co-cycle}
\alpha_t^{V(\chi h)}\left(\molhh{V(\chi h)}(F)\right)=
U^\chi_h(t)\star\alpha_t\left(\molhh{V(\chi h)}(F)\right)\star U^{\chi}_h(t)^{-1}\,.
\eeq
The unitaries $U^\chi_h(t)$ satisfy the co-cycle condition $U^\chi_h(t+s)=U^\chi_h(t)\star\alpha_t(U^\chi_h(s))$ for sufficiently small parameters. From this one can infer that the infinitesimal generator of $U^\chi_h(t)$ is
\beq
\label{infinitesimal generator of the co-cycle}
K^\chi_h\doteq
\frac 1i\left.\frac{d}{dt}U^\chi_h(t)\right|_{t=0}=
\molhh{V(\chi h)}\left(V(h\dot{\chi}^-)\right),
\eeq
where $\dot{\chi}^-(t)\doteq\dot{\chi}(t)\Theta(-t)$, with $\Theta$ denoting the Heaviside step function.

Based on this, the authors of \cite{FredenhagenLindner} first construct a $\beta$--KMS state w.r.t. the $V(\chi h)$--dynamics on $\Aon_{V(\chi h)}(\Sigma_\epsilon)$, which is clearly a $\beta$--KMS state w.r.t. to the $V(\chi)$--dynamics on $\Aon_{V(\chi)}(\mO)$ with $\mO$ as above. In order to obtain a KMS state w.r.t. $V(\chi)$ rather then $V(\chi h)$ on $\Aon_{V(\chi)}(\Sigma_\epsilon)$, the adiabatic limit $h\to 1$ is taken in the following sense due to van Hove. A sequence $(a_h)_{h\in\mD(\bR^3)}$ admits van Hove--limit $\vhlim_{h\to 1}a_h$ if, for all possible choices of sequences $(h_n)_n\subseteq\mD(\bR^3)$ such that $0\leq h_n\leq 1, h|_{B_n=1},\ h|_{\bR^3\setminus B_{n+1}}=0$, it holds that $\lim_{n\to\infty}a_{h_n}$ is finite and does not depend on the sequence $(h_n)_n$.

After these preparatory considerations we now collect the subsidiary and final results of \cite{FredenhagenLindner} in the form of a single theorem for the sake of brevity.

\begin{theorem}[\cite{FredenhagenLindner}]\label{th:KMSmassive}
Let $\kms$ be the $\beta$-KMS state on $\Aon$ with respect to $\{\alpha_t\}$.
\begin{itemize}

\item[(i)] For $h\in\mD(\bR^3)$ s.t. $h|_{B_r}=1$ and $\mO\subset\Sigma_\epsilon$ s.t. the causal completion of $\mO$ is a subset of $(-\epsilon,\epsilon)\times B_r$, the linear functional $\kms_{V(\chi h)}:\Aon_{V(\chi h)}(\mO)\to\bC$
\beq\label{interacting KMS state}
\kms_{V(\chi h)}(F)\doteq
\frac{\kms\left(F\star U^\chi_h(i\beta)\right)}{\kms(U^\chi_h(i\beta))},
\eeq
is a $\beta$-KMS state with respect to $\{\alpha_t^{V(\chi h)}\}$.

\item[(ii)] The interacting KMS state (\ref{interacting KMS state}) can be perturbatively expanded in terms of the connected (also called truncated) correlation functions $\kms_c$ of $\kms$,  viz.
\beq
\label{formula for the interacting KMS state}
\kms_{V(\chi h)}(F)=
\sum^\infty_{n= 0}(-1)^n\int_{\beta\cS_n}\ \kms_c\left(F\star\alpha_{iu_1}(K^\chi_h)\star\ldots\star\alpha_{iu_n}(K^\chi_h)\right)
dU,
\eeq
where $dU\doteq du_1\ldots du_n$ and $\beta\cS_n\doteq\{(u_1,\ldots,u_n)\in\bR^n|\ 0\leq u_1\leq\ldots\leq u_n\leq\beta\}$.

\item[(iii)] For arbitrary $\chi,\chi^\prime\in\cI_\ep$, $\kms_{V(\chi h)}(F)=\kms_{V(\chi^\prime h)}(F)$ for all $F\in\Aon_{V(\chi h)}(\mO)$ with $\mO$ as in (i). Thus, prior to the adiabatic limit $h\to 1$, the interacting KMS state depends at most on $h$ and the size of the time--slab  $\Si_\ep$ of Minkowski spacetime.

\item[(iv)] In the adiabatic limit $h\to 1$, the generator $K^\chi_h=\mR^\hbar_{V(\chi h)}(V(h\dot{\chi}^-))$ \eqref{infinitesimal generator of the co-cycle} in \eqref{formula for the interacting KMS state} may be replaced by $\mR^\hbar_{V(h)}(V(h\dot{\chi}^-))$. By doing so, each term in the formal sum \eqref{formula for the interacting KMS state} is multilinear in $h$. Consequently, if for all $n\in\bN$ and $F_0,F_1,\ldots, F_n\in\mA(\Sigma_{2\epsilon})$ the functions
\begin{gather*}
\beta\cS_n\times\bR^{3n}\ni(u_1,\ldots,u_n;\bx_1,\ldots,\bx_n)\mapsto\ 
\kms_c\left(F_0\star
\alpha_{iu_1,\bx_1}(F_1)\star\ldots\star\alpha_{iu_n,\bx_n}(F_n)
\right),
\end{gather*}
are in $L^1(\beta\cS_n\times\bR^{3n})$, then the van Hove--limit
\begin{gather}\label{interacting KMS state in the adiabatic limit}
\kms_{V(\chi)}(F)\doteq
\vhlim_{h\to 1}\ \kms_{V(\chi h)}(F),
\end{gather}
exists and defines a $\beta$-KMS state on $\Aon_{V(\chi)}(\Sigma_\epsilon)$ with respect to $\{\alpha_t^{V(\chi)}\}$. Here $\alpha_{t,\bx}$ denotes a spacetime translation by $(t,\bx)$ implemented on $\Aon_{V(\chi)}(\Sigma_\epsilon)$ in analogy to the time translation $\alpha_t$ and the analytic continuation of $\alpha_{t,\bx}$ to imaginary $t$ is understood in the weak sense and well--defined on account of the KMS property of $\omega^\beta$.

\item[(v)] In the case of a massive Klein--Gordon field on Minkowski spacetime, the $\beta$--KMS states  on $\Aon$ satisfy the integrability condition in (v) for all $0<\beta\leq\infty$, including the vacuum for $\beta=\infty$. Consequently, in this case, \eqref{interacting KMS state} and \eqref{interacting KMS state in the adiabatic limit} define an interacting $\beta$--KMS state on $\Aon_{V(\chi)}(\Sigma_\epsilon)$.
\end{itemize}
\end{theorem}

\subsection{Independence of the interacting KMS state on the temporal cut--off}
\label{sec:independence}
As argued in the previous subsection, the algebra $\Aon_{V(\chi)}(\Sigma_\epsilon)$ may be considered as a representation of the algebra $\Aon_{V(f)}(\mink)$ in the adiabatic limit $f\to 1$. Consequently, by pullback, the KMS state on $\Aon_{V(\chi)}(\Sigma_\epsilon)$, constructed in \cite{FredenhagenLindner} as reviewed above, may be considered as a KMS state on $\Aon_{V(f)}(\mink)$ in the adiabatic limit. One expects that, in the absence of phase transitions, a $\beta$--KMS state is uniquely determined by $\beta$ and the one--parameter automorphism group. Consequently, $\kms_{V(\chi)}$ should be independent of both $\chi\in\cI_\epsilon$ and $\epsilon$. The question whether this holds was left open in \cite{FredenhagenLindner}, cf. the comments at the end of Section 4 in \cite{FredenhagenLindner}. In the following, we shall prove that this is indeed the case.

\begin{proposition}\label{pr:chi_independence} Using the notation of Section \ref{sec:KMSmassive}, consider an arbitrary but fixed $\epsilon>0$, the corresponding finite--time slab $\Sigma_\epsilon$ of Minkowski spacetime, and an arbitrary but fixed $\chi\in\cI_\epsilon$. Moreover, let $F_1,\ldots,F_n$ be arbitrary elements of $\Ftloc$ with $\supp(F_i)\subset \Sigma_\epsilon$ for all $i\in\{1,\ldots,n\}$ and let $\kms_{V(\chi)}$ be the $\beta$--KMS state on $\Aon_{V(\chi)}(\Sigma_\epsilon)$ constructed as in Theorem \ref{th:KMSmassive}. Then, the following statements hold for the expectation value
$$E_{\chi}\doteq \kms_{V(\chi)}\left(\mR^\hbar_{V(\chi)}(F_1)\star\ldots\star\mR^\hbar_{V(\chi)}(F_n)\right)\in\bC.$$ 

\begin{itemize}
\item[(i)] $E_\chi$ is independent of $\chi$, i.e. $E_\chi=E_{\chi^\prime}$ for all $\chi,\chi^\prime\in\cI_\epsilon$.
\item[(ii)] $E_\chi$ is independent of $\epsilon$, i.e. for all $\epsilon^\prime>\epsilon$ and all $\chi\in\cI_\epsilon$, $\chi^\prime\in\cI_{\epsilon^\prime}$, $E_\chi=E_{\chi^\prime}$.
\end{itemize}
\end{proposition}
\begin{proof}
Proof of (i). We prove the statement for a single generator $\mR^\hbar_{V(\chi)}(F)$, the general case follows analogously. To this avail, we set
$$
f_n(h_1,h_2,h_3,\chi)\doteq (-1)^n\int_{\beta S_n}\!\!\!dU\ \kms_c\left(\mR^\hbar_{V(h_1\chi)}(F)\star
\prod^\star_{j\in\{1,\ldots,n\}}\alpha_{iu_j}\left(\mR^\hbar_{V(h_3\chi)}(V(h_2\dot{\chi}^-))\right)\right)\,.
$$
By Theorem \ref{th:KMSmassive} (iv), we have $E_\chi=\vhlim_{h\to 1}\sum_{n} f_n(1,h,1,\chi)$. We would like to show that this van Hove--limit is independent of $\chi$. We shall demonstrate this for the $n=1$ contribution to the formal sum, the case of general $n$ can be shown by entirely analogous arguments. To this end, we consider a van Hove--sequence $(h_k)_k$ such that $\vhlim_{h\to 1}f_1(1,h,1,\chi)=\lim_{k\to\infty}f_1(1,h_k,1,\chi)$ and that the derivatives of $h_k$ are uniformly bounded in $k$. We analyse the difference $f_1(1,h_k,1,\chi)-f_1(h_k,h_k,h_k,\chi)$ by splitting it as
$$f_1(1,h_k,1,\chi)-f_1(h_k,h_k,h_k,\chi)=R_k(\chi) + S_k(\chi)\,,$$
$$R_k(\chi)\doteq f_1(1,h_k,1,\chi)-f_1(1,h_k,h_k,\chi)\,,\qquad S_k(\chi)\doteq f_1(1,h_k,h_k,\chi)-f_1(h_k,h_k,h_k,\chi)\,.$$
For sufficiently large $k$, $S_k(\chi)$ vanishes because $V(\chi h_k)-V(\chi)$ is causally disjoint from $\supp(F)$ for $k\gg 1$. The same argument does not hold for $R_k(\chi)$, because the support of $V(h_k\dot{\chi}^-)$ is increasing in line with the support of $V(\chi h_k)$. However, for $k\gg 1$, 
the supports of $\chimoller(V(h_k\dot{\chi}^-))-\chikmoller(V(h_k\dot{\chi}^-))$ and $\chimoller(F)$, considered as elements of the free algebra $\mA(\Sigma_{2\epsilon})$, are spacelike separated for large $k$, with the spacelike separation monotonically increasing in $k$. 

We would like to deduce from this fact that $\lim_{k\to\infty}R_k(\chi)=0$. Yet, as discussed at the end of Section 4 in \cite{FredenhagenLindner} the clustering properties of $\kms_c$ recalled in Theorem \ref{th:KMSmassive} (v) and proved in \cite{FredenhagenLindner} are not sufficient to control the limit $\lim_{k\to\infty}f_1(1,h_k,h_k,\chi)$, but stronger clustering properties are necessary. In fact, one needs to prove that for all $A_i \in \mA(\Sigma_{2\epsilon})$ and all $\bm\in\bN^k$, the function
$$
F^{\bm}_{k}(u_1,\bx_1,\dots, u_k,\bx_k)(g_0,\dots, g_k) \doteq 
\omega_c^\beta\left([A_0]_{g_0\chi}^{(m_0)}\star \alpha_{u_1,\bx_1}\left([A_1]_{g_1\chi}^{(m_1)}\right) \star \ldots \star \alpha_{u_n,\bx_k}\left([A_k]^{(m_k)}_{g_k\chi}\right) \right)
$$
is uniformly bounded by an absolutely integrable function for all $g_i$ in some bounded set of $\mD(\mM)$. Here, $[A]_g^{(m)}$ denotes the $m$--th perturbation order of $\mR^\hbar_{V(g)}(A)$. We prove in Lemma \ref{pr:clustering} that, as conjectured in \cite{FredenhagenLindner}, this clustering property indeed holds for the massive Klein--Gordon field on Minkowski spacetime, whence $\lim_{k\to\infty}R_k(\chi)=0$.

 As the final ingredient of our proof, we note that, by Theorem \ref{th:KMSmassive} (iii), $f_1(h_k,h_k,h_k,\chi) = f_1(h_k,h_k,h_k,\chi^\prime)$ for all $\chi,\chi^\prime\in\cI_\epsilon$ and $k$ large enough so that the causal completion of $\supp(F)$ is contained in $((-\epsilon,\epsilon)\times B_k)$, where we recall that $B_k\subset \bR^3$ is the sphere of radius $k$ centred at $0$ and that, by assumption, $h_k=1$ on $B_k$. Collecting all these observations, we may compute for arbitrary $\chi,\chi^\prime\in\cI_\epsilon$
\begin{align*}
&\lim_{k\to\infty}\left(f_1(1,h_k,1,\chi)-f_1(1,h_k,1,\chi^\prime)\right)\\
=&\lim_{k\to\infty}\left(f_1(h_k,h_k,h_k,\chi)-f_1(h_k,h_k,h_k,\chi^\prime)+R_k(\chi)+S_k(\chi)-R_k(\chi^\prime)-S_k(\chi^\prime)\right)\\
=&\,\, \,0\,.
\end{align*}

Proof of (ii). In order to simplify notation, we define an equivalence relation on $\mD(\bR)$ by
$$\chi\sim_\ep\chi^\prime\quad\Leftrightarrow\quad\chi,\chi^\prime\in \cI_\ep\,.$$
With this notation, we can rephrase the result of (i) by saying that $E_\chi = E_{\chi^\prime}$ if $\chi\sim_\ep\chi^\prime$. In order to prove the statement, we first consider $\ep^\prime\in(\ep,2\ep)$. Then, for all $\chi\in \cI_\ep$ and all $\chi^{\prime}\in \cI_{\ep^{\prime}}$, we can find a $\chi^{\prime\prime}\in \cI_\ep\cap \cI_{\ep^\prime}$ such that $\chi\sim_\ep\chi^{\prime\prime}\sim_{\ep^{\prime}}\chi^{\prime}$. If $\ep^\prime>2\ep$ we can choose a \textit{finite} sequence $\ep_1,\ldots,\ep_n$, with $\ep<\ep_{1}<2\ep$, $\ep_j<\ep_{j+1}<2\ep_j$ for all $j\in\{1,\ldots,n-1\}$,  $\ep_{n}<\ep^\prime<2\ep_{n}$, and $\chi_1,\ldots,\chi_{n+1}\in\mD(\bR)$ such that $\chi_1\in\cI_\ep\cap \cI_{\ep_1}$, $\chi_{j+1}\in\cI_{\ep_j}\cap \cI_{\ep_{j+1}}$ for all $j\in\{1,\ldots,n-1\}$, $\chi_{n+1}\in\cI_{\ep_n}\cap \cI_{\ep^\prime}$. With this choice of $\chi_j$, we obtain
$\chi\sim_\ep \chi_1 \sim_{\ep_1}\ldots\sim_{\ep_n}\chi_{n+1}\sim_{\ep^\prime} \chi^\prime$,
which proves the statement on account of (i).
\end{proof}

\subsection{KMS states for interacting massless fields in Minkowski spacetime}
\label{sec:KMSmassless}

We shall now combine the generalised Principle of Perturbative agreement (gPPA) with the construction of interacting massive KMS states reviewed in Section \ref{sec:KMSmassive} in order to construct KMS states for massless interacting fields on Minkowski spacetime. To this avail, we combine the notation used in  \ref{sec:KMSmassive} and in the earlier sections of this paper. Quantities indexed by ``$1$" shall refer to the free massless Klein--Gordon field. For definiteness we choose the $\star$--product on the free algebra $\mA_1 = (\Fmuc,\star_1)$ which is induced by the two--point function of the quasifree massless vacuum state $\Delta^+_1(f,g)\doteq\omega^\infty_1(f,g)$, $f,g\in\mD(\mM)$. We consider an arbitrary local interaction $V(f)\in\Floc$, spelling out the test function $f$ which localises the support of the interaction to a compact region of spacetime. Finally, we consider a quadratic local functional $Q(f)$ of the form $[Q(f)](\phi)\doteq \frac12 m^2_Q \int_\mM f \phi^2 d\mu_g$, where $m_Q > 0$ is an arbitrary but fixed ``virtual mass'' and $f$ is a positive test function.

In order to apply the gPPA, we recall the relevant objects for the convenience of the reader, where we use a slightly different notation than in Section \ref{sec:PPA} and \ref{sec:gPPA} in order to spell out the dependence of quantities on $f$. $\mR_{1,Q(f)}:\Fmuc\to\Fmuc$ is the classical M\o ller map defined as the pullback of $$R_{1,Q(f)}:\mE(\mM)\to\mE(\mM)\,,\qquad R_{1,Q(f)} \doteq \Id + \De^R_{1+Q(f)} \circ Q(f)^{(1)} = \Id + \De^R_{1+Q(f)} m^2_Q f\,,$$ where $\De^R_{1+Q(f)}$ is the advanced Green's operator of the Klein--Gordon operator $P_{1+Q(f)}\doteq P_1+Q(f)^{(1)}= -\Box_g + m^2_Q f$. We define \beq\label{eq:kmsdelta2}\De^+_{1+Q(f)} = R_{1,Q(f)} \circ \De^+_1 \circ R_{1,Q(f)}^\dagger\,\eeq which is of Hadamard form w.r.t. $P_{1+Q(f)}$, and a corresponding $\star$--product $\star_{1+Q(f)}$ such that $\mR_{1,Q(f)}:\mA_{1+Q(f)}\doteq(\Fmuc,\star_{1+Q(f)})\to\mA_1$ is a $*$--isomorphism, cf. Theorem \ref{th:classical homomorphism}. Note that we use the same Minkowski metric for both field theoretic models such that we do not have to deal with two different spaces of microcausal functionals and two different integration measures. By Theorem \ref{th:PPA} we can fix the renormalisation freedom of the time--ordered map in such a way that $T_{1+Q(f)}=\beta_{1,Q(f)}\circ T_1$ for all multilocal functionals. Under these conditions, Theorem \ref{th:gPPA} implies that the classical M\o ller map ${\mR}_{1,Q(f)}$ is a $*$--isomorphism between the algebras of interacting observables $\mA_{1+Q(f),T_{1+Q(f)}(V(f)-Q(f))}$ and $\mA_{1,V(f)}$, cf. Remark \ref{rm:qmoller_is_isomorphism_V}.

We now fix $\epsilon>0$ and consider the $\epsilon$--neighbourhood of a Cauchy surface $\Sigma$ of Minkowski spacetime $\Sigma_\epsilon\doteq(-\epsilon,\epsilon)\times\Sigma$. We further choose a temporal cut--off $\chi\in\cI_\epsilon\subset \mD(\bR)$, cf. \eqref{def:temporal_cutoff}. By the above considerations and arguments reviewed in Section \ref{sec:KMSmassive}, we have for all $\mO\subset\Sigma_\epsilon$ the following $*$--isomorphisms 
\beq\label{eq:isos1}
\mA_{1,V(f)}(\mO)\;\simeq \;\mA_{1+Q(f),T_{1+Q(f)}(V(f)-Q(f))}(\mO)\;\simeq\;\mA_{1+Q(f),T_{1+Q(f)}(V(\chi f)-Q(\chi f))}(\mO)
\eeq
induced by 
\begin{align}\label{eq:isos1_generators}
\molh_{1,V(f)}(F)&= \left[{\mR}_{1,Q(f)}\circ\molh{1+Q(f)}{T_{1+Q(f)}(V(f)-Q(f))}\circ {\beta_{1,Q(f)}}\right](F)\\
&=U\star_1\left[{\mR}_{1,Q(f)}\circ\molh{1+Q(f)}{T_{1+Q(f)}(V(\chi f)-Q(\chi f))}\circ {\beta_{1,Q(f)}}\right](F)\star_1 U^{-1}\notag\,,
\end{align}
where $U\in\mA_{1}$ is a unitary depending on $f$ and $\chi$ but not on $F$, $F$ is an arbitrary element of $\Ftloc[1]$ with $\supp(F)\subset \mO$, and where we used that ${\mR}_{1,Q(f)}$ maps unitaries to unitaries and that ${\beta_{1,Q(f)}}$ preserves the support.

We recall that changing the support of $f$ in $\mA_{1,V(f)}(\mO)$ outside of the causal completion of $\mO$ gives an isomorphic algebra. By isomorphy, this applies also to the other two algebras in \eqref{eq:isos1}. Consequently, we may consider the algebraic adiabatic limit $f\to 1$ for all three algebras in \eqref{eq:isos1}. In this limit, we define a one--parameter automorphism group $\{\alpha^{1,V(f)}_t\}$ on $\mA_{1,V(f)}$ by means of the free time--evolution $\{\alpha_t\}$ on $\mA_1$ as reviewed in Section \ref{sec:KMSmassive}. By isomorphy, this induces one--parameter automorphism groups also on the other two algebras in \eqref{eq:isos1}. 

We would now like to identify the algebraic adiabatic limit of the rightmost algebra in \eqref{eq:isos1} with the limit $f\to 1$ in the strict sense. To this avail, we need to check whether ${\beta_{1,Q(f)}}$, $T_{1+Q(f)}$ and $\star_{1+Q(f)}$ are well--defined in the limit $f\to 1$. For the latter, we note that it is not difficult to check that the corresponding limit of $\De^+_{1+Q(f)}$ defined in \eqref{eq:kmsdelta2} is well--defined and gives the two--point function of the (quasifree) vacuum state for the mass $m_Q$ since $\De^+_1$ is the corresponding correlation function of the massless state, cf. Lemma \ref{pr:adiabaticvacuum}. Similarly $T_{1+Q(f)}$ is well--defined in the adiabatic limit because $\De^F_{1+Q(f)}$ is well--defined in this limit and the renormalisation in $T_{1+Q(f)}$ is performed in a way which is analytic in $Q(f)$ as required by the axioms for $T$.

This implies that the limit of $\beta_{1,Q(f)}=\mR_{1,Q(f)}^{-1}\circ\molh{1}{Q(f)}$ is well--defined and corresponds to the contraction exponential $\alpha_d$ with $d$ given by the difference of the massive and massless Feynman propagator of the corresponding (quasifree) vacuum states; we refer to Proposition \ref{pr:betalocal} for the interpretation of this contraction exponential on general local functionals.

The above considerations imply that $\beta_{1,Q(1)}$ and $T_{1+Q(1)}$ commute with the free dynamics $\{\alpha_t\}$, which in turn entails that the one--parameter automorphism group on the rightmost algebra in \eqref{eq:isos1} induced via isomorphy by the interacting dynamics on $\mA_{1,V(f)}(\mO)$ in the adiabatic limit is in fact that interacting dynamics on the massive interacting algebra $\mA_{1+Q(1),T_{1+Q(1)}(V(\chi)-Q(\chi))}(\mO)$. We may construct a $\beta$--KMS state w.r.t. to this dynamics as reviewed in Theorem \ref{th:KMSmassive}, where we recall that this state does not depend on $\chi$ and $\epsilon$ by Proposition \ref{pr:chi_independence}. By \eqref{eq:isos1} and the above discussion, this induces a $\beta$--KMS state w.r.t. to the interacting dynamics on $\mA_{1,V(f)}(\mO)$ in the adiabatic limit.

A natural question is whether the $\beta$--KMS state for zero mass and interaction $V$ constructed as above depends on the virtual mass $m_Q$ used in the construction. We expect that this is not the case by uniqueness of the $\beta$--KMS state in the absence of phase transitions\footnote{In the massless case, we have to add the condition that the free KMS state is quasifree in order to exclude zero modes.}. In fact, if we choose two different $m_Q$, $m_{Q^\prime}$, then we obtain two chains of isomorphisms of the form \eqref{eq:isos1}. This and the above discussion imply that the corresponding algebras in the adiabatic limit $\mA_{1+Q(1),T_{1+Q(1)}(V(\chi)-Q(\chi))}(\mO)$ and $\mA_{1+Q^\prime(1),T_{1+Q^\prime(1)}(V(\chi)-Q^\prime(\chi))}(\mO)$ are isomorphic. Let us denote this isomorphism by $i_{Q,Q^\prime}$. We know that $i_{Q,Q^\prime}$ intertwines the interacting dynamics on the two algebras and that the construction of the $\beta$--KMS states on these algebras performed as in Theorem \ref{th:KMSmassive} is based on the generators of cocycles induced by the corresponding one--parameter automorphism groups. Consequently, $i_{Q,Q^\prime}$ intertwines these KMS states and thus the massless KMS state for the interaction $V$ induced by $m_Q$, $m_{Q^\prime}$ is the same.

For definiteness, we may subsume the above considerations in form of a proposition which we formulate only for local observables for simplicity.
\begin{proposition}
Consider an arbitrary local interaction $V$, arbitrary $\epsilon>0$ and $\chi\in\cI_\epsilon$ and an arbitrary $Q(f)$ of the form $[Q(f)](\phi)\doteq \frac12 m^2_Q \int_\mM f \phi^2 d\mu_g$ with $m_Q>0$. Moreover let $\star_1$ be a $\star$--product for the massless Klein--Gordon field which is induced by a time--translation invariant $\De^+_1$ of Hadamard form and denote by $\kms_{1+Q(1),T_{1+Q(1)}(V(\chi)-Q(\chi))}$ the KMS state for mass $m_Q$ and interaction $T_{1+Q(1)}(V-Q)$ constructed as reviewed in Theorem \ref{th:KMSmassive}. Then the renormalisation freedom of the time--ordered map can be fixed in such a way that, for all $F_1,\ldots,F_n\in\Floc$ with $\cup_i\supp(F_i)\subset \Sigma_\epsilon$
\begin{align*}&\phantom{=}\;\lim_{f\to 1}\kms_{1,V(f)}\left(\molh{1}{V(f)}(F_1)\star_1 \ldots\star_1\molh{1}{V(f)}(F_n)\right)\\
&\doteq \kms_{1+Q(1),T_{1+Q(1)}(V(\chi)-Q(\chi))}\left(\molh{1+Q(1)}{T_{1+Q(1)}(V(\chi)-Q(\chi))}(F_1)\star_{1+Q(1)}\ldots\right.\\
&\qquad\qquad\qquad\qquad\qquad\qquad\left.\ldots\star_{1+Q(1)} \molh{1+Q(1)}{T_{1+Q(1)}(V(\chi)-Q(\chi))}(F_n)\right)
\end{align*}
defines a $\beta$--KMS state for the massless Klein--Gordon field with interaction $V$ which is independent of $\chi$, $\epsilon$ and $m_Q$.
\end{proposition}

\begin{remark}We would like to stress that we use the construction above only for the case where the mass of the field for which we would like to construct a KMS state is zero and the virtual mass is positive. In particular we do not make any claims about the case of a tachyonic (imaginary) mass.

Even so one may doubt the above construction because treating a massless field as a massive one with the negative mass square as a perturbation seems inconsistent as it is known that massless and massive (linear) theories have different properties. However Lemma \ref{pr:neumann series} shows that the perturbative expansion of the classical M\o ller map converges for localised mass perturbations of arbitrary sign and magnitude on Minkowski spacetime. This result can presumably be generalised by proving that in a theory on Minkowski spacetime with interaction $W=\mu Q + \lambda V$ with $Q\in\Floc$ quadratic and $V\in\Floc$ arbitrary, expressions which are perturbative in $\mu$ converge at each order in $\lambda$. Finally Lemma \ref{pr:adiabaticvacuum} shows that the ``resummed propagators'' obtained from the abovementioned convergent perturbation series converge in the adiabatic limit as long as both the ``initial'' and ``final'' mass are either vanishing or positive (in four spacetime dimensions), see also \cite{Aste:2007ii} for related and essentially equivalent considerations. Consequently, we may consistently ``add'' or ``remove'' masses perturbatively as long as we avoid the tachyonic regime. In fact, Lemma \ref{pr:adiabaticvacuum} does not hold for tachyons, because the the ``naive tachyonic vacuum'' does not exist due to infrared problems.
\end{remark}

\begin{remark} In the vacuum case $\beta=\infty$ there is no thermal mass, but our construction can be still applied and yields a ground state for zero mass and arbitrary interaction $V$ because we may introduce the virtual mass $m_Q>0$ at will. In this context, we would like to stress that the convexity properties of the total (``free'' + ``interacting'') potential are not altered by the introduction of the virtual mass. For this reason the virtual mass has no effect on stability properties of the theory. Furthermore, in $\lambda \phi^4$ theory in a thermal state, the convexity of the potential is even reinforced by the presence of the thermal mass. 
\end{remark}

\begin{remark} We do not see any obstruction to apply the same construction as above in order to construct KMS states for the Klein--Gordon field with arbitrary mass and arbitrary interaction $V$ on static spacetimes $(\mN,g_\mathrm{us})$, $\mN\subset \bR^4$, perturbatively over Minkowski spacetime under the conditions that a) the metric $g_\mathrm{us}$ differs from the Minkowski metric $g_\mink$ only on a compact set b) $g_\mathrm{us}>g_\mink$, cf. \eqref{def:metriccondition} and c) $g_\mathrm{us}$ and $g_\mink$ have a common time--like Killing vector field. An example would be a compact diamond in the static patch of de Sitter spacetime.
\end{remark}

\section*{Acknowledgments}
The authors would like to thank Klaus Fredenhagen, Stefan Hollands and Jochen Zahn for interesting and helpful discussions. The work of T.-P.H. has been supported by a Research Fellowship of Deutsche Forschungsgemeinschaft (DFG).

\appendix

\section{Factorisation property of the \texorpdfstring{$S$}{S}--matrix}
\label{sec:factorisation}

In this section we would like to show that for $F,G\in\Ftloc$, $F\gtrsim G$ it holds
\begin{gather}\label{factorisation}
S_{F+G+V}=S_{F+V}\star S_V^{-1}\star S_{V+G},
\end{gather}
for any $V\in\Floc$; here the inverse is meant w.r.t. to the $\star$-product.

To prove this factorisation property we need a

\begin{lemma}\label{lemma}
Let $F$, $G\in\Fmuc$ be arbitrary functionals whose functional derivatives $F^{(n)}$ are supported on the diagonal $\sD_n\subset \mM^n$ but do not necessarily satisfy $\WF(F^{(n)})\perp T\sD_n$ and assume further that $\supp(F)\cap\supp(G)$ is a subset of a Cauchy surface $\Sigma$, whereas $\supp(F)\subset J^+(\Sigma)$ and $\supp(G)\subset J^-(\Sigma)$. In this case $F\T G$ is given by $F\star G+R_\Sigma(F,G)$ where $R_\Sigma(F,G)$ is supported on $\Sigma$ but not uniquely defined.
\end{lemma}
\begin{proof}
By a direct computation we get up to renormalisation
\begin{gather}
F\T G =F\cdot G+
\ldots+
\frac{\hbar^n}{n!}\left(\Delta^F\right)^{\otimes n}\left(F^{(n)},G^{(n)}\right)+\ldots
\end{gather}
We first consider $n=1$: in this case we have
\begin{gather}\label{eq:Blemma1}
\Delta^F\left(F^{(1)},G^{(1)}\right)=\int_{\mM^2} F^{(1)}(x)\Delta^F(x,y)G^{(1)}(y)\,d\mu_g(x)d\mu_g(y)\,.
\end{gather}
The above integral is supported on $x\gtrsim y$ and the diagonal $x=y$. In the latter case it is actually ill--defined in general, while in the former case its value is
$$
\int_{x\in J^+(y)\setminus\{y\}} F^{(1)}(x)\Delta^+(x,y)G^{(1)}(y)\,d\mu_g(x)d\mu_g(y)\,,
$$
because $\Sigma$ is a Cauchy surface and the singularities of $F^{(1)}(x)$ and $G^{(1)}(y)$ are at most space--like so that the pointwise product in the integrand is a well--defined distribution of compact support. The diagonal contribution of \eqref{eq:Blemma1} may be evaluated by extending the distribution $F^{(1)}(x)\Delta^F(x,y)G^{(1)}(y)$ to the diagonal by a direct application of the theorems for extensions of distributions. Altogether we find
\begin{gather}
\Delta^F\left(F^{(1)},G^{(1)}\right)=
\Delta^+\left(F^{(1)},G^{(1)}\right)+
R_\Sigma^1(F,G),
\end{gather}
where the last summand is supported on $\Sigma$. We may proceed similarly in the case $n>1$ and find that $R_\Sigma(F,G)\doteq\sum_n R^n_\Sigma(F,G)$ is of the stated form.
\end{proof}

We can now prove the factorisation property of the $S$-matrix.
\begin{proposition}
For any $F,G,V\in\Floc$ such that $F\gtrsim G$, \eqref{factorisation} holds.
\end{proposition}
\begin{proof}
Let be $\Sigma$ a Cauchy surface such that $\supp(F)\subset J^+(\Sigma)\setminus\Sigma$ and $\supp(G)\subset J^-(\Sigma)\setminus\Sigma$.
Decompose $V=V_++V_-$, with $V_\pm$ being such that $\supp(V_\pm)\subset J^\pm(\Sigma)$.
We now recall 
$$
S_{A+B}=S_A\star S_B\qquad\textrm{if }A\gtrsim B.
$$
Using this, we have
\begin{eqnarray*}
S_{F+V+G}&=&
S_{F+V_+}\star S_{V_-+G}+R_\Sigma\\ &=&
S_{F+V_+}\star S_{V_-}\star S_{V_-}^{-1}\star
S_{V_+}^{-1}\star S_{V_+}\star S_{V_-+G}+R_\Sigma \\&=&
S_{F+V}\star S_{V_-}^{-1}\star
S_{V_+}^{-1}\star S_{V+G}+\widetilde{R}_\Sigma\\ &=&
S_{F+V}\star S_{V}^{-1}\star S_{V+G}+\widetilde{\widetilde{R}}_\Sigma.
\end{eqnarray*}
where we used the above Lemma (\ref{lemma}) and $R_\Sigma$, $\widetilde{R}_\Sigma$ and $\widetilde{\widetilde{R}}_\Sigma$ are contributions supported on $\Sigma$ which are a priori not uniquely defined. However, since $\Sigma$ is to a large extent arbitrary, and the left hand side is independent of $\Sigma$, the only possible definition is $\widetilde{\widetilde{R}}_\Sigma=0$.
\end{proof}

\section{Convergence of the Neumann series of the classical M\o ller map on Minkowski spacetime}
\label{sec:convergence}
\begin{lemma}\label{pr:neumann series}
Suppose that $(\mM,g_1)$ is Minkowski spacetime, consider a Klein--Gordon field with action $\cS_1$ of the form \eqref{def:quadraticactions} with $A_1=0$, $j_1=0$, and $Q\in\Floc$ of the form \eqref{quadratic potential} with $G=0$, $A=0$ and $j=0$. Then 
 the Neumann series in \eqref{eq:Neumann series} converges to $\qcmoller$ in the topology of $\mE(\mM)$.
\end{lemma}
\begin{proof}
It is sufficient to prove that, for an arbitrary $\phi\in\mE(\mM)$ and denoting by $r^n$ the $n$--fold composition of $r$ \eqref{eq:Neumann series}, the series $\sum_n r^n(\phi)$ converges uniformly to $\qcmoller(\phi)$ in some generic compact set $O\subset{\mM}$ and that the same happens to its derivatives.

To this end, let us indicate by $K$ the support of $\q(\phi)(x)=M(x)\phi(x)$. We use standard Minkowski coordinates $x=(t,{\bf x})$ and restrict our attention to a compact region  $I\times N$ which contains $K$ and the compact set $O$ where we want to ensure convergence of the series $\sum_n r^n(\phi)$. In more detail, $I=[t_0,t_1]$ is an interval of time and $N$ is a compact region of a spacelike Cauchy hypersurface $\Sigma=\mathbb{R}^3$, which is chosen in such a way that $J^+(K)\cap(I\times \Sigma)$ is properly contained in $I\times N$.
For any spacelike--compact smooth function $\psi$ and its spatial Fourier transform $\widehat{\psi}(t,{\bf k})$ we introduce the norm 
\[
\|\widehat{\psi}\|_{1,\infty} =\sup_{t\in I} \|\widehat{\psi}(t,\cdot)\|_1.
\]

Let us analyse $r(\phi)$ for an arbitrary, but fixed smooth function $\phi$. First of all notice that $r(\phi)=r(\chi \phi)$ for all spacelike--compact smooth functions $\chi$ which equal $1$ on $J^+(K)$. Later on, the function $\chi$ is needed in order to ensure the existence of the spatial Fourier transform of $\chi\phi$. Observe that
\[
\widehat{r(\phi)}(t,{\bf k}) = c_0 \int_{t_0}^t ds \;\theta(t-s)\frac{\sin(\omega({\bf k})(t-s))}{\omega({\bf k})} \widehat{M \chi\phi}(s,{\bf k}) 
\]
where $c_0$ is a suitable power of $2\pi$, $\theta$ is the Heaviside step function and $\omega({\bf k})={\bf k}^2+M_1$ is the spectral function associated with the background equation of motion. Hence
\[
\|\widehat{r(\phi)}(t,\cdot)\|_{1} \leq c_0\int d^3{\bf k} \int_I ds \;\theta(t-s)\left|\frac{\sin(\omega({\bf k})(t-s))}{\omega({\bf k})} \right|   \left|\widehat{ M \chi\phi}(s,{\bf k})\right|.
\]
Since the function $\omega({\bf k})$ is real and positive
\[
\|\widehat{r(\phi)}(t,\cdot)\|_{1}
\leq  c_1
(t_1-t_0) \int_{t_0}^t ds \|\widehat{ M} \|_{1,\infty}    \|  \widehat{\chi\phi}(s,\cdot)\|_{1}\,,
\]
where $c_1$ is a suitable power of $2\pi$.
In a similar way, one can show that for every $n>0$
\[
\|\widehat{r^n (\phi)}\|_{1,\infty}
\leq 
c_n\frac{(t_1-t_0)^{2n}  \|\widehat{ M} \|_{1,\infty}^n}{n!}  \|  \widehat{\chi\phi}\|_{1,\infty},
\]
where the factor $n!$ arises because of an integration over an $n-$dimensional simplex and $c_n$ is a suitable power of $2\pi$.
Since for every spacelike compact smooth function $\|\psi\|_\infty$ is controlled by the norm $\|\hat\psi\|_{1,\infty}$, which in turn is controlled by the seminorms defining the topology of $\mE(\mM)$, we find that 
 the series $\sum_{n\geq 1} r^n(\phi)$ converges uniformly. Following a similar path, we can prove that also the derivatives of 
 $\sum_{n\geq 1} r^n(\phi)$ converge uniformly, concluding the proof.
\end{proof}
\section{Clustering properties of free massive KMS states in Minkowski spacetime}
\label{sec:clustering}

\begin{lemma}\label{pr:clustering} We denote by $[A]_g^{(m)}$ the $m$--th perturbation order of $\mR^\hbar_{V(g)}(A)$ and use the notation of Section \ref{sec:KMS}. The connected correlation functions $\omega_c^\beta$ of the $\beta$--KMS state of the Klein--Gordon field on Minkowski spacetime with mass $m>0$ satisfies the following property. For all $A_i \in \mA(\Sigma_{2\epsilon})$ and all $\bm\in\bN^k$, the function
$$
F^{\bm}_{k}(u_1,\bx_1,\dots, u_k,\bx_k)(g_0,\dots, g_k) \doteq 
\omega_c^\beta\left([A_0]_{g_0\chi}^{(m_0)}\star \alpha_{u_1,\bx_1}\left([A_1]_{g_1\chi}^{(m_1)}\right) \star \ldots \star \alpha_{u_n,\bx_k}\left([A_k]^{(m_k)}_{g_k\chi}\right) \right)
$$
is uniformly bounded by an absolutely integrable function for all $g_i$ in some bounded set of $\mD(\mM)$.
\end{lemma}
\begin{proof}
The authors of \cite{FredenhagenLindner} have already shown, that, for fixed $g_i$, $F^{\bm}_{k}$ satisfies the strong clustering property, namely that $F^{\bm}_{k}$ decays exponentially for large $\bx_i$
\[
|F^{\bm}_{k}(u_1,\bx_1,\dots, u_k,\bx_k)| \leq c e^{-\frac{m}{\sqrt{k}}r},
\qquad
r = \sqrt{\sum_{i=1}^k\left( u_i^2+|\bx_i|^2 \right)}.
\]
The point which is left open is to control how this estimate (and in particular the constant $c$), depends on the functions $g_i$.

With this in mind we follow the proof of \cite[Theorem 3]{FredenhagenLindner} and track the appearance of the $g_i$. To this avail, we consider an arbitrary oriented and connected graph $G$ joining $n+1$ vertices. We denote by $E(G)$ the edges of this graph, and for any $l\in E(G)$ we denote by $s(l)$ and $r(l)$ the source respectively range of this edge. Denoting by $X\doteq(x_0,\ldots,x_n)$ the vertices of $G$, we set 
\begin{equation}\label{eq:def-psi}
\Psi(X,Y) = \left. \prod_{l\in E(G)}\frac{\delta^2}{\delta\phi_{s(l)}(x_l)   \delta \phi_{r(l)}(y_l) } (B_0\otimes \dots\otimes B_n) \right|_{(\phi_0,\dots,\phi_n)=0},
\end{equation}
and we have to analyse how this quantity depends on $g_i$, where $B_i\doteq[A_i]^{(m_i)}_{g_i}$ and $A_i$ are supported in a neighbourhood of the origin. We notice that $A_i$ is of compact support and $\chi$ is also of compact support in time, hence, the support of 
$B_i$ is contained in $(\supp \chi \cap J^-(\supp A_i))$ and thus compact and contained in a ball of sufficiently large radius $R$. As a consequence of this fact, $\Psi(X,Y)$ is a distribution of compact support. This quantity is of interest because $F^{\bm}_{k}(u_1,\bx_1,\dots, u_k,\bx_k)(g_0,\dots, g_k)$ is a sum of terms of the form 
$$\int dX\,dY\,\prod_{l\in E(G)}\Delta^{+,\beta}(\overline{x_l}-\overline{y_l})\Psi(X,Y)$$
with $\overline{x}_l \doteq (x^l_0 + iu_{s(l)}, \bx_l + \bz_{s(l)})$, $\overline{y}_l \doteq (y^l_0 + iu_{r(l)}, \by_l + \bz_{r(l)})$ and $\Delta^{+,\beta}$ denoting the two--point function of $\omega^\beta$.

In the proof of \cite[Theorem 3]{FredenhagenLindner}, the need of estimating the Fourier transform
$\widehat\Psi(-P,P)$ for $P=((p^{0}_0,\bp^0),\dots, (p^{n}_0,\bp^n))$ with $p_0^k = i \sqrt{{\bp^k}^2+m^2}$ arises. In particular, since the original $\Psi$ is a distribution of compact support, the Paley--Wiener--Schwartz theorem permits to control the exponential growth of $\widehat\Psi(-P,P)$ for $P$ of the above--mentioned form and large values of $\bp^l$ in terms of $R$. We shall now analyse how these estimates depend on the various  $g_i$. 

We start by observing that the perturbation potential $V(g_i\chi)$ depends linearly on $g_i\chi$. We can thus write
\[
B_i=[A_i]^{(m_i)}_{g_i\chi} =  \int_{\mM^{m_i}} dz_1 \dots dz_{m_i}\,  g_i (z_1)\dots  g_i(z_{m_0}) G(z_1,\dots z_{m_i})
\]
where $G$ is a distribution of compact support with values in $\mA(\mO)$, $\mO$ being a neighbourhood of the origin. This implies that the coupling functions $g_i$ appear in $\Psi$ as
\[
\Psi(X,Y)  =  \left\langle t (X,Y,Z),  \underbrace{g_1(z_{1,1})\otimes \ldots \otimes g_1(z_{1,{m_1}})}_{m_1 \text{ times}} \otimes    \ldots \otimes \underbrace{g_n(z_{n,1})\otimes \ldots \otimes g_n(z_{n,{m_k}})}_{m_n \text{ times}}     \right\rangle 
\]
where we used the compact notation $Z\doteq(z_{1,1},\ldots,z_{n,m_n})$. Here, $t$ is suitable distribution of compact support in $X$, $Y$ and $Z$ with values in $\mA(\mO)$ which does not depend on the $g_i$. We may thus readily apply the Paley--Wiener--Schwartz theorem to obtain that, the Fourier--Laplace transform $\widehat{t}$ of $t$ is an entire analytic function which satisfies
\begin{equation}\label{eq:pw}
|\widehat{t} (P_1,P_2,K)| \leq C(1+\|P_1\|_1+\|P_2\|_1+\|K\|_1)^N e^{R \|\text{Im}(P_1)\|_1}e^{R \|\text{Im}(P_2)\|_1}e^{R \|\text{Im}(K)\|_1}
\end{equation}
where $R$ is the radius of the ball $B_R$ centered in the origin which contains the support of $t$ and $C$ and $N$ are suitable constants. 

From this estimate we can now obtain a similar estimate for $\Psi$, namely,
\[
|\widehat{\Psi} (-P,P)| =     \int_{\bR^{4w}}  dK  \left|\widehat{t} (-P,P,K)\right|    \widehat{g}_i(k_1) \dots \widehat{g}_n(k_{w}) 
\]
where $w=\sum_i m_i$ is the dimension of $Z$ and hence of $K$.
Since $N$ in \eqref{eq:pw} is bounded, since the integral is performed over the real numbers and since $g_i$ are smooth compactly supported functions, the integral can be easily performed to get 
\[
|\widehat{\Psi} (-P,P)| \leq 
C(1+\|P\|_1)^N e^{2R \|\text{Im}(P)\|_1}  \prod_i \sum_{|\alpha_i| \leq M} \|  D^{\alpha_i} g_i\|_\infty^{m_i}.
\]
Where now the constant $C$ does not depend on the $g_i$ anymore and $M$ is further suitable finite constant.
Using this more refined estimate in the proof of \cite[Theorem 3]{FredenhagenLindner}, the desired uniform bound can be found. 
\end{proof}

\section{Convergence of the state induced by the classical M\o ller map in the adiabatic limit on Minkowski spacetime}

\begin{lemma}\label{pr:adiabaticvacuum}
Let $\Delta^+_1(x,y)$ be the two--point function of the vacuum state with mass $m_1\ge 0$ and let $Q(f)\doteq \frac12(m_2^2-m_1^2)\int_\mM f \phi^2 d\mu_g$ with $m_2\ge 0$ (but not necessarily $m_2>m_1$). Then the integral kernel of $\Delta^+_{1+Q(f)} \doteq  R_{1,Q(f)}\circ \Delta^+_1 \circ R^\dagger_{1,Q(f)}$ converges to the integral kernel of the two--point function of the vacuum state with mass $m_2$ in the adiabatic limit $f\to 1$.
\end{lemma}

\begin{proof}We prove this statement in the spirit of QFT on cosmological spacetimes. To this avail we note that $\Delta^+_{1+Q(f)}$ is well--defined for all $f$ which have past--compact support and thus in particular for $f$ which are constant in space and have a time dependence of the form $f(t)=\int^t_{-\infty} \chi(\tau)d\tau$ for a $\chi\in\mD(\bR)$ which is positive and has unit integral. In this case, we may decompose the integral kernel of $\Delta^+_{1+Q(f)}$ as

$$\Delta^+_{1+Q(f)}(t_1,\bx_1, t_2, \bx_2)=\lim_{\epsilon\downarrow 0}\frac{1}{8\pi^2}\int \overline{T_k(t_1)} T_k(t_2) e^{i\bk\left(\bx_1-\bx_2\right)} e^{-\epsilon k}d\bk \,,$$
where $k=|\bk|$ and the temporal modes $T_k(t)$ satisfy the differential equation 
\beq\label{eq:adiabaticvacuumproof1}\left(\partial^2_t + \omega^2(t) \right)T_k(t)=0\,,\qquad \omega^2(t)\doteq k^2 +m^2_1+(m^2_2-m^2_1)f(t)\eeq
and the normalisation condition
$$\overline{\dot{T}_k(t)}T_k(t)-\overline{T_k(t)}\dot{T}_k(t)=i\,.$$
Here $\dot{\,}$ denotes a derivative w.r.t. $t$. Moreover, by construction the modes $T_k(t)$ equal the $m_1$--vacuum modes 
$$T_{1,k}\doteq  \frac{1}{\sqrt{2 \omega_1}} e^{-i\omega_1 t}\,,\qquad \omega^2_1 \doteq k^2 + m^2_1$$
in the past of $\supp f$ (up to an irrelevant constant phase).

We would like to show that in the adiabatic limit $f\to 1$ the modes $T_k$ converge to the $m_2$--vacuum modes $T_{2,k}$ defined in analogy to $T_{1,k}$. We do this by comparing both $T_k$ and $T_{2,k}$ to the adiabatic modes 
$$T_{a,k}(t) \doteq  \frac{1}{\sqrt{2 \omega(t)}} e^{-i\int^t_{t_0}\omega(\tau) d\tau}$$
in the adiabatic limit, cf. \eqref{eq:adiabaticvacuumproof1} for the definition of $\omega(t)$. Obviously we have $\lim_{f\to 1}T_{a,k} = T_{2,k}$. In order to check that $\lim_{f\to 1}(T_{a,k}-T_k) = 0$ as well, we observe that the adiabatic modes $T_{a,k}$ satisfy the differential equation 
\beq\label{eq:adiabaticvacuumproof2}\left(\partial^2_t + \omega^2(t) + \lambda(t) \right)T_{a,k}(t)=0\,,\qquad 
\lambda = \frac{1}{2} \frac{\ddot{\omega}}{\omega} - \frac{3}{4} \left(\frac{\dot\omega}{\omega}\right)^2  =
\frac{1}{4}\frac{(\omega^2)\ddot{\;}}{\omega^2} - \frac{5}{16} \left(\frac{{(\omega^2)}\dot{\;}}{\omega^2}\right)^2\,.
\eeq
and also equal the $m_1$--vacuum modes $T_{1,k}$ in the past of $\supp f$. We may thus construct the modes $T_k$ by means of a perturbation series over $T_{a,k}$ by treating $-\lambda$ as a perturbation potential. This gives 
\beq\label{eq:adiabaticvacuumproof3}
T_k = \sum_{n=0}^\infty   R_\lambda^n (T_{a,k})\,,\qquad R_\lambda(h) \doteq \int_{-\infty}^t \frac{\sin (\int_\tau^{t} \omega(\tau_1) d\tau_1 )}{\sqrt{\omega(t)\omega(\tau)}} \lambda(\tau)  h(\tau) d\tau\,, 
\eeq
where $R_\lambda|_{\lambda=1}$ is the retarded propagator for the differential equation \eqref{eq:adiabaticvacuumproof2}. Using \eqref{eq:adiabaticvacuumproof3} we may estimate the difference $T_{a,k}-T_k$ as
$$
|T_{a,k}-T_k| \leq \frac{1}{\sqrt{2\omega}}\left| -1  + \exp{ \int_{-\infty}^{\infty} \frac{|\lambda|}{\omega}   d\tau } \right|\,,
$$
Using the precise form of $\lambda$, we may further estimate the argument of the exponential as
\beq
\int_{-\infty}^{\infty} \frac{|\lambda|}{\omega}   d\tau  \leq C \left(\int_{-\infty}^{\infty}  |m_2^2-m_1^2||{(\omega^2)\ddot{\;}}| d\tau + \int _{-\infty}^{\infty}  (({\omega^2})\dot{\;})^2 d\tau    \right)
\eeq
for a suitable (dimensionful) constant $C$. 

Let us now consider $\chi_\mu(t) \doteq \frac{1}{\mu}\chi\left(\frac{t}{\mu}\right)$ in place of $\chi$ in the definition of the cutoff function $f$. For this special class of cutoffs, the adiabatic limit corresponds to $\lim_{\mu\to \infty}$. Indeed, with this parametrisation we have
$$
\int_{-\infty}^{\infty} \frac{|\lambda|}{\omega}   d\tau  
\leq C (m_2^2-m_1^2)^2  \left(\int_{-\infty}^\infty   |\dot{\chi}| d\tau + \int_{-\infty}^\infty  \chi^2 d\tau    \right) \frac{1}{\mu}\,.
$$
This proves $\lim_{f\to 1}(T_{a,k}-T_k) = 0$ and thus $\lim_{f\to 1}(T_{2,k}-T_k) = 0$.

The argument above fails in the case $m_1=0$ or $m_2=0$ on account of singularities at $k=0$. In order to see that these singularities are irrelevant, we assume without loss of generality $m_1=0$ and $m_2>0$ and recall that we have taken $f(t)$ to be monotonically increasing. Consequently, $\omega^2(t)=k+m^2_1+(m_2^2-m_1^2)f(t)$ is monotonically increasing as well and introducing the energy per ``retarded mode''
$$E_k\doteq |\dot{T}_k|^2 + \omega^2 |T_k|^2$$
we find
$$\frac{d}{dt}\left(\frac{E_k}{\omega^2}\right)= - |T_k|^2 \frac{(\omega^2)\dot{\;}}{\omega^4}\le 0\,.$$
For an arbitrary $t=t_0$ such that $f(t_0)=0$ we have $E_k(t_0)=\omega(t_0)=k$. From the temporal behaviour of $E_k/\omega^2$ we may thus infer
$$
|T_k(t)|^2 \le \frac{E_k(t)}{\omega^2(t)} \le \frac{E_k(t_0)}{\omega^2(t_0)}  = \frac{1}{k}\,. 
$$
Note that this bound is uniform in $f$ and thus holds in the adiabatic limit. Moreover the same bound holds also for the (constant) energy per mode of the $m_2$--vacuum modes $T_{2,k}$. Using these bounds, we may investigate the adiabatic limit as follows. We have (omitting the $\epsilon$--prescription)
\begin{align*}\Delta^+_{2}(x_1, x_2)-\Delta^+_{1+Q(f)}(x_1, x_2)=&\frac{1}{8\pi^2}\int_{\bR^3} \left(\overline{T_{2,k}(t_1)} T_{2,k}(t_2)-\overline{T_k(t_1)} T_k(t_2)\right) e^{i\bk\left(\bx_1-\bx_2\right)} d\bk \\
\doteq &\int_{\bR^3} I(\bk,x_1,x_2) d\bk = \int_{B_\epsilon} I(\bk,x_1,x_2) d\bk + \int_{\bR^3\setminus B_\epsilon} I(\bk,x_1,x_2) d\bk
\,,
\end{align*}
where $B_\epsilon\subset \bR^3$ is the solid sphere of radius $\epsilon$. The $k>\epsilon$ contribution to the integral converges in the adiabatic limit by the arguments used above.  The $k<\epsilon$ contribution can be estimated using the bounds previously discussed. Indeed $|I(\bk,x_1,x_2)|\le (4\pi^2 k)^{-1}$ and thus the integral over $B_\epsilon$ is $\cO(\epsilon^2)$ and negligible in the limit $\epsilon\to 0$ uniformly in $f$.

 Note that the proof fails in the tachyonic case $m^2_1<0$ or $m^2_2<0$ because the infrared singularity in this case is more severe than in the one occurring in the massless case, in particular the singularity in $\omega^{-1}$ is not integrable anymore.

\end{proof}

\end{document}